\documentclass[11pt]{article}
\usepackage[left=1in,top=1in,right=1in,bottom=1in,head=.1in,nofoot]{geometry}

\usepackage{setspace,url,bm,amsmath} %

\usepackage{breakurl}
\usepackage[breaklinks]{hyperref}

\usepackage{titlesec} %
\titlelabel{\thetitle.\quad} %

\titleformat*{\section}{\bf\Large\center}

\usepackage{titletoc}

\usepackage{graphicx} %
\usepackage{bbm}
\usepackage{latexsym}
\usepackage{caption}
\usepackage[margin=20pt]{subcaption}
\usepackage{hyperref}
\usepackage{booktabs}
\usepackage{enumitem}

\usepackage[table]{xcolor}

\newcommand{\GG}[1]{}

\usepackage{amsthm}
\usepackage{amssymb}
\usepackage{amsmath}
\usepackage{color}

\usepackage{comment}
\theoremstyle{definition}

\newtheorem*{theorem*}{Theorem}
\newtheorem{theorem}{Theorem}
\newtheorem*{rmk*}{Remark}
\newtheorem{rmk}{Remark}
\newtheorem{proposition}{Proposition}
\newtheorem{lemma}{Lemma}

\newtheorem{definition}{Definition}
\newtheorem{corollary}{Corollary}
\newtheorem*{corollary*}{Corollary}

\def\rank{\text{r}}

\def\rev{\color{black}}

\def\c{\complement}

\usepackage{natbib} %
\bibpunct{(}{)}{;}{a}{}{,} %

\usepackage{etoolbox} %
\apptocmd{\sloppy}{\hbadness 10000\relax}{}{} %

\usepackage{color}
\usepackage{listings}

\DeclareMathOperator*{\argmin}{arg\,min}

\newcommand{\lxr}{\color{black}}

\def\I{\mathbbm{1}}

\usepackage{color}

\allowdisplaybreaks

\def\vecge{\succcurlyeq}
\def\vecle{\preccurlyeq}
\usepackage{textcomp}

\usepackage{soul}

\def\Unif{\text{Unif}}

\usepackage[textsize=scriptsize, textwidth = 2cm, shadow]{todonotes}

\usepackage[sc]{mathpazo}

\usepackage{float}

\RequirePackage[normalem]{ulem}

\begin{document}

\onehalfspacing
\title{\bf 
Randomization Inference beyond the Sharp Null:\linebreak
Bounded Null Hypotheses and Quantiles of Individual Treatment Effects
}
\author{
	Devin Caughey, Allan Dafoe, Xinran Li and Luke Miratrix
\footnote{
Devin Caughey, 
Department of Political Science, 
MIT, Cambridge, MA, USA (e-mail: \href{mailto:caughey@mit.edu}{caughey@mit.edu}).
Allan Dafoe,  
Centre for the Governance of AI, Oxford, UK
(e-mail: \href{mailto:allandafoe@gmail.com}{allandafoe@gmail.com}).
Xinran Li, 
Department of Statistics, University of Chicago, Chicago, IL, USA (e-mail: \href{mailto:xinranli@uchicago.edu}{xinranli@uchicago.edu}).
Luke Miratrix, 
Graduate School of Education, Harvard University, Cambridge, MA, USA
(e-mail: \href{mailto:lmiratrix@g.harvard.edu}{lmiratrix@g.harvard.edu}).
}
}
\date{}
\maketitle

\begin{abstract}
	\noindent Randomization inference (RI) is typically interpreted as testing Fisher's ``sharp'' null hypothesis that all unit-level effects are exactly zero. This hypothesis is often criticized as restrictive and implausible, making its rejection scientifically uninteresting. We show, however, that many randomization tests are also valid for a ``bounded'' null hypothesis under which the unit-level effects are all {\rev non-positive} (or all {\rev non-negative}) but are otherwise heterogeneous. In addition to being more plausible a priori, bounded nulls are closely related to substantively important concepts such as monotonicity and Pareto efficiency. Reinterpreting RI in this way also dramatically expands the range of inferences possible in this framework. We show that exact confidence intervals for the maximum (or minimum) unit-level effect can be obtained by inverting tests for a sequence of bounded nulls. We also generalize RI to cover inference for quantiles of the individual effect distribution as well as for the proportion of individual effects larger (or smaller) than a given threshold. The proposed confidence intervals for all effect quantiles are simultaneously valid, in the sense that no correction for multiple analyses is required, and are thus a ``free lunch'' added to conventional RI. In sum, our reinterpretation and generalization provide a broader justification for randomization tests and a basis for exact nonparametric inference for effect quantiles. We illustrate our methods with simulations and applications, finding that Stephenson rank statistics can provide more informative results than the more common Wilcoxon rank or difference-in-means statistics. 
    We also provide an R package {\rev \textsf{RIQITE}} implementing the proposed approach.
\end{abstract}
{\bf Keywords}: 
causal inference, potential outcome, quantiles of individual treatment effects, randomization test, treatment
effect heterogeneity

\section{Introduction}

\subsection{Randomization inference and sharp null hypotheses}\label{sec:ri_sharpnull}

Randomization inference (RI), also known as permutation inference, is a general statistical framework for making inferences about treatment effects. RI originated with \cite{Fisher35a}, who showed that if the treatment is randomly assigned to units, the hypothesis that no unit was affected by treatment---the ``sharp null of no effects''---can be tested exactly,  with no further assumptions, by comparing an observed test statistic with its distribution across alternative realizations of treatment assignment. More generally, the Fisher randomization test can be applied to any null hypotheses that is sharp in the sense that under it all potential outcomes are known from the observed data. Furthermore, by testing a sequence of hypotheses, randomization tests can also be used to create exact nonparametric confidence intervals (CIs) for treatment effects \citep{Lehmann63a}. RI thus provides a unified framework of statistical inference that requires neither parametric assumptions about the data-generating distribution nor asymptotic approximations that can be unreliable in small samples \citep{Rosenbaum02a}.

Nevertheless, RI has been criticized from various angles. One line of criticism focuses on the sharp null hypothesis, which has long been dismissed as ``uninteresting and academic'' \citep[173]{Neyman35a}. \citet{Gelman11a}, for example, argues that ``the so-called Fisher exact test almost never makes sense, as it's a test of an uninteresting hypothesis of exactly zero effects (or, worse, effects that are nonzero but are identical across all units).'' The crux of this critique is that the sharp null ``does not accommodate heterogeneous responses to treatment,'' making it ``very restrictive''  \citep[330]{Keele15a}.
A related objection is that using RI for interval estimation requires assumptions that are arguably as strong as those of its parametric and large-sample competitors. In particular, deriving interpretable CIs for treatment effects typically requires the assumption that effects are constant across units or vary according to some known model \citep[e.g.,][]{Bowers13a}.

Defenders of RI have responded in various ways.
Some argue that RI is nevertheless useful for assessing whether treatment had any effect at all, as a preliminary step to determine whether further analysis is warranted \citep[e.g.,][]{Imbens15a}. An alternative proposal, advanced by
\citet{Chung13a}, is to employ “studentized” test statistics that render permutation tests asymptotically valid under a weak null hypothesis
\citep[e.g.,][]{Ding:2017aa, Wu:2018aa, Fogarty:2019aa, cohen2020gaussian}.\footnote{The term \textit{weak} is typically used to refer to the null hypothesis of no average effect, in contradistinction to the stronger hypothesis of no effect whatsoever \citep[e.g.,][A-32]{FreedmanPisaniPurves97a}. We use \textit{weak} more generally, to refer to any hypothesis that stipulates the value of some function of the unit-level treatment effects (e.g., their average or a given quantile) but otherwise allows for arbitrary effect heterogeneity. A weak null is a ``composite'' hypothesis in the sense that it encompasses multiple configurations of potential outcomes rather than a single one as a sharp null does.} 
Some scholars defend the constant-effects assumption more forthrightly, regarding it as a convenient approximation that is preferable to the shortcomings of parametric methods, such as their sensitivity to assumptions about tail behavior \citep{rosenbaum2010design} or inability to account for complex treatment assignments \citep{HoImai06a}. 

Though reasonable, all these defenses presume that randomization tests are exactly valid only as tests of a sharp null hypothesis.
We offer a more fundamental defense. As we show, many randomization tests of the sharp null are exactly valid under a corresponding hypothesis under which unit-level effects are bounded but otherwise hetereogenous. This result in turn provides the basis for CIs for the maximum (or minimum) individual effect and, by extension, for any quantile of the distribution of unit-level effects. These results substantially expand the applicability of RI and permit assessment of the substantive magnitude of treatment effects as well as their statistical significance.

\subsection{A motivating example}\label{sec:moti_example}

To motivate our approach, consider \citeauthor{HSMHSD10}'s (\citeyear{HSMHSD10}) experimental study of the effectiveness of professional development for elementary teachers. The study compared 164 teachers assigned to take a professional development course with 69 control subjects  (for full details, see Section \ref{sec:profession}). The average gain score, based on tests on content knowledge before and after the courses, was much higher among the teachers who took the courses (see Figure~\ref{fig:conf_profession}(a).
A Stephenson rank-sum randomization test (see Section \ref{sec:stephenson}) yields a $p$-value near 0, and Lehmann-style test inversion yields a $90\%$ CI of $[16.7, \infty)$ for a constant treatment effect.\footnote{The $p$-value is approximated by Monte Carlo with $10^6$ simulated assignments. Consequently, the standard error of this Monte Carlo approximation is at most $5\times 10^{-4}$.}
These results strongly suggest that the professional development improves teacher's content knowledge, but their precise interpretation is less clear.
The $p$-value indicates decisive rejection of the sharp null of no effects, but is rejecting this hypothesis really informative?
And if a constant-effects assumption is not plausible, how should we interpret the CI?

Our paper sheds light on both of these questions.
First, because the Stephenson 
rank sum
test is also valid under the null that all unit-level effects {\rev are} bounded above at 0, the $p$-value reported above justifies rejection of the hypothesis that no teacher's content knowledge increased as a result of the course.
Second, without a constant-effects assumption, the CI reported above can be interpreted as a confidence statement about the \textit{maximum} effect---specifically, that the hypothesis that no effect was larger than 
16.6
can be rejected at a significance level of 0.1.

{\rev
Moreover, we can generalize these results to obtain simultaneously valid CIs for all quantiles of the effect distribution. 
These inferences are visualized in Figure \ref{fig:conf_profession}(b), which reports all the
simultaneous $90\%$ one-sided confidence intervals  with finite lower bounds.
Specifically, 
the horizontal lines in Figure \ref{fig:conf_profession}(b) represent one-sided confidence intervals for the largest $117$ individual effects. 
The 146th and 165th largest individual effects, for example, have $90\%$ CIs of $(0, \infty)$ and $[6.660, \infty)$, respectively.

In addition, the lower confidence limit for the number of units with effects greater than $c$ 
is equivalent to the number of quantiles of individual effects 
whose confidence intervals do not cover $c$. 
We can thus read off these confidence intervals using Figure \ref{fig:conf_profession}(b). Consider the dashed vertical line of $c=0$ as an example. 
Because the line of $c=0$ only intersects the confidence intervals for $k \le 145$, we know at least $233 - 145 = 88$ units, or $88/233 = 37.8\%$, have positive effects.
By the same logic, a $90\%$ confidence interval for the number of units with effects larger than $6$ is $[69, 233]$. That is, $69/233 = 29.6\%$ had an effect of at least 6.

}

\begin{figure}[htbp]
	\centering
	\begin{subfigure}{.5\textwidth}
		\centering
		\includegraphics[width=1\linewidth]{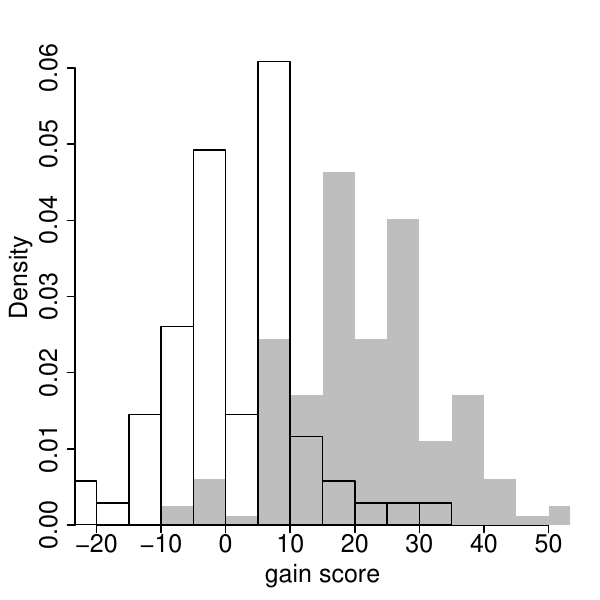}
		\caption{}
	\end{subfigure}%
	\begin{subfigure}{.5\textwidth}
		\centering
		\includegraphics[width=1\linewidth]{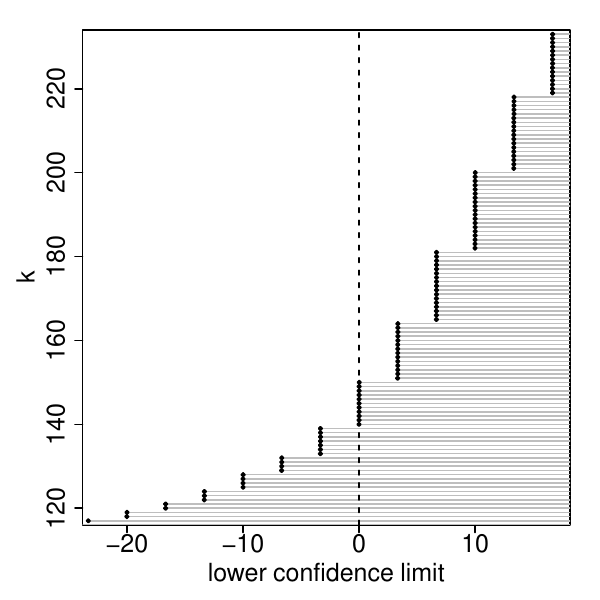}
		\caption{}
	\end{subfigure}%
	\caption{
	Histograms of observed gain scores and 
	$90\%$ confidence intervals for quantiles of individual effects in the study of the effectiveness of professional development for elementary teachers.
    (a) shows the histograms of the observed gain scores in treatment (grey) and control (white) groups, respectively. 
    (b) shows the $90\%$ simultaneous lower confidence limits for all quantiles of individual effects. Details of their generation are in Section~\ref{sec:applications}. 
     The uninformative lower confidence limits of $-\infty$ for individual effects at lower ranks are omitted.
	}\label{fig:conf_profession}
\end{figure}

\subsection{Our contribution}\label{sec:contribution}

The inferences illustrated in Section~\ref{sec:moti_example} are grounded in a novel set of theoretical results. In particular, in this paper we prove the following:
\begin{enumerate}[label=(\roman*), topsep=1ex,itemsep=-0.3ex,partopsep=1ex,parsep=1ex]
    \item For any randomization test that employs a test statistic with one of several properties, one-sided rejection of the sharp null hypothesis that the treatment $\tau_i$ equals some constant $\delta_i$ for each unit $i$ also implies rejection of any null hypothesis under which $\tau_i \leq \delta_i$ for all $i$.
    \item Test statistics with the requisite properties include the difference-in-means, the Wilcoxon rank sum, and many other commonly used statistics.
    \item For tests in this class, inverting a sequence of tests provides confidence intervals for the maximum (or minimum) individual effect.
    \item For rank-based members of this class, 
    {\rev when treatment assignments are exchangeable across all units,}  
    this procedure can be extended to yield simultaneously valid confidence intervals for all quantiles of the treatment-effect distribution (and analogously for the proportion of units with effects larger/smaller than a given threshold).
    \item Confidence intervals for the range of unit-level effects can be obtained by combining two one-sided randomization tests, providing an exact test of effect heterogeneity.
    \end{enumerate}
These results have important implications for statistical practice.

First, ``bounded" hypotheses are often of substantive interest in themselves. Unlike sharp hypotheses, which are concentrated at a particular point in the parameter space, bounded hypothesis cover a range of treatment effects. Consequently, their prior plausibility is greater and their rejection is thus more informative. In addition, there are a number of theoretically or methodologically important special cases of bounded hypotheses. In economics, for example, a change is considered a ``Pareto improvement" if it makes at least one person better off while hurting no one \citep[34]{mishan82introd_polit_econom}. Consequently, the claim that an intervention was Pareto improving can be assessed by testing the bounded null hypothesis that all effects were greater than or equal to zero. Similarly, instrumental-variable estimation of causal effects is typically conducted under a monotonicity assumption that the instrument has non-negative or non-positive effects on the treatment \citep{AngristImbensRubin96a}. 

Second, these results provide a basis for inferences regarding the distribution of treatment effects across units. Specifically, they permit interval estimation of treatment effect quantiles. Existing nonparametric methods for estimating causal effects focus overwhelmingly on average effects of various kinds. Although averages are often the best summary statistic, quantiles characterize the effects of treatment more completely and robustly, especially when effects are highly heterogeneous.
Although there are existing RI methods for "quantile treatment effects" (the treated--control difference in \textit{outcome} quantiles; e.g., \citealt{CattaneoEtAl15a}), to our knowledge ours is the first aimed at quantiles of the distribution of \textit{effects}\footnote{\rev For the distinction between \textit{differences in outcome distributions} (e.g., average or quantile treatment effects) and \textit{distributions of outcome differences} (e.g., the distribution of individual effects), see \citet[][157--8]{manski2009identification}.} 
Furthermore, two one-sided CIs can be combined to yield tests and CIs for the range of treatment effects.

Significantly, the above advances are a "free lunch" in the sense that they require no additional assumptions. Users of RI can continue to use the same procedures while interpreting them in richer ways, or they can use our open-source \textsf{R} package \textsf{RIQITE} (\url{https://github.com/li-xinran/RIQITE}) to supplement their analyses with quantile CIs and other extensions.\footnote{{\rev The package's Github page includes installation instructions, detailed explanations of the main functions in R documentation, and a simple illustrating example. In addition, the supplementary materials for this paper contain a replication of the real data analyses in the paper using \textsf{RIQITE}.}}

The remainder of this paper is organized as follows. Section \ref{sec:framework} formally introduces our framework. Section \ref{sec:statistics} discusses the properties of test statistics to which our results apply and provides examples. Section \ref{sec:broader} demonstrates the validity of randomization tests under bounded null hypotheses and explains how this justifies CIs for the maximum effect.
Section \ref{sec:inf_quant_all} generalizes these results to CIs for effect quantiles. 
Section \ref{sec:twosided} shows how CIs for the range of effects can be derived from two one-sided randomization tests.
Section \ref{sec:simu} conducts simulation studies for the performance of these procedures under various conditions.
Section \ref{sec:applications} applies the methods to an empirical evaluation of an experimental professional development program (another application, testing the monotonicity assumption of an instrumental variable, is in Appendix 7.2 of the supplementary materials).
The paper concludes with a discussion of the broader implications of our results.
All proofs of theorems and associated results are relegated to the supplementary materials for conciseness.

\section{Framework, notation, and the Fisher randomization test}
\label{sec:framework}

\subsection{Potential outcomes, treatment effects and treatment assignment}\label{sec:potential_outcome}

We consider a randomized experiment on $n$ units with two treatment arms. Using the potential outcome framework \citep{Neyman23a, Rubin:1974wx}, we use $Y_i(1)$ and $Y_i(0)$ to denote the potential outcomes under treatment and control, respectively, for units $i=1, \ldots, n$. We use $\bm{Y}(1) = (Y_1(1), Y_2(1), \ldots, Y_n(1))^\top$ and $\bm{Y}(0) = (Y_1(0), Y_2(0), \ldots, Y_n(0))^\top$ to denote the treatment and control potential outcome vectors for all units. 
Let $\tau_i = Y_i(1) - Y_i(0)$ be the individual treatment effect for unit $i$, 
and 
$\bm{\tau} = (\tau_1, \tau_2, \ldots, \tau_n)^\top$ be the  vector of individual treatment effects.\footnote{
As a side note, we focus on general outcomes with the treatment effect in the difference scale, which may not be most appropriate in some applications; see, e.g., \citet{Edwards1963} and \citet{geng2008}.
}
Let $Z_i$ be the treatment assignment for unit $i$, where $Z_i$ equals 1 if the unit receives the active treatment and zero otherwise, 
and $\bm{Z} = (Z_1, Z_2, \ldots, Z_n)^\top$ be the treatment assignment vector for all units. 
For each unit $i$, 
the observed outcome is one of its two potential outcomes, depending on the treatment assignment $Z_i$. 
Specifically, 
$Y_i \equiv Y_i(Z_i) = Z_i Y_i(1) + (1-Z_i) Y_i(0)$. 
Let $\bm{Y} = (Y_1, Y_2, \ldots, Y_n)^\top$ be the observed outcome vector for all units. 
For descriptive convenience, for any treatment assignment vector $\bm{z} \in \{0,1\}^n$, 
 define $\bm{Y}(\bm{z}) = \bm{z} \circ \bm{Y}(1) + (\bm{1} - \bm{z}) \circ \bm{Y}(0)$ to denote the corresponding observed outcome vector, where $\circ$ stands for the element-wise multiplication. 
We can then write the observed outcome $\bm{Y}$ as $\bm{Y} = \bm{Y}(\bm{Z}) = \bm{Z} \circ \bm{Y}(1) + (1-\bm{Z}) \circ \bm{Y}(0)$.

In this paper, we conduct randomization inference where all potential outcomes are viewed as fixed constants, 
and use randomization of the treatment assignments as the ``reasoned basis'' for inference \citep{Fisher35a}. 
This is equivalent to conducting inference conditional on all the potential outcomes, and %
thus requires no model or distributional assumptions on the potential outcomes. 
The distribution of the treatment assignment $\bm{Z}$, also called the treatment assignment mechanism, is what governs statistical inference. 
We use $\mathcal{Z} \subset \{0,1\}^n$ to denote the set of all possible treatment assignments for the $n$ units, 
and characterize the treatment assignment mechanism %
by the probability mass function $\Pr(\bm{Z}=\bm{z})$ for all $\bm{z} \in \mathcal{Z}$. 
One class of treatment assignment mechanism that will receive special attention is the \textit{exchangeable} treatment assignment, formally defined as follows. 

\begin{definition}\label{def:exchange}
    A treatment assignment mechanism is exchangeable if $\bm{Z}$ is invariant under permutation of its coordinates, i.e., 
    $(Z_1, Z_2, \ldots, Z_n) \sim (Z_{\pi(1)}, Z_{\pi(2)}, \ldots, Z_{\pi(n)})$ for any permutation $\pi(\cdot)$ of $\{1,2,\ldots, n\}$. 
\end{definition}

Popular treatment assignment mechanisms satisfying Definition \ref{def:exchange} include 
\textit{Bernoulli randomized experiment} (BRE) and \textit{completely randomized experiment} (CRE). 
Specifically, 
under a BRE, the treatment assignments $Z_i$'s are independent and identically distributed (i.i.d.) Bernoulli random variables with probability $p$, for some $p\in (0,1)$. 
Under a CRE, 
$m$ units will be randomly assigned to treatment, and the remaining $n - m$ units will be assigned to control, 
where $m$ and $n-m$ are fixed positive integers.

We use $\Unif[0,1]$ to denote a uniform random variable on $[0,1].$
We introduce 
$\vecle$ 
to denote element-wise inequality between two vectors: 
for any two vectors $\bm{\xi}$ and $\bm{\eta}$ of the same dimension, 
$\bm{\xi} \vecle \bm{\eta}$ 
if and only if each coordinate of $\bm{\xi}$ is less than or equal to the corresponding coordinate of  $\bm{\eta}$. 

\subsection{Sharp null hypotheses and imputation of potential outcomes}

\citet{Fisher35a} proposed to test the null hypothesis that the treatment has no effect for any unit, 
that is %
$\bm{Y}(1) = \bm{Y}(0)$ or 
equivalently
$\bm{\tau} = \bm{0}$. 
Such a null hypothesis is called a sharp null hypothesis, %
under which 
all missing potential outcomes %
can be imputed from the observed data. 
Specifically, under Fisher's null of no effect, 
we have $Y_i(1) = Y_i(0) = Y_i$ for all units $i=1,\ldots,n$. 
After imputing all the potential outcomes, we can know the null distribution of any test statistic exactly, 
which further provides an exact $p$-value. 
Such a testing procedure is often called the Fisher randomization test, and the resulting $p$-value is called the randomization $p$-value. 

The Fisher randomization test also works for general sharp null hypotheses, where the individual treatment effect for unit $i$ is $\delta_i$ for $i=1,\ldots, n$, where $\bm{\delta} = (\delta_1, \delta_2, \ldots, \delta_n)^\top \in \mathbb{R}^n$ is a  predetermined constant vector. %
That is  
\begin{align}\label{eq:null_delta}
	H_{\bm{\delta}}: \bm{\tau} = \bm{\delta}. 
\end{align}
Under the null $H_{\bm{\delta}}$ in \eqref{eq:null_delta}, we are able to impute all the potential outcomes. 
Specifically, under $H_{\bm{\delta}}$ and given the observed data $\bm{Z}$ and $\bm{Y}$, 
the imputed treatment and control potential outcome vectors for all units are, respectively, 
\begin{align}\label{eq:impute_potential_outcome}
	\bm{Y}_{\bm{Z}, \bm{\delta}}(1) & = \bm{Y}+ (\bm{1} - \bm{Z}) \circ \bm{\delta} = \bm{Z} \circ \bm{Y}(1) + (\bm{1}-\bm{Z}) \circ \{ \bm{Y}(0) + \bm{\delta} \}, 
	\nonumber
	\\ 
	\bm{Y}_{\bm{Z}, \bm{\delta}}(0) & = \bm{Y} - \bm{Z} \circ \bm{\delta} 
	= \bm{Z} \circ \{ \bm{Y}(1) - \bm{\delta} \} + (\bm{1}-\bm{Z}) \circ \bm{Y}(0), 
\end{align}
where we use the subscripts $\bm{Z}$ and $\bm{\delta}$ to emphasize that the imputed potential outcomes are deterministic functions of the observed treatment assignment, the null hypothesis of interest and the true potential outcomes. Importantly, 
the \textit{imputed} potential outcomes in \eqref{eq:impute_potential_outcome} are generally different from the \textit{true} potential outcomes; they are the same if and only if the sharp null $H_{\bm{\delta}}$ is true. 

\subsection{Fisher randomization test}\label{sec:frt}

There are at least two popular approaches for conducting Fisher randomization tests for general sharp null hypothesis $H_{\bm{\delta}}$ in \eqref{eq:null_delta}, represented by the textbooks of \citet{Rosenbaum02a} and \citet{Imbens15a}, respectively. The main distinction between these approaches lies in the choice of test statistic.  
For conciseness, 
we here focus on the approach from \citet{Rosenbaum02a}, and relegate the discussion of the other to Appendix 2 of the supplementary materials. 

Let $t(\cdot, \cdot): \mathcal{Z} \times \mathbb{R}^n \rightarrow \mathbb{R}$ denote a generic function with two arguments: the first a treatment assignment vector $\bm{z} \in \mathcal{Z}$, and the second an outcome vector $\bm{y} \in \mathbb{R}^n$. 
Following \citet{Rosenbaum02a}, we test with a test statistic of the form $t(\bm{Z}, \bm{Y}_{\bm{Z}, \bm{\delta}}(\bm{0}))$, which depends on the observed treatment assignment $\bm{Z}$ and the imputed control potential outcomes $\bm{Y}_{\bm{Z}, \bm{\delta}}(0)$ in \eqref{eq:impute_potential_outcome}. 
Often $t(\bm{Z}, \bm{Y}_{\bm{Z}, \bm{\delta}}(0))$ compares the imputed control potential outcomes of treated units to the observed values of the control units (e.g., $t$ might be the average difference between the treated outcomes adjusted by their corresponding $\delta_i$'s and the control outcomes).
If $H_{\bm{\delta}}$ is true, 
the imputed control potential outcome $\bm{Y}_{\bm{Z}, \bm{\delta}}(0)$ always equals the true $\bm{Y}(0)$, 
and thus
we can compute what we would have seen for any alternative treatment assignment vector $\bm{a}\in \mathcal{Z}$ as
$t(\bm{a}, \bm{Y}_{\bm{a}, \bm{\delta}}(0)) = t(\bm{a}, \bm{Y}(0)) = t(\bm{a}, \bm{Y}_{\bm{Z}, \bm{\delta}}(0))$. 
This means we can 
directly calculate the distribution of $t(\bm{Z}, \bm{Y}_{\bm{Z}, \bm{\delta}}(\bm{0}))$ if $H_{\bm{\delta}}$ were true.
Therefore, 
for the null $H_{\bm{\delta}}$, 
the imputed randomization distribution of the test statistic has the following tail probability for $c \in \mathbb{R}$:  
\begin{align}\label{eq:tail_prob_tilde}
G_{\bm{Z}, \bm{\delta}} (c) \equiv 
\Pr\left\{
t(\bm{A}, \bm{Y}_{\bm{Z}, \bm{\delta}}(0) )
\ge c
\right\}
= \sum_{\bm{a} \in \mathcal{Z}} \Pr(\bm{A} = \bm{a}) \I\left\{
t(\bm{a}, \bm{Y}_{\bm{Z},\bm{\delta}}(0))
\ge 
c
\right\},
\end{align}
where $\bm{A}$ denotes a generic random treatment assignment vector following the same distribution as $\bm{Z}$.
The corresponding randomization $p$-value is the tail probability evaluated at the observed value of the test statistic:
\begin{align}\label{eq:p_val_imp_control}
	p_{\bm{Z}, \bm{\delta}} \equiv  G_{\bm{Z}, \bm{\delta}} \left\{ t(\bm{Z},\bm{Y}_{\bm{Z}, \bm{\delta}}(0))\right\} 
	= 
	\sum_{\bm{a} \in \mathcal{Z}} \Pr(\bm{A} = \bm{a}) \I\left\{
	t(\bm{a}, \bm{Y}_{\bm{Z},\bm{\delta}}(0))
	\ge 
	t(\bm{Z}, \bm{Y}_{\bm{Z}, \bm{\delta}}(0))
	\right\}. 
\end{align} 
When %
$H_{\bm{\delta}}$ is true (i.e., $\bm{\tau} = \bm{\delta}$), the imputed randomization distribution $G_{\bm{Z}, \bm{\delta}}(\cdot)$ in \eqref{eq:tail_prob_tilde} reduces to $G_{\bm{Z}, \bm{\delta}} (c) = \Pr\{t(\bm{A}, \bm{Y}(0) ) \ge c\}$, and the randomization $p$-value $p_{\bm{Z}, \bm{\delta}} = p_{\bm{Z}, \bm{\tau}}$ is stochastically larger than or equal to $\Unif[0,1]$ (the difference is due solely to the discrete nature of the $p$-value distribution). 
That is, 
$p_{\bm{Z}, \bm{\delta}}$ in \eqref{eq:p_val_imp_control} is a valid  $p$-value for testing the sharp null $H_{\bm{\delta}}$.

In contrast to typical descriptions of randomization inference, 
we do not simplify the imputed potential outcomes in \eqref{eq:tail_prob_tilde} and \eqref{eq:p_val_imp_control} to the true ones, 
because later we will investigate the property of the randomization $p$-value even when the sharp null hypothesis fails. 
We emphasize that both  $G_{\bm{Z}, \bm{\delta}}(\cdot)$ and $p_{\bm{Z},\bm{\delta}}$ are deterministic functions 
of the treatment assignment $\bm{Z}$, 
the null hypothesis of interest $\bm{\delta}$,  
the true potential outcomes $(\bm{Y}(1),\bm{Y}(0))$, 
and the treatment assignment mechanism, 
where the latter two are fixed and the dependence on them is suppressed. 
Moreover, the randomness in $G_{\bm{Z}, \bm{\delta}} (\cdot)$ and $p_{\bm{Z}, \bm{\delta}}$ comes solely from the random treatment assignment $\bm{Z}$. 

We can of course 
work with test statistics that uses imputed treatment potential outcomes
(simply by switching the labels of treatment and control and changing the signs of the outcomes).
We can also work with test statistics that only use the observed outcomes; see Appendix A2 of the supplementary materials for details.

\section{Test statistics: properties and examples}
\label{sec:statistics}

\subsection{Three properties of test statistics}\label{sec:property_stat}

To provide the broader justification for the Fisher randomization test discussed above, we require test statistics with certain properties which we now discuss.

The first property is called \textit{effect increasing} (a term borrowed from \citealp{Rosenbaum02a}). %
Intuitively, an effect increasing statistic $t(\bm{z}, \bm{y})$, viewed as a function of the outcome vector $\bm{y}$ with $\bm{z}$ fixed, is 
increasing\footnote{Throughout the paper, 
an increasing function refers to a nondecreasing function, not a strictly increasing function. Similarly, a decreasing function refers to a nonincreasing function.} 
in those $y_i$'s with $z_i=1$ and 
decreasing 
in those $y_i$'s with $z_i=0$. 
We formally define it as follows. 
\begin{definition}\label{def:effect_increase}
	A statistic $t(\cdot, \cdot)$ is said to be effect increasing, 
	if $
	t(\bm{z}, \bm{y} + \bm{z} \circ \bm{\eta} + (1-\bm{z}) \circ \bm{\xi} ) \ge t(\bm{z}, \bm{y})
	$ 
	for any $\bm{z} \in \mathcal{Z}$ and $\bm{y}, \bm{\eta}, \bm{\xi} \in \mathbb{R}^n$ with 
	$\bm{\eta} \vecge \bm{0} \vecge \bm{\xi}$. 
\end{definition}

The second property is called \textit{differential increasing}. 
Intuitively, 
for a differential increasing statistic, 
if the outcomes for a subset of units are increased, 
then the change of the statistic is maximized when it happens that this subset of units received treatment.
We formally define it as follows. 

\begin{definition}\label{def:diff_increase}
	A statistic $t(\cdot, \cdot)$ is said to be differential increasing, 
	if 
	$
	t(\bm{z}, \bm{y}+\bm{a} \circ \bm{\eta}) - t(\bm{z}, \bm{y}) 
	\le 
	t(\bm{a}, \bm{y}+\bm{a} \circ \bm{\eta}) - t(\bm{a}, \bm{y}) 
	$
	for any $\bm{z}, \bm{a} \in \mathcal{Z}$ and $\bm{y}, \bm{\eta} \in \mathbb{R}^n$ with $\bm{\eta} \vecge \bm{0}$. 
\end{definition}

The third property is called \textit{distribution free}; see also \citet{Rosenbaum2007interference}. 
Different from the previous two properties, this property depends not only on the test statistic $t(\cdot, \cdot)$, but also on the treatment assignment mechanism. 
In particular, for 
a distribution free test statistic $t(\cdot, \cdot)$, the distribution of $t(\bm{Z}, \bm{y})$ does not depend on the  value of $\bm{y} \in \mathbb{R}^n$. 
We formally define it as follows. 

\begin{definition}\label{def:dist_free}
	A statistic $t(\cdot, \cdot)$ is said to be distribution free if, for any $\bm{y}, \bm{y}' \in \mathbb{R}^n$, $t(\bm{Z}, \bm{y})$ and $t(\bm{Z}, \bm{y}')$ follow the same distribution, where $\bm{Z}$ %
	follows the treatment assignment mechanism. 
\end{definition}

For the above three properties, one does not necessarily imply the other; see Appendix A3 of the supplementary materials for more details.
However, many commonly used test statistics are both effect increasing and differential increasing, and many rank-based test statistics (with 
appropriate handling of 
ties)  are also distribution free 
when the treatment assignment
is 
exchangeable (e.g., a 
Bernoulli
or completely randomized experiment, as defined in Section \ref{sec:potential_outcome});
see the next subsection for more details.

\subsection{Two classes of special test statistics}\label{sec:stat_example}

Two general classes of test statistics, which cover many test statistics commonly used for randomization tests, satisfy the properties introduced in Section \ref{sec:property_stat}.

The first class of test statistics has the following form: 
\begin{align}\label{eq:stat_class_y}
	t_1(\bm{z}, \bm{y}) = \sum_{i=1}^n z_i \psi_{1i}(y_i) - \sum_{i=1}^n (1-z_i)\psi_{0i}(y_i),
\end{align}
where the $\psi_{1i}(\cdot)$'s and $\psi_{0i}(\cdot)$'s are constant functions from $\mathbb{R}$ to $\mathbb{R}$ but can depend on the treatment assignment mechanism. 
For example, for a CRE with $m$ units receiving treatment, setting $\phi_{1i}(y) = y/m$ and $\phi_{0i}(y) = y/(n-m)$ gives the difference-in-means estimator.

For statistics of (\ref{eq:stat_class_y}), we have:

\begin{proposition}\label{prop:stat_property_dim}
Statistics $t_1(\cdot, \cdot)$ of the form in \eqref{eq:stat_class_y} are both effect increasing (Definition \ref{def:effect_increase}) and differential increasing (Definition \ref{def:diff_increase}) if $\psi_{1i}(\cdot)$ and $\psi_{0i}(\cdot)$ are 
monotone increasing functions for all $1\le i \le n$.
\end{proposition}

The second class of test statistics depends only on the ranks of the outcomes. 
For any vector $\bm{y} \in \mathbb{R}^n$ and for each $1\le i \le n$, 
we 
use $\rank_i(\bm{y})$ to denote the rank of the $i$th coordinate of $\bm{y}$, 
where larger rank corresponds to larger outcome value.
We assume that all ties are broken by unit ordering (see Appendix A1.1 of the supplementary materials for more details).
Then the second class of test statistics has the following form: 
\begin{align}\label{eq:stat_class_rank}
t_2(\bm{z}, \bm{y}) 
& =  
\sum_{i=1}^n z_i 
\phi (\rank_i(\bm{y})), 
\end{align}
where $\phi(\cdot)$ is a constant function from $\mathbb{R}$ to $\mathbb{R}$. 
For example, if we choose $\phi(r) = r$ to be the identity function, 
then 
$t_2(\bm{Z}, \bm{Y})$ reduces to the Wilcoxon rank sum statistic. 

These statistics satisfy Definitions \ref{def:effect_increase} and \ref{def:dist_free},
introduced in Section \ref{sec:property_stat}, under the following mild conditions:

\begin{proposition}\label{prop:stat_property_rank}
Under the tie-breaking rule discussed above, the statistic $t_2(\cdot, \cdot)$ in \eqref{eq:stat_class_rank} is
(a) effect increasing if $\phi(\cdot)$ is a monotone increasing function, and
(b) distribution free if the treatment assignment $\bm{Z}$ is exchangeable as in Definition \ref{def:exchange} as well as independent of the index ordering of the units (which can be achieved by randomly shuffling the order of units before analysis).
\end{proposition}

\subsection{Stephenson rank sum statistics}\label{sec:stephenson}

Despite the generality of the test statistics discussed before, in this subsection we focus on a special class of test statistics developed by \citet{Stephenson85a}, because of its often superior power to detect extreme effects. 
Stephenson rank sum statistics can be defined as the count of subsets of size $s$ in the data in which the largest response is in the treated group.
This can also be represented as a two-sample statistic of form \eqref{eq:stat_class_rank}, with a rank score function of 
$$
\phi(r) = \binom{r-1}{s-1} \quad \text{if } r \ge s, 
\ \ 
\text{ and }
\ \ 
\phi(r) = 0 \quad \text{otherwise},
$$
for some fixed integer $s \ge 2$. 
 The Stephenson rank sum statistic with $s=2$ is almost equivalent to the Wilcoxon rank sum statistic.\footnote{The Stephenson ranks with $s=2$ are each one less than the corresponding Wilcoxon ranks, leading to almost identical behavior \citep[1168]{Rosenbaum07a}.}
However, as $s$ increases beyond 2, the Stephenson ranks place more and more weight on the larger responses.

\cite{Rosenbaum07a} proposed to use Stephenson ranks to detect uncommon-but-dramatic responses to treatment. Intuitively, this is because as the subset size $s$ increases, it becomes increasingly likely that the largest response in a given subset will be one with an unusually large treatment effect.
Thus, compared to the difference-in-means and the Wilcoxon rank sum, 
which often perform well
against a constant location shift, the Stephenson rank sum has greater power against alternatives under which effects are heterogeneous and a few are highly positive.
As we will later see in our simulation studies, this advantage of Stephenson rank sum statistics is particularly relevant for our proposed theory and methods.

\section{Broader justification for Fisher randomization test}\label{sec:broader}

\subsection{Validity of randomization $p$-values for bounded nulls}\label{sec:valid_bound}

As discussed in Section \ref{sec:frt}, 
the randomization $p$-value $p_{\bm{Z}, \bm{\delta}}$ is always valid for testing the sharp null $H_{\bm{\delta}}$ in \eqref{eq:null_delta} given any test statistic $t(\cdot, \cdot)$. 
Our question is, can the randomization $p$-value $p_{\bm{Z}, \bm{\delta}}$ also be valid for testing a weak null hypothesis that does not fully specify all the individual treatment effects?

We will demonstrate shortly that,  
under certain intuitive conditions on the test statistic, 
the 
randomization $p$-value $p_{\bm{Z}, \bm{\delta}}$
for testing the sharp null hypothesis $H_{\bm{\delta}}$ is also valid for testing a bounded null, 
which states that each individual treatment effect is less than or equal to the corresponding coordinate of $\bm{\delta}$.  
We formally introduce 
this bounded null hypothesis as follows: 
\begin{align}\label{eq:null_less_delta}
H_{\vecle \bm{\delta}}: 
\bm{\tau} \vecle \bm{\delta},
\end{align}
where $\bm{\delta}$ is a constant vector in $\mathbb{R}^n$. 
Importantly,  the null $H_{\vecle \bm{\delta}}$ in \eqref{eq:null_less_delta} is composite, under which the exact distribution of the test statistic is generally unknown. 
The bounded null hypotheses can often be of interest in practice, 
and the choice of $\bm{\delta}$ in \eqref{eq:null_less_delta} depends on the application of interest. 
For example, 
we can choose $\bm{\delta}= \bm{0}$ 
if we are interested in whether the treatment has a positive effect for any unit. 
This is related to Pareto efficiency and monotonicity assumption in instrumental variable analysis (see the application in Appendix 7.2 of the supplementary materials).
If we are also interested in the magnitude of the effects, 
we can choose $\bm{\delta}= c \bm{1}$, under which the null hypothesis in \eqref{eq:null_less_delta} assumes that all individual treatment effects are at most $c$; 
see Section \ref{sec:ci_max_effect}.

The following theorem shows that the original Fisher randomization test, designed only 
for testing sharp null hypotheses, can also be valid for testing certain bounded null hypotheses.

\begin{theorem}\label{thm:null_broader_monotone_control}
    (a) If 
	the test statistic $t(\cdot, \cdot)$ is either differential increasing or effect increasing, 
	then for any constant $\bm{\delta} \in \mathbb{R}^n$, 
	the corresponding randomization $p$-value $p_{\bm{Z}, \bm{\delta}}$ in \eqref{eq:p_val_imp_control} for the sharp null $H_{\bm{\delta}}$ in \eqref{eq:null_delta} is also valid for testing the bounded null $H_{\vecle\bm{\delta}}$ in \eqref{eq:null_less_delta}, 
	i.e., under $H_{\vecle\bm{\delta}}$,  
	$
	\Pr(p_{\bm{Z}, \bm{\delta}} \le \alpha ) \le \alpha
	$
	for any $\alpha \in (0,1)$. 
	(b) If 
	the test statistic $t(\cdot, \cdot)$ is differential increasing, or it is both effect increasing and distribution free, 
	then for any possible assignment $\bm{z} \in \mathcal{Z}$, 
	the corresponding randomization $p$-value $p_{\bm{z}, \bm{\delta}}$ in \eqref{eq:p_val_imp_control}, 
	viewed as a function of $\bm{\delta}\in \mathbb{R}^n$, is monotone %
	increasing, 
	i.e., 
    $
	p_{\bm{z}, \bm{\delta}} \le p_{\bm{z}, \overline{\bm{\delta}}} 
	$
	for any $\bm{\delta} \vecle \overline{\bm{\delta}}$. 
\end{theorem}

From Theorem \ref{thm:null_broader_monotone_control}(a), 
the rejection of $H_{\bm{\delta}}$ also implies that there exists at least one unit $i$ whose treatment effect is larger than $\delta_i$. 
For the usual null $H_{c \bm{1}}$ of additive treatment effect $c$ for some $c \in \mathbb{R}$, 
the rejection of $H_{c \bm{1}}$ then implies that there exists at least one unit whose treatment effect is larger than $c$. 
Furthermore, 
Theorem
\ref{thm:null_broader_monotone_control}(a) also 
holds for general assignment mechanisms beyond BRE and CRE, such as blocking \citep{Miratrix2013} or rerandomization \citep{rerand2012, asymrerand}. 

Theorem \ref{thm:null_broader_monotone_control}(b) says all hypothesis in the ``shadow'' of a rejected hypothesis are also rejected.
For general outcomes, Theorem \ref{thm:null_broader_monotone_control}(b) helps provide meaningful confidence sets for $\bm{\tau}$ under practical computational constraints. 
For any $\alpha \in (0,1)$, 
the $1-\alpha$ confidence set for $\bm{\tau}$ by inverting randomization tests has the following equivalent forms: 
\begin{align}\label{eq:ci_both}
\left\{\bm{\delta}: p_{\bm{Z}, \bm{\delta}} > \alpha, \bm{\delta} \in \mathbb{R}^n \right\}  
& = 
\{\bm{\delta}: p_{\bm{Z}, \bm{\delta}} \le \alpha, \bm{\delta} \in \mathbb{R}^n \}^\c
= 
\bigcap_{\bm{\delta}: p_{\bm{Z}, \bm{\delta}} \le \alpha,   \bm{\delta} \in \mathbb{R}^n }
\left\{\bm{\eta}: \bm{\eta} \vecle \bm{\delta}, \bm{\eta} \in \mathbb{R}^n \right\}^\c.
\end{align}
Our confidence set consists of all points $\bm{\eta}$ not in the ``shadow'' of any given hypothesis that can be rejected. 
Therefore, for any vector $\bm{\delta}$ with a corresponding randomization $p$-value less than or equal to $\alpha$, 
the $1-\alpha$ confidence set \eqref{eq:ci_both} must be a subset of 
$\{\bm{\eta}: \bm{\eta} \vecle \bm{\delta}, \bm{\eta} \in \mathbb{R}^n \}^\c$, 
a set in which no element is uniformly bounded by $\bm{\delta}$.

\subsection{Inference for the maximum individual treatment effect}\label{sec:ci_max_effect}

Consider a sharp null hypothesis $H_{c\bm{1}}$ of constant treatment effect $c$ for some $c\in \mathbb{R}$. 
From Theorem \ref{thm:null_broader_monotone_control}(a), 
given our conditions, the randomization $p$-value for the sharp null $H_{c\bm{1}}$ can also be valid for the bounded null $H_{\vecle c\bm{1}}$, which, letting $\tau_{\max} \equiv \max_{1\le i\le n} \tau_i$, is equivalent to the null hypothesis that $\tau_{\max}\le c$. 
This immediately implies that 
inverting randomization tests for a sequence of constant treatment effects 
can provide confidence sets for the maximum individual effect $\tau_{\max}$. 
Moreover, from Theorem \ref{thm:null_broader_monotone_control}(b), 
the resulting confidence sets 
can be
intervals %
of forms $(\underline{c}, \infty)$ or $[\underline{c}, \infty)$. 
We summarize the results in the following corollary.

\begin{corollary}\label{cor:ci_max_imp_control}
	(a) If the test statistic $t(\cdot, \cdot)$ is either differential increasing or effect increasing, 
	then for any $\alpha \in (0,1)$, 
	the set 
	$
	\{c: p_{\bm{Z}, c\bm{1}} > \alpha, c \in \mathbb{R} \}
	$
	is a $1-\alpha$ confidence set for the maximum individual effect $\tau_{\max}$. 
	(b) 
	If the test statistic is differential increasing, or it is both effect increasing and distribution free, 
	then  
	the confidence set must have the form of $(\underline{c}, \infty)$ or $[\underline{c}, \infty)$ with $\underline{c} = \inf \{c: p_{\bm{Z}, c\bm{1}} > \alpha, c \in \mathbb{R} \}$. 
\end{corollary}

The confidence intervals in Corollary \ref{cor:ci_max_imp_control}
can also be thought of as intervals stating where at least some of the individual treatment effects lie, and the more homogeneous the effects are, the more individual effects these intervals will contain.

\section{Randomization test for quantiles of individual treatment effects}\label{sec:inf_quant_all}

\subsection{Randomization test for general null hypotheses}\label{sec:ri_general_null}
For a general null hypothesis of interest, 
e.g., $\bm{\tau} \in \mathcal{H} \subset \mathbb{R}^n$, 
we can always obtain a valid $p$-value by maximizing the randomization $p$-value $p_{\bm{Z}, \bm{\delta}}$ in \eqref{eq:p_val_imp_control} over $\bm{\delta} \in \mathcal{H}$. 
That is, 
$\sup_{\bm{\delta} \in \mathcal{H}} p_{\bm{Z}, \bm{\delta}}$ 
is a valid $p$-value for testing the null hypothesis of $\bm{\tau} \in \mathcal{H}$. 
Unfortunately, 
{\rev 
such an optimization can be quite challenging, due to the complicated dependence of the imputed null distribution $G_{\bm{Z}, \bm{\delta}} (\cdot)$ in \eqref{eq:tail_prob_tilde} on the hypothesized effects $\bm{\delta}$. 
Moreover, the $p$-value $p_{\bm{Z}, \bm{\delta}}$ is the value of $G_{\bm{Z}, \bm{\delta}} (\cdot)$ evaluated at the realized value of the test statistic, $t(\bm{Z},\bm{Y}_{\bm{Z}, \bm{\delta}}(0))$, which itself 
also depends on
$\bm{\delta}$. 
To ease computation, we use a distribution free test statistic, reducing the optimization for $\sup_{\bm{\delta} \in \mathcal{H}} p_{\bm{Z}, \bm{\delta}}$ to one for the realized value of the test statistic. 

From Definition \ref{def:dist_free}, 
the imputed randomization distribution \eqref{eq:tail_prob_tilde} of a distribution free test statistic $t(\cdot, \cdot)$ under the null hypothesis $H_{\bm{\delta}}$ in \eqref{eq:null_delta} %
has the following equivalent forms: 
\begin{align}\label{eq:G0}
G_{\bm{Z}, \bm{\delta}} (c) & =  \sum_{\bm{a} \in \mathcal{Z}} \Pr(\bm{A} = \bm{a}) \I\left\{
t(\bm{a}, \bm{Y}_{\bm{Z},\bm{\delta}}(0))
\ge 
c
\right\}
= \sum_{\bm{a} \in \mathcal{Z}} \Pr(\bm{A} = \bm{a}) \I\left\{
t(\bm{a}, \bm{y})
\ge 
c
\right\}
= 
G_0(c), 
\end{align}
where $\bm{y} \in \mathbb{R}^n$ can be any fixed vector and 
$G_0(c)$ is a tail probability function that does not depend on the observed assignment $\bm{Z}$ or the null hypothesis of interest $\bm{\delta}$. 
By the definition in \eqref{eq:p_val_imp_control} and the fact that $G_0(c)$ is 
decreasing 
in $c$, 
\eqref{eq:G0} implies that 
\begin{align}\label{eq:sup_p_val}
\sup_{\bm{\delta} \in \mathcal{H}} p_{\bm{Z}, \bm{\delta}} = 
\sup_{\bm{\delta} \in \mathcal{H}} G_{\bm{Z}, \bm{\delta}} \left\{ t(\bm{Z},\bm{Y}_{\bm{Z}, \bm{\delta}}(0))\right\} 
= 
\sup_{\bm{\delta} \in \mathcal{H}} G_0 \left\{ t(\bm{Z},\bm{Y}_{\bm{Z}, \bm{\delta}}(0))\right\} 
\le 
G_0 \left\{ \inf_{\bm{\delta} \in \mathcal{H}} t(\bm{Z},\bm{Y}_{\bm{Z}, \bm{\delta}}(0))\right\}. 
\end{align} 
Therefore, the right hand side of \eqref{eq:sup_p_val} is also a valid $p$-value for testing the null of $\bm{\tau} \in \mathcal{H}$, 
and 
more importantly, its optimization becomes much simpler, 
because it now involves only minimization over a known and generally closed-form function of $\bm{\delta}$. 
Furthermore, as demonstrated in the next subsection, when considering null hypotheses on quantiles of individual effects and using rank sum statistics, such an optimization can have a closed-form solution.
For completeness, we summarize the results below. 

\begin{theorem}\label{thm:general_null}
	For any distribution free test statistic and any constant region $\mathcal{H}\subset \mathbb{R}^n$, %
	the supremum of the randomization $p$-value $\sup_{\bm{\delta} \in \mathcal{H}} p_{\bm{Z}, \bm{\delta}}$,  
    as well as its upper bound on the right hand side of \eqref{eq:sup_p_val}, 
	is valid for testing the null hypothesis of $\bm{\tau} \in \mathcal{H}$. 
\end{theorem}

We give two additional remarks. 
First, the distribution free property for test statistics is also utilized by \citet{Rosenbaum2007interference} for analyzing the magnitude of treatment effects in the presence of interference. 
In particular, it helps derive the distributions of certain statistics under a uniformity trial. 
By contrast, the distribution free property here mainly 
helps ease the computation of the valid $p$-value in \eqref{eq:sup_p_val}. 
{\rev 
Second, as suggested by a reviewer, another way to overcome the complex dependence of $G_{\bm{Z}, \bm{\delta}}(\cdot)$ on $\bm{\delta}$ is through asymptotic approximation, under which we can approximate $G_{\bm{Z}, \bm{\delta}}(\cdot)$ by a closed-form expression of $\bm{\delta}$. 
However, as discussed shortly, when considering null hypotheses on quantiles of individual effects, 
some units are allowed to have infinitely large individual effects. 
Consequently, we will consider cases allowing elements of $\bm{\delta}$ to take extreme and even infinite values, 
which 
may 
destroy the asymptotic approximation for $G_{\bm{Z}, \bm{\delta}}(\cdot)$.
On the contrary, using Theorem \ref{thm:general_null} with rank-based distribution free test statistics,  
the inference will be not only finite-sample valid but also robust to extreme values in $\bm{\delta}$; see also the related discussion in Remark \ref{rmk:rank_robust}. 
}

\subsection{Inference for quantiles of individual treatment effects}\label{sec:test_quantile}

We sort the true individual treatment effects in an increasing order: $\tau_{(1)} \le \tau_{(2)} \le \ldots \le \tau_{(n)}$, 
where $\tau_{(n)}$ is equivalently the maximum individual effect $\tau_{\max}$ studied in Section \ref{sec:ci_max_effect}. 
In this subsection, instead of only the maximum individual effect, we intend to infer general quantiles of the individual treatment effects $\tau_{(k)}$'s for $1 \le k \le n$, where $k=n$ corresponds to the maximum or largest individual effect, $k=n-1$ corresponds to the second largest individual effect, and so on.
Specifically, for any $1 \le k \le n$ and any constant $c \in \mathbb{R}$, we consider the following null hypothesis that the individual effect of rank $k$ is at most $c$: 
\begin{align}\label{eq:H_nc_k} 
H_{k, c}: \tau_{(k)} \le c. 
\end{align}
In the special case of $k=n$, $H_{n, c}$ reduces to the bounded null $H_{\vecle c\bm{1}}$ as in \eqref{eq:null_less_delta}.
Define 
$
\mathcal{H}_{k, c} = \{\bm{\delta} \in \mathbb{R}^n: \delta_{(k)} \le c \} \subset \mathbb{R}^n
$
as the set of vectors whose %
elements of rank $k$
are smaller than or equal to $c$. 
Then the null hypothesis $H_{k, c}$ in \eqref{eq:H_nc_k} can be equivalently represented as $\bm{\tau} \in \mathcal{H}_{k, c}$.

We consider testing the null hypothesis $H_{k, c}$ in \eqref{eq:H_nc_k} using the randomization $p$-value $p_{\bm{Z}, \bm{\delta}}$ in \eqref{eq:p_val_imp_control} with an effect increasing and distribution free statistic $t(\cdot, \cdot)$.  
Specifically, we focus on 
randomized experiments with exchangeable treatment assignment including BRE and CRE, 
and use the test statistic  
in \eqref{eq:stat_class_rank} with a monotone %
increasing 
function $\phi(\cdot)$ and 
a random tie-breaking rule based on the ordering of the units, assuming
that the ordering of the units has been randomly permuted and is independent of the treatment assignment.  
For descriptive convenience, we call such a statistic a rank score statistic, formally defined as follows. 

\begin{definition}\label{def:rank_score}
	A statistic $t(\cdot, \cdot)$ is a rank score statistic, if it 
	 can be written as $t(\bm{z}, \bm{y}) =  \sum_{i=1}^n z_i \phi (\rank_i(\bm{y}))$, where the score function $\phi(\cdot)$ is monotone increasing and the rank function $\rank (\cdot)$ uses a random tie-breaking rule.  
\end{definition}

From Proposition \ref{prop:stat_property_rank}, 
under experiments with exchangeable treatment assignment, 
the rank score statistic in Definition \ref{def:rank_score} is both effect increasing and distribution free. 
Consequently, 
from Theorem \ref{thm:general_null}, 
to test the null hypothesis $H_{k, c}$ in \eqref{eq:H_nc_k}, it suffices to minimize the value of the test statistic $t(\bm{Z}, \bm{Y}_{\bm{Z}, \bm{\delta}}(0))$ over $\bm{\delta} \in \mathcal{H}_{k, c}$. 
As we demonstrate below, this minimization has a closed-form solution. 
Let $m = \sum_{i=1}^n Z_i$ be the number of treated units, 
and $\mathcal{I}_k$ be the set of indices of treated units with the largest $\min(n-k, m)$ observed outcomes for $1\le k \le n$; 
when $k=n$, $\mathcal{I}_n$ is an empty set. 
We then define a column vector as follows: 
\begin{align}\label{eq:xi_k}
\bm{\xi}_{k,c} = (\xi_{1k,c}, \xi_{2k, c}, \ldots, \xi_{nk, c}) \in \mathbb{R}^n, 
\quad 
\text{ where } \ 
\xi_{ik, c} = 
\begin{cases}
\infty, & \text{if } i \in \mathcal{I}_k, \\
c, & \text{otherwise},
\end{cases}
\quad
(1\le i \le n). %
\end{align}

\begin{theorem}\label{thm:test_Hkc}
	Take a randomized experiment with exchangeable treatment assignment as in Definition \ref{def:exchange}, 
	and any rank score statistic $t(\cdot, \cdot)$
	in Definition \ref{def:rank_score}.
 Then, for any $1 \le k \le n$ and any constant $c \in \mathbb{R}$, 
	\begin{align}\label{eq:p_kc}
	p_{\bm{Z}, k, c}
	& \equiv 
	\sup_{\bm{\delta} \in \mathcal{H}_{k,c}} p_{\bm{Z}, \bm{\delta}} 
	= 
	G_0 \left\{ \inf_{\bm{\delta} \in \mathcal{H}_{k,c}} t(\bm{Z},\bm{Y}_{\bm{Z}, \bm{\delta}}(0))\right\}
	= 
	G_0 \left\{ t(\bm{Z},\bm{Y} - \bm{Z} \circ \bm{\xi}_{k,c} )\right\}
	\end{align} 
	is a valid $p$-value for testing the null hypothesis $H_{k, c}$ in \eqref{eq:H_nc_k}, 
	where $G_0$ and $\bm{\xi}_{k,c}$ are defined in \eqref{eq:G0} and \eqref{eq:xi_k}, respectively. 
	Specifically, 
	under $H_{k, c}$, 
	$\Pr(p_{\bm{Z}, k, c} \le \alpha) \le \alpha$ for any $\alpha \in (0,1)$. 
\end{theorem}

From Theorem \ref{thm:test_Hkc}, we are able to test whether any quantile of the individual treatment effects is bounded above by any constant. 
When $k=n$, 
the null hypothesis $H_{n, c}$ in \eqref{eq:H_nc_k} reduces to $H_{c\bm{1}}$ in \eqref{eq:null_less_delta}, 
the vector $\bm{\xi}_{n, c}$ reduces to $c\bm{1}$, 
the $p$-value $p_{\bm{Z}, k, c}$ in \eqref{eq:p_kc} reduces to $p_{\bm{Z}, c\bm{1}}$ in \eqref{eq:p_val_imp_control}, 
and Theorem \ref{thm:test_Hkc} reduces to a special case of Theorem \ref{thm:null_broader_monotone_control}. 

In Theorem \ref{thm:test_Hkc}, 
intuitively, when testing null hypothesis $H_{k,c}$ of $\tau_{(k)} \le c$, we allow the $\tau_{(j)}$'s with $j>k$ to be arbitrarily large. 
Moreover, we assign these infinity values to the treated units with largest outcomes to minimize the value of the test statistic, or equivalently to maximize the randomization $p$-value. 
{\rev As a result,} 
the calculation of $p_{\bm{Z}, k, c}$ in \eqref{eq:p_kc} involves ranking vectors with infinite elements. 
In practice, we can replace those infinite elements of $\bm{\xi}_{k,c}$ by any constant larger than the difference between the maximum treated observed outcome and the minimum control observed outcome, and the value of $p_{\bm{Z}, k, c}$ will remain the same. 
For simplicity, we 
use infinity, and view two negative infinite elements as equal in ranking.
This is also compatible with the \textsf{R} software.

\begin{rmk}\label{rmk:rank_robust}
    The ranking aspect of rank statistics plays an important role in making the test statistic distribution free and thus eases the computation. 
    From Theorem \ref{thm:test_Hkc} and the discussion before, 
    the rank statistic also has the advantage that it is robust to extreme outcome values. Specifically, although we allow some individual treatment effects to be infinity when maximizing the $p$-value over $\bm{\delta}\in \mathcal{H}_{k,c}$, 
    the rank statistic is still able to provide significant $p$-values against the null; see, e.g.,  the simulation in Section \ref{sec:simu_visual} and the application in Section \ref{sec:profession}. 
\end{rmk}

\subsection{Confidence intervals for quantiles of individual treatment effects}\label{sec:ci_qi}

Similar to Corollary \ref{cor:ci_max_imp_control}, 
we are now able to construct confidence sets for quantiles of the individual treatment effects based on  Theorem \ref{thm:test_Hkc}. 
Moreover, the $p$-value $p_{\bm{Z}, k, c}$ in \eqref{eq:p_kc} 
enjoys a certain monotonicity property that 
helps simplify the confidence sets.  
We summarize the results in the following theorem.

\begin{theorem}\label{thm:con_interval_nc}
	Take a randomized experiment with exchangeable treatment assignment as in Definition \ref{def:exchange}, and any rank score statistic in Definition \ref{def:rank_score}.
 Then, for any $1\le k\le n$ and any $\alpha \in (0,1)$, we have
	(a) for any fixed $\bm{z}$ and $k$, $p_{\bm{z}, k, c}$, defined as in \eqref{eq:p_kc}, is increasing in $c$, and 
	(b) a $1-\alpha$ confidence set for $\tau_{(k)}$ is $\{c: p_{\bm{Z}, k, c} > \alpha, c \in \mathbb{R} \}$, which must have the form of 
	$(\underline{c}, \infty)$ or $[\underline{c}, \infty)$ with $\underline{c} = \inf \{c: p_{\bm{Z}, k, c} > \alpha, c \in \mathbb{R} \}$. 
\end{theorem}

Theorem \ref{thm:con_interval_nc} generalizes Corollary \ref{cor:ci_max_imp_control} to all quantiles of the individual effects. 
The intervals from Theorem \ref{thm:con_interval_nc} also give a sense of the sizes of effects across all units, and help understand effect heterogeneity. 
For a specific $k$, the interval for $\tau_{(k)}$ states where the largest $n-k+1$ individual effects lie with certain confidence.

{\rev Importantly, the inference in Theorem \ref{thm:con_interval_nc} on quantiles of individual effects}
can sometimes be more appropriate than \citet{Neyman23a}'s inference on the average treatment effect. 
Specifically, 
when the outcomes have heavy tails and outliers, 
the average effect may be sensitive to these outliers, 
and the finite population asymptotic approximation \citep{lidingclt2016} may work poorly. 
However, the quantiles are more robust to outliers than the average. 
Moreover, the inference in Theorem \ref{thm:con_interval_nc} is exactly valid in finite samples and does not require any large-sample approximation. 

\begin{rmk}	
	When $k \le n-m$, 
	$\mathcal{I}_k$ in \eqref{eq:xi_k} contains the indices of all treated units, 
	whose treatment effects are all hypothesized to be arbitrarily large. 
	The resulting confidence interval for $\tau_{(k)}$ is usually
	the uninformative $(-\infty, \infty).$  
	Therefore, $m$, the size of the treatment group, can affect the performance of the method in Theorems \ref{thm:test_Hkc} and \ref{thm:con_interval_nc}. 
	Moreover, generally larger $m$ can lead to more quantiles of effects with informative confidence intervals. 
	This asymmetric role of the treatment and control group sizes comes from the fact that the randomization $p$-value $p_{\bm{Z}, \bm{\delta}}$ uses only the imputed control potential outcomes. 
	When the treatment group size is expected to be small, we may want to 
    use the randomization $p$-value involving the imputed treatment potential outcomes, which can be achieved
    by switching the labels for treatment and control and changing the signs of the outcomes. 
\end{rmk}

\subsection{Inference for the number of units with effects larger than a threshold}\label{sec:number_effect_larger_c}

We now consider an equivalent form of the null hypothesis $H_{k, c}$ in \eqref{eq:H_nc_k}, which relates to the proportion of units with effects larger than a certain threshold. 
Because such a quantity can often be of interest in practice, we 
give a detailed discussion on its statistical inference below. 

For any constant $c \in \mathbb{R}$, define 
\begin{align}\label{eq:n_c}
	n(c) = \sum_{i=1}^{n} \I(\tau_i > c)
\end{align}
as the number of units whose treatment effects are larger than $c$. 
We can verify that, for any $1\le k \le n$ and $c\in \mathbb{R}$,  $\tau_{(k)} \le c$ if and only if $n(c) \le n-k$. 
Therefore, the null hypothesis $H_{k,c}$ in \eqref{eq:H_nc_k} has the following equivalent forms: 
\begin{align}\label{eq:H_nc_k_equ} 
H_{k, c}: \tau_{(k)} \le c  \Longleftrightarrow \bm{\tau} \in \mathcal{H}_{k, c} \Longleftrightarrow n(c) \le n-k, \qquad (1\le k \le n, c\in \mathbb{R}).
\end{align}

Theorem \ref{thm:test_Hkc} immediately implies that we are able to test null hypotheses about the number of units with effects larger than any threshold, as shown in the following theorem. 
For descriptive convenience, we define $p_{\bm{z}, 0, c} = 1$ for any $\bm{z}$ and $c$, due to the fact that $H_{0,c}$ is true by definition.

\begin{corollary}\label{cor:infer_number_larger_threshold}
	Take a randomized experiment with exchangeable treatment assignment as in Definition \ref{def:exchange}, and any rank score statistic in Definition \ref{def:rank_score}.
    Then, (a) the $p$-value $p_{\bm{Z}, k, c}$ in \eqref{eq:p_kc} is valid for testing the null hypothesis $H_{k, c}$ as given in \eqref{eq:H_nc_k} and \eqref{eq:H_nc_k_equ} for any $1\le k \le n$ and $c \in \mathbb{R}$; 
	(b) 
	{\rev for any fixed $\bm{z}$ and $c$, $p_{\bm{z}, k, c}$, defined as in \eqref{eq:p_kc}, is decreasing in $k$;}
	(c) 
	a $1-\alpha$ confidence set for $n(c)$ in \eqref{eq:n_c}, i.e, the number of units with effects larger than $c$, is 
	$
	\{n-k: p_{\bm{Z}, k, c} > \alpha, 0\le k \le n \}, 
	$
    and this set has
  the form of 
	$\{j: n -\overline{k}  \le j \le n\}$
	with $\overline{k} = \sup \{k: p_{\bm{Z}, k, c} > \alpha, 0\le k \le n \}$. 
\end{corollary}

From Theorem \ref{thm:con_interval_nc} and Corollary \ref{cor:infer_number_larger_threshold}, 
 we can know that, by construction, 
the $1-\alpha$ lower confidence limit of $n(c)$ is 
equivalently 
the number of quantiles of individual effects
$\tau_{(k)}$'s whose $1-\alpha$ 
confidence intervals do not cover $c$.
See Appendix A1.3 of the supplementary materials for further connections with related work.

\subsection{Simultaneous inference for quantiles $\tau_{(k)}$'s and numbers $n(c)$'s }\label{sec:simultaneous_equivalence}

From Theorem \ref{thm:con_interval_nc}, 
we are able to construct $1-\alpha$ confidence intervals for all quantiles of individual treatment effects $\tau_{(k)}$'s. 
Similarly, 
from Corollary \ref{cor:infer_number_larger_threshold}, 
we are able to construct $1-\alpha$ confidence intervals for the numbers $n(c)$'s of units with effects larger than the thresholds $c$'s. 
As demonstrated shortly, 
these confidence intervals
will cover their corresponding truth simultaneously with probability at least $1-\alpha$, in the sense that there is no need for any correction due to multiple analyses.  

The set $\mathcal{H}_{k,c}$ introduced in Section \ref{sec:test_quantile} has the following equivalent forms: 
\begin{align}\label{eq:setHkc}
\mathcal{H}_{k, c} = \left\{\bm{\delta} \in \mathbb{R}^n: \delta_{(k)} \le c \right\} = 
\Big\{\bm{\delta} \in \mathbb{R}^n:  \sum_{i=1}^{n} \I(\delta_i > c) \le n-k \Big\} 
\subset \mathbb{R}^n,
\end{align}
in parallel with the equivalence relationship in  \eqref{eq:H_nc_k_equ}. 
Using \eqref{eq:setHkc}, 
we can represent the confidence intervals for the quantiles $\tau_{(k)}$'s and the numbers $n(c)$'s as confidence sets for the treatment effect vector $\bm{\tau}$. 
Specifically, 
for any $1\le k \le n$, the $1-\alpha$ confidence interval for $\tau_{(k)}$ in Theorem \ref{thm:con_interval_nc} has the following equivalent form as a $1-\alpha$ confidence set for $\bm{\tau}$: 
\begin{align}\label{eq:con_tau_k_equi_form}
	\tau_{(k)} \in \{c: p_{\bm{Z}, k, c} > \alpha, c \in \mathbb{R} \} 
	& \Longleftrightarrow 
	\bm{\tau} \in 
	\bigcap_{c: p_{\bm{Z}, k, c} \le \alpha } \mathcal{H}_{k, c}^\c,
\end{align}
and 
for any $c \in \mathbb{R}$, the $1-\alpha$ confidence interval for $n(c)$ in Corollary \ref{cor:infer_number_larger_threshold} has the following equivalent form: 
\begin{align}\label{eq:con_n_c_equi_form}
	n(c) \in \{n-k: p_{\bm{Z}, k, c} > \alpha, 0\le k \le n \}
	& \Longleftrightarrow 
	\bm{\tau} \in 
	\bigcap_{k: p_{\bm{Z}, k, c} \le \alpha } \mathcal{H}_{k, c}^\c.
\end{align}
Therefore, the combination of all the confidence intervals for the $\tau_{(k)}$'s can be viewed as a confidence set for all the individual treatment effects $\bm{\tau}$, which is the intersection of the sets in \eqref{eq:con_tau_k_equi_form} over 
$1\le k \le n.$
Similarly, the combination of all the confidence intervals for the $n(c)$'s can be viewed as a confidence set for $\bm{\tau}$, which is the intersection of the sets in \eqref{eq:con_n_c_equi_form} over all $c\in \mathbb{R}$. 
As shown in the following theorem, these two confidence sets for $\bm{\tau}$ are the same, 
and more importantly, 
they are indeed confidence sets with at least $1-\alpha$ coverage probability. 

\begin{theorem}\label{thm:simultaneous_valid}
	Under 
    a randomized experiment with exchangeable treatment assignment as in Definition \ref{def:exchange} and 
	using the $p$-value $p_{\bm{Z}, k, c}$ in \eqref{eq:p_kc} with any rank score statistic $t(\cdot, \cdot)$ in Definition \ref{def:rank_score}, 
	for any $\alpha \in (0,1)$, 
	the intersection of $1-\alpha$ confidence intervals for all $\tau_{(k)}$'s, viewed as a confidence set for the individual treatment effect vector $\bm{\tau}$, is the same as that for all $n(c)$'s. 
	In particular, it has the following equivalent forms: 
	\begin{align}\label{eq:joint_set}
		\bigcap_{k = 1}^n \bigcap_{c: p_{\bm{Z}, k, c} \le \alpha } \mathcal{H}_{k, c}^\c 
		= 
		\bigcap_{c \in \mathbb{R}}
		\bigcap_{k: p_{\bm{Z}, k, c} \le \alpha } \mathcal{H}_{k, c}^\c
		= 
		\bigcap_{k, c: \ p_{\bm{Z}, k, c} \le \alpha } 
		\mathcal{H}_{k, c}^\c
	\end{align}
	Moreover, %
	it 
	has at least $1-\alpha$ probability to cover the true individual treatment effects $\bm{\tau}$, i.e., 
	\begin{align*}
	\Pr\Big(
	\bm{\tau} \in 
	\bigcap_{k, c: p_{\bm{Z}, k, c} \le \alpha } \mathcal{H}_{k, c}^\c 
	\Big) \ge 1-\alpha. 
	\end{align*}
\end{theorem}

From Theorem \ref{thm:simultaneous_valid}, in practice, we can simultaneously construct confidence intervals for all quantiles of individual effects, or equivalently numbers of units with effects larger than any threshold. 
These intervals for individual effects are simply projections (summaries) of the complex confidence set in \eqref{eq:joint_set} of possible individual treatment effect vectors.
Moreover, these confidence intervals can be conveniently visualized, as illustrated in Section \ref{sec:moti_example} using Figure \ref{fig:conf_profession}(b).
    Note that the confidence interval for the maximum individual effect is the same as that under usual randomization inference with a constant treatment effects assumption. 
    By the simultaneous validity in Theorem \ref{thm:simultaneous_valid}, 
    we can get confidence intervals on all quantiles of individual effects as free lunches, 
    because these additional intervals will not reduce our confidence levels.

\section{Extension: two-sided alternatives and effect range}
\label{sec:twosided}

In the previous discussion, we mainly focused on one-sided testing for the treatment effect $\bm{\tau}$, where the alternative hypotheses favor larger treatment effects. 
In fact, these results immediately imply that we are also able to test alternative hypotheses favoring smaller treatment effects. 
We can achieve this simply by multiplying the outcomes by $-1$ or by switching the labels for treatment and control.
By Bonferroni correction, we can also construct confidence intervals for all quantiles of individual effects using both sides of alternatives.

It is also 
possible to combine the confidence intervals for the maximum and minimum individual effects into a single confidence statement about the range of treatment effects. 
Suppose $\hat{\tau}_{\max}^L$
is a $1-\alpha/2$ lower confidence limit for the maximum individual effect $\tau_{\max}$, 
and 
$\hat{\tau}_{\min}^U$
is a $1-\alpha/2$ upper confidence limit for the minimum individual effect $\tau_{\min}$. 
Using Bonferroni correction, 
we are $1-\alpha$ confident that the effect range $\tau_{\max} - \tau_{\min}$ is at least $\hat{\tau}_{\max}^L - \hat{\tau}_{\min}^U$, 
based on which we are able to test whether the treatment effect is constant, an issue discussed in detail in \citet{Ding16a}. 
For completeness, 
we summarize the results in the following theorem. 

\begin{theorem}\label{thm:effect_range}
	Suppose that $[\hat{\tau}_{\max}^L, \infty)$ is a $1-\alpha/2$ confidence interval for $\tau_{\max}$, 
	and $(-\infty, \hat{\tau}_{\min}^U]$ is a $1-\alpha/2$ confidence interval for $\tau_{\min}$. 
	Then 
	(a) $[\max\{\hat{\tau}_{\max}^L - \hat{\tau}_{\min}^U, 0\}, \infty)$ is a $1-\alpha$ confidence interval for the effect range $\tau_{\max} - \tau_{\min}$; 
	(b) for the null hypothesis of constant treatment effect, i.e., $H_{c\bm{1}}$ holds for some $c\in \mathbb{R}$, rejecting the null if and only if $\hat{\tau}_{\max}^L - \hat{\tau}_{\min}^U>0$ leads to a valid test at significance level $\alpha$. 
\end{theorem}

\section{Simulation studies}\label{sec:simu}

\subsection{A simulation study for inferring the maximum individual effect}\label{sec:simu_max}

We first conduct a simulation study to investigate the power of randomization tests with different test statistics for detecting positive maximum individual treatment effect $\tau_{\max}$, including the settings where the average treatment effect is close to zero or even negative. 
In particular, we investigate difference-in-means, Wilcoxon rank sum and Stephenson rank sum as test statistics in a completely randomized experiment.
From Theorem \ref{thm:null_broader_monotone_control}, the randomization $p$-values $p_{\bm{Z}, \bm{\delta}}$ with these test statistics are all valid for testing the bounded null $H_{\vecle \bm{\delta}}$.

We generate the potential outcomes as i.i.d.\  $(Y_i(0), \tau_i)$ pairs from the following model and randomize half of the units to treatment group and the remaining to control group:
\begin{align}\label{eq:generate}
	\begin{pmatrix}
    Y_i(0)\\
    \tau_i
    \end{pmatrix}
    \sim 
    \mathcal{N}
    \left(
    \begin{pmatrix}
    0\\
    \tau_0
    \end{pmatrix}, \ 
    \begin{pmatrix}
    1 & \rho \omega \\
    \rho\omega & \omega^2
    \end{pmatrix}
    \right), 
    \quad 
    Y_i(1) = Y_i(0) + \tau_i, 
\end{align}
where $\tau_0$ characterizes the magnitude of the average treatment effect, 
$\rho$ reflects the correlation between the individual treatment effect and the control potential outcome, 
and $\omega$ characterizes the variability of the individual treatment effect. 
If $\rho$ takes a positive value, then units with larger control potential outcomes tend to have larger individual treatment effects.

We test the bounded null $H_{\vecle \bm{0}}$ that all individual treatment effects are non-positive (or equivalently 
$\tau_{\max}\le 0$).
Figures \ref{fig:max_effect} shows the power of the 
test 
using different test statistics with sample size $n=120$ and significance level $ 0.1$, under different parameter values of $(\tau_0, \omega, \rho)$. 

We first see that the performance of the difference-in-means and Wilcoxon rank sum (equivalent to Stephenson rank sum with $s=2$ under the CRE) statistics are very similar. 
When the average treatment effect is 
non-positive, 
both of them have almost no power to detect positive maximum  effect. 
However, when $s$ increases the Stephenson rank sum statistic is
able to detect the presence of positive maximum effects, even when the average treatment effect is non-positive (see bottom-left).

\begin{figure}[t!]
	\centering
		\includegraphics[width=0.75\linewidth]{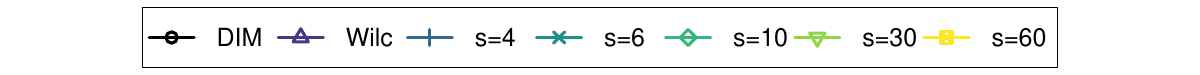}
	\begin{subfigure}{.33\textwidth}
		\centering
		\includegraphics[width=1\linewidth]{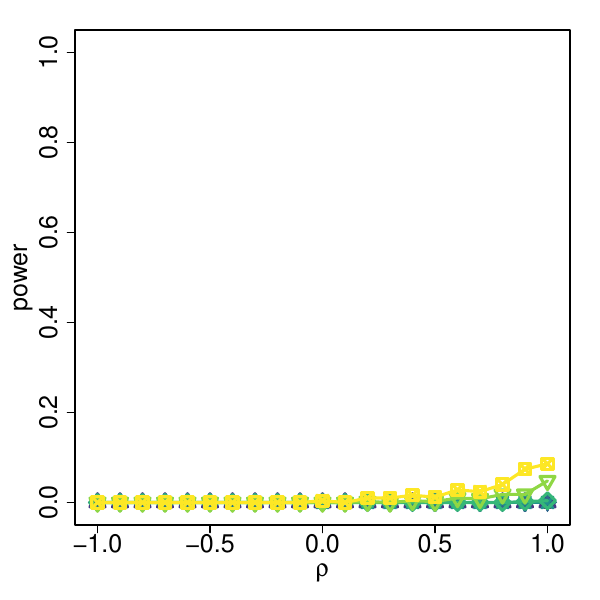}
		\caption{$\tau_0 = -1, \omega = 0.5$}
	\end{subfigure}%
	\begin{subfigure}{.33\textwidth}
		\centering
		\includegraphics[width=1\linewidth]{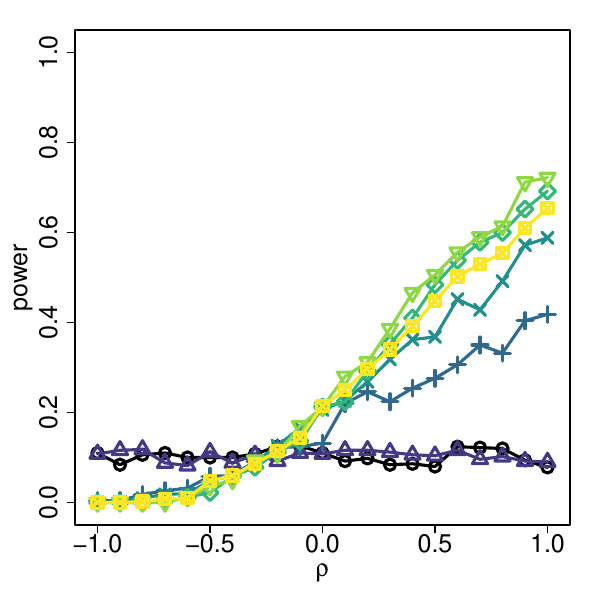}
		\caption{$\tau_0 = 0, \omega = 0.5$}
	\end{subfigure}%
	\begin{subfigure}{.33\textwidth}
		\centering
		\includegraphics[width=1\linewidth]{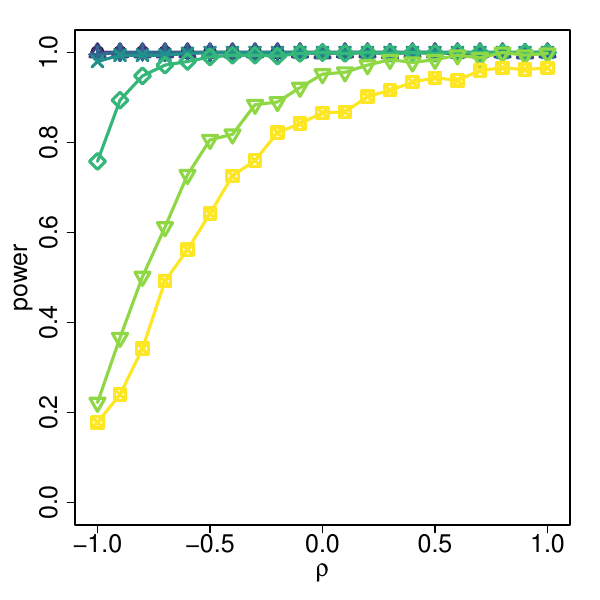}
		\caption{$\tau_0 = 1, \omega = 0.5$}
	\end{subfigure}
	\begin{subfigure}{.33\textwidth}
		\centering
		\includegraphics[width=1\linewidth]{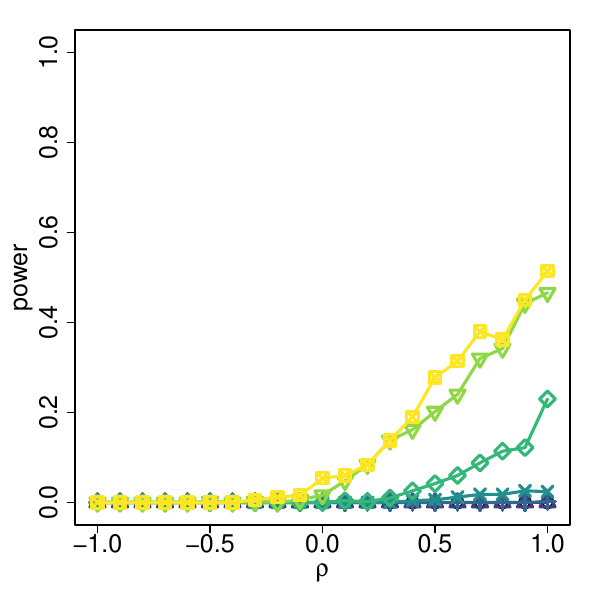}
		\caption{$\tau_0 = -1, \omega = 1$} 
	\end{subfigure}
	\begin{subfigure}{.33\textwidth}
		\centering
		\includegraphics[width=1\linewidth]{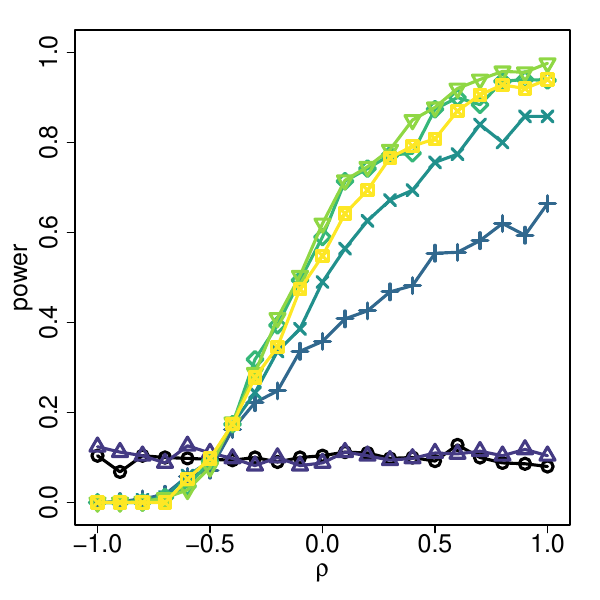}
		\caption{$\tau_0 = 0, \omega = 1$} 
	\end{subfigure}%
	\begin{subfigure}{.33\textwidth}
	\centering
	\includegraphics[width=1\linewidth]{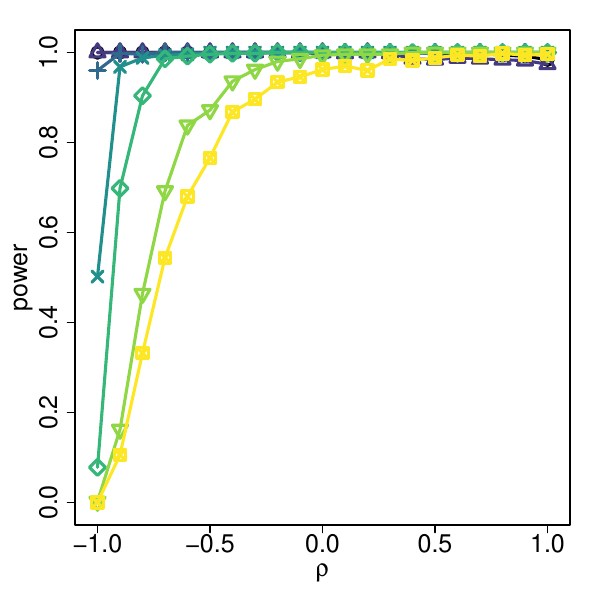}
	\caption{$\tau_0 = 1, \omega = 1$}
	\end{subfigure}
	\caption{Power of randomization tests using $p_{\bm{Z}, \bm{0}}$ with different test statistics for the null hypothesis $H_{\vecle \bm{0}}$ or equivalently $\tau_{\max} \le 0$
    at significance level equals 0.1.
    The potential outcomes are generated from \eqref{eq:generate} with sample size $n=120$ and different values of $(\tau_0, \omega, \rho)$. 
    }\label{fig:max_effect}
\end{figure}

The choice of $s$ for the Stephenson rank sum statistic
is a researcher choice.
As the value of $s$ increases, 
the Stephenson rank places greater weights on larger outcomes.
Therefore, intuitively, the ``optimal'' choice of $s$ will depend on the right tails of the distributions of treatment and control potential outcomes. 
First, from Figures \ref{fig:max_effect}(a), (b), (d) and (e), when the average treatment effect is non-positive in expectation and the individual effects are not very negatively correlated with the control potential outcomes (i.e., $\rho$ is not very small), 
the power of the Stephenson rank sum test generally increases with $s$. 
This is intuitive, since in this case the treated group will tend to have larger outcomes than the control group. 
Second, from Figures \ref{fig:max_effect}(c) and (f), when the average treatment effect is positive, the power of the Stephenson rank test can decrease with $s$, especially for small or negative $\rho$.
This is also intuitive, since in this case the control group is more likely to have larger outcomes, which will reduce the observed Stephenson rank sum statistic, and thus reduce power.
Overall, we suggest 
a moderately large $s$ for randomization tests to infer maximum treatment effects.
As a side note, \citet{Conover88a} examined the asymptotic relative efficiency of a closely related class of test statistics, 
and found that when only a small fraction of treated respond, the optimal subset size $s$ is between 5 and 6.

In sum, although the Wilcoxon rank sum statistic is commonly used in practice and has greatest relative power when treatment effects are close to constant, 
the Stephenson rank sum statistic can be preferred %
due to its sensitivity to extreme treatment effects, which can lead to tighter confidence intervals for the maximum effect when the maximum differs greatly from the mean or median.
It is even possible for a Stephenson rank sum test to reject the bounded null $H_{c\bm{1}}$ for positive values of $c$ when the average treatment effect estimate is negative, if some individual treatment effects are sufficiently positive. 
Thus, when treatment effects are heterogeneous, the behavior of the Stephenson rank sum test can differ markedly from the Wilcoxon rank sum or difference-in-means, while, like them, still providing a valid test for the bounded null hypothesis.
This also means that it can have greater power when using permutation testing in the classic sense of testing whether there is any violation of the sharp null of no treatment effects whatsoever.

\subsection{A simulation study for inferring quantiles of individual effects}\label{sec:simu_visual}

\begin{figure}[htb]
    \centering
		\includegraphics[width=0.8\linewidth]{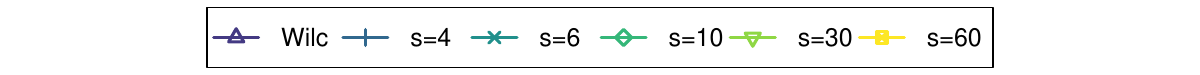}
	\begin{subfigure}{.33\textwidth}
		\centering
		\includegraphics[width=1\linewidth]{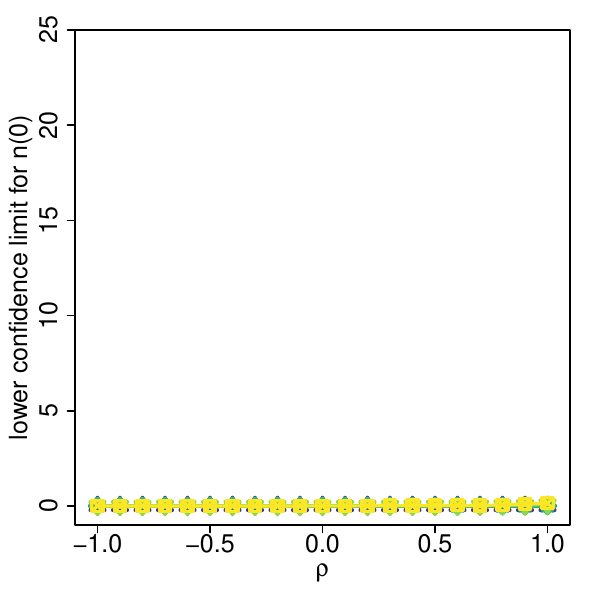}
		\caption{$\tau_0 = -1, \omega = 0.5$}
	\end{subfigure}%
	\begin{subfigure}{.33\textwidth}
		\centering
		\includegraphics[width=1\linewidth]{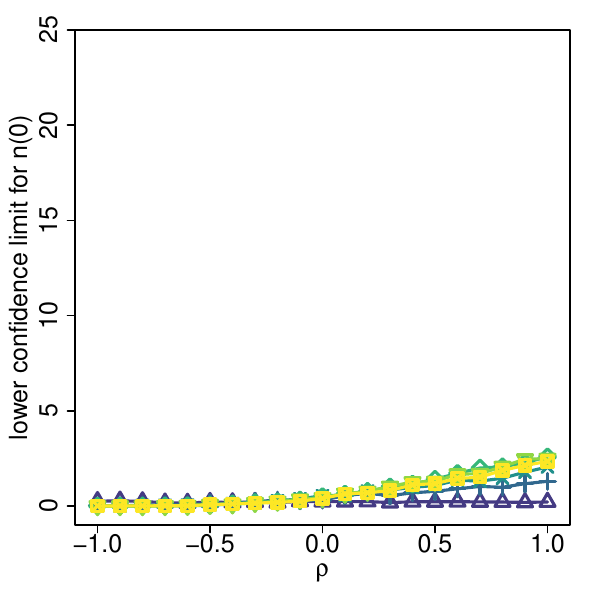}
		\caption{$\tau_0 = 0, \omega = 0.5$}
	\end{subfigure}%
	\begin{subfigure}{.33\textwidth}
		\centering
		\includegraphics[width=1\linewidth]{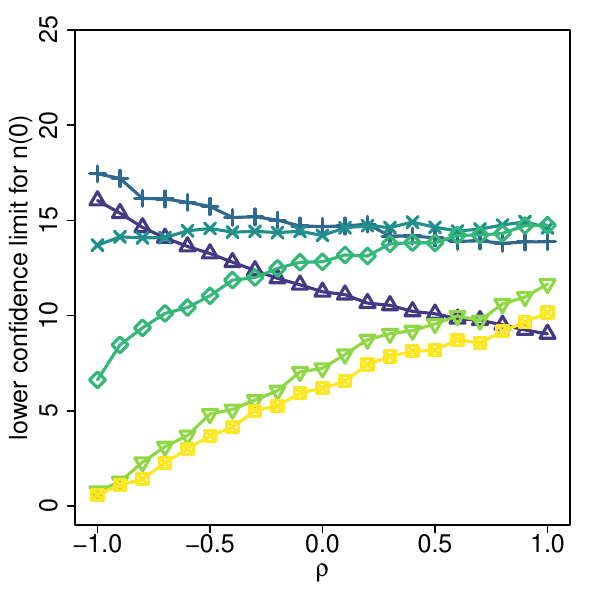}
		\caption{$\tau_0 = 1, \omega = 0.5$}
	\end{subfigure}
	\begin{subfigure}{.33\textwidth}
		\centering
		\includegraphics[width=1\linewidth]{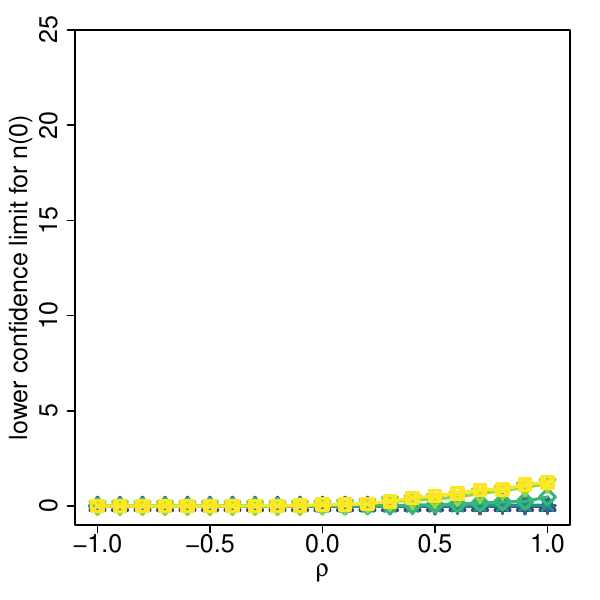}
		\caption{$\tau_0 = -1, \omega = 1$} 
	\end{subfigure}
	\begin{subfigure}{.33\textwidth}
		\centering
		\includegraphics[width=1\linewidth]{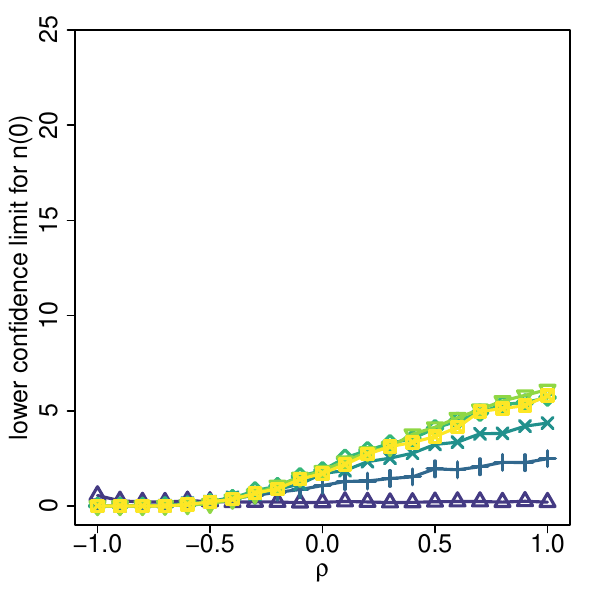}
		\caption{$\tau_0 = 0, \omega = 1$} 
	\end{subfigure}%
	\begin{subfigure}{.33\textwidth}
	\centering
	\includegraphics[width=1\linewidth]{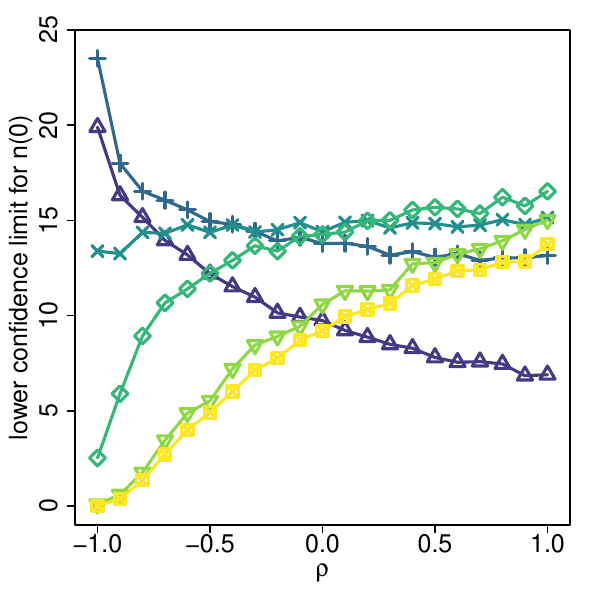}
	\caption{$\tau_0 = 1, \omega = 1$}
	\end{subfigure}
	\caption{Average lower limits of $90\%$ confidence intervals for the number of units with positive effects $n(0)$. 
		The potential outcomes are generated from \eqref{eq:generate} with sample size $n=120$ and different values of $(\tau_0, \omega, \rho)$. 
	}\label{fig:simu_quant}
\end{figure}

We next conduct a simulation study to investigate the power of different rank score statistics for detecting positive quantiles of individual treatment effects. 
We generate the data in the same way as Section \ref{sec:simu_max}, 
and 
focus on the inference of the number of units with effects larger than zero, i.e, $n(0)$ as defined in 
\eqref{eq:n_c}. 
The lower confidence limit of $n(0)$ is 
equivalently 
the number of %
$\tau_{(k)}$'s whose 
confidence intervals do not cover zero.

Figure \ref{fig:simu_quant}
shows the average lower bounds of the $90\%$ confidence intervals for 
$n(0)$ 
using 
the Stephenson rank sum statistics with parameter $s$ ranging from 2 to 60, where $s=2$ corresponds to the Wilcoxon rank sum statistic. 
From Figures \ref{fig:simu_quant}(a), (b), (d) and (e), 
when the average treatment effect is 
less than or equal to zero, 
the Wilcoxon rank sum statistic has almost no power to provide informative confidence intervals for the number of units with positive effects. 
However, the Stephenson rank sum statistics are able to detect significant numbers of positive treatment effects, 
where larger $s$ tends to give larger lower confidence limits. 
From Figures \ref{fig:simu_quant}(c) and (f), 
when the average treatment effect is positive, 
the power of the Stephenson rank sum statistic 
becomes non-monotone in $s$. 
In particular, very large value of $s$ can lead to deteriorated confidence limits for $n(0)$, especially when 
the individual treatment effect is negatively correlated with the control potential outcome. 
The intuition is similar to that discussed in Section \ref{sec:simu_max}: 
the control group is likely to have larger outcomes and the Stephenson rank with large $s$ places greater weights on these larger outcomes, making the test less powerful to detect positive individual effects.

Intuitively, $s$ values that are too high may lose power due to overly relying on few units, and $s$ values that are too low  may not be able to
take advantage of the tail behavior of larger effects. 
Moreover, 
additional simulation results in Appendix A6.1 of the supplementary materials show that different Stephenson rank statistics can be preferred for different estimands of interest, even under the same data generating process.

In general, selecting $s$ is a decision the researcher has to make when planning their data analysis.
To aid with this, we provide functions in our developed package that, under a user-specified data generating model, compare the performance of various Stephenson rank sum statistics in terms of either power for testing the null hypotheses for given quantiles of individual effects or the median magnitude of confidence limits for the number of units with effects passing any given threshold. 
Users can then specify a range of possible distributions of control-side potential outcomes, treatment impact models, and potential correlations of the two, along with their primary estimand of interest (such as $\tau_{(k)}$ for some $k$ or $n(c)$ for some $c$), and compare the performances of different $s$ values for these hypothetical scenarios at a given sample size $n$.
They would then select $s$ based on which value generally had superior performance for the targeted estimands of interest across the explored scenarios.
The distributions of outcomes and effects, and their correlation, would be obtained from (historical) empirical data and prior knowledge.

In our applications, we mainly consider the Stephenson rank sum statistic with $s=6$ or $10$,
based on this process and the other simulations in this paper.
Generally $s=6$ appears to be a versatile initial choice in the absence of empirical information.

\section{Evaluating the effectiveness of professional development}\label{sec:applications}
\label{sec:profession}

\citet{HSMHSD10} studied the effectiveness of professional development on elementary teachers, classrooms and students using a randomized experiment conducted at eight national research sites. 
A sample of fourth grade teachers were randomly assigned to treatment and control, where treated teachers would participate in a professional development course encompassing eight three-hour sessions focusing on the teaching of electric circuits. We are interested in the effect of professional development on the teachers' electric circuits content knowledge, as measured by the gain scores based on tests before and after the professional development courses. 
The actual experiment was randomized within site and school, and the active treatment had three versions with regard to additional activities for the development of pedagogical content knowledge. 
For simplicity and to illustrate our approach, we analyze it as a completely randomized treatment-control experiment and exclude teachers with missing outcomes\footnote{Since the outcomes of all control teachers are missing in one of the eight sites, we exclude that site in our analysis.}, resulting in 164 treated teachers and 69 control teachers. 
Figure \ref{fig:conf_profession}(a) shows the histograms 
of the observed gain scores in the treatment and control groups.
The treated teachers tend to  have larger gain scores, and the corresponding histogram is close to a positive shift of that for the control teachers, with a magnitude of 15 to 20. Therefore, intuitively, we expect our approach to infer a significant proportion of teachers with positive effects as there is little sign of treatment effect heterogeneity and a large share of treated teachers' outcomes are larger than nearly all control teacher outcomes.

Before proceeding with our analysis, we need to select the tuning parameter $s$ for our Stephenson statistic.
We do this via a power simulation, generating a series of datasets with an empirical control-side distribution bootstrapped from the control units in our data, the average treatment effect as estimated from our data, different hypothesized distributions of treatment effects (constant, normal, exponential in shape) with different levels of impact variation, and different correlations of treatment impact with baseline outcome.
We calibrate an assumed degree of treatment variation by assuming zero correlation and comparing the variances of the treated and control groups, but explore other values as well.
We additionally explore both positive and negative correlations as a sensitivity check.
For each dataset we calculate simultaneous confidence intervals for the largest 100 individual effects, using the method described in Section \ref{sec:inf_quant_all}.
For each of these top 100, we then calculate the median lower bound of their confidence intervals along with power (the probability of the CI excluding zero) across simulation runs.

Overall, we found the lower confidence bounds were fairly consistent for $s$ between around 3 and 8 across scenarios, indicating some degree of robustness against choice of $s$, with better performance for lower $s$ for uncorrelated and negatively correlated effects.
We also examined the performance on the number of significant units by examining the lower confidence bound on $n(0)$ (see Section~\ref{sec:simultaneous_equivalence}) across runs; here, $s = 8$ generally gave the largest number of significant units.
Considering that a small amount of variation is consistent with the marginal distributions of the treatment and control groups, we select $s = 6$ for our analysis.
Appendix A7.1 of the supplementary materials and the 
replication file both give further details of the above procedure and aggregation, along with plots showing how power changes as a function of $s$ under different scenarios.

With our selected $s=6$ we analyze the original data, calculating confidence intervals for all quantiles.
Figure \ref{fig:conf_profession}(b), from our motivating example at the beginning of our paper, shows the $90\%$ lower confidence limits for all the $\tau_{(k)}$'s using the Stephenson rank sum statistic with $s=6$ that have finite lower limits (the lower confidence limits for $\tau_{(k)}$'s with $k \le 116$ are all negative infinity).
From Figure \ref{fig:conf_profession}(b), the lower confidence limits of $\tau_{(k)}$'s with $146 \le k \le 233$ are all larger than zero, implying that a $90\%$ confidence interval for $n(0)$ is $[88, 233]$. 
Equivalently, we are $90\%$ confident that at least $88/233 = 37.8\%$ units would 
benefit from the professional development courses. 
Similarly, a $90\%$ confidence interval for $n(6)$ is $[69, 233]$. That is, we are $90\%$ confident that at least $69/233 = 29.6\%$ teachers would have gained six more points in the test if they had participated in the professional development.

We repeated the above using the usual Wilcoxon rank sum statistic.
This gave noninformative limits for all $\tau_{(k)}$'s with $k \le 159$, rather than $116$.
The lower confidence limits for $n(0)$ and $n(6)$ are respectively $59$ and $48$, corresponding to $25.3\%$ and $20.6\%$ of the teachers. 
Obviously, the inference results using Stephenson ranks are far more informative.

{\rev Finally, we apply Theorem \ref{thm:effect_range} to study the effect range. 
The confidence intervals for the minimum and maximum effects overlap substantially, yielding a completely uninformative confidence interval for the effect range and precluding rejection of the hypothesis of a constant effect.
This is consistent with the graphical evidence in Figure \ref{fig:conf_profession}(a), where the distribution of the observed outcome for treated units is close to a shift of that for control units, suggesting a constant treatment effect is plausible.}

\section{Conclusion and Discussion}\label{sec:discuss}

The rise of nonparametric causal inference in the Neyman-Rubin tradition, with its emphasis on average effects and on effect heterogeneity as the rule rather than the exception, has rightly prompted greater skepticism of statistical methods that rely on parametric assumptions. It is perhaps no surprise that this skepticism has also extended to RI, which despite its freedom from other distributional assumptions has traditionally been motivated in terms of shift hypotheses or other highly structured models of treatment effects \cite[e.g.,][]{Lehmann63a, Rosenbaum02a}.

We have argued that the view of RI prevalent among statisticians and applied researchers
---that RI is useful only for assessing the typically uninteresting sharp hypothesis that treatment had no effect at all---is too limited. We have proved that %
randomization tests can be valid under a more general bounded null, that this fact can be exploited to derive confidence intervals for the maximum or minimum effect, and that many familiar test statistics can lead to tests with this property. 
We then extended the RI for the maximum (or minimum) effect to general quantiles of the individual effects, which in turn provides confidence intervals for the number (or equivalently proportion) of units with effects larger than (or smaller than) any threshold. 
Moreover, the confidence intervals for all quantiles of individual effects are simultaneously valid.

Mainly, we hope our novel perspective on RI tempers the skepticism that many applied scientists hold towards this otherwise-appealing mode of statistical inference.
RI is by no means a cure-all; nor is it a substitute for average treatment effect estimation when that is the goal of the analysis. But in many cases, particularly when sample size is small or outcomes are heavy-tailed, 
it is the most reliable form of statistical inference. Even when this is not the case, it can yield unique insight into the pattern of treatment effects across the whole sample. For these reasons, RI deserves a secure place in the applied statistician's toolbox.

\section*{Supplementary Materials}

The supplementary materials contain two files. 
The first file includes 
(i) further discussion of primary theoretical results, 
(ii) a discussion on the validity of randomization tests using test statistics of form $t(\bm{Z}, \bm{Y})$ for bounded null hypotheses,  
(iii) the proofs of all theorems, corollaries and propositions, 
(iv) further simulation studies, 
and (v) further applications.
The second file includes (i) installation of the developed R package, 
(ii) replication for the three data analyses in the paper and the supplementary materials, 
and (iii) illustration for the R functions to compare the power of different Stephenson rank sum statistics. 
The R package \textsf{RIQITE} implementing the proposed methods is available at \url{https://github.com/li-xinran/RIQITE}.

\section*{Acknowledgments}

We thank the Editor, the Associate Editor and 
all the
reviewers for insightful and constructive comments. 
For helpful input we also thank Peter Aronow, Jake Bowers, Joanna Dafoe, 
Peng Ding, 
Danny Hidalgo, Greg Huber, Kosuke Imai, Luke Keele, Kelly Rader, 
Paul Rosenbaum,
Brandon Stewart, and seminar participants at PolMeth 2016.
We also thank the late Natasha Chichilnisky-Heal for posing the question that stimulated some of our thinking on this paper.

\section*{Data availability}
All data used are either publicly available or available within the supplementary materials. 
The dataset in Section \ref{sec:applications} from \citet{HSMHSD10} is available in the R package \textsf{RIQITE} at \url{https://github.com/li-xinran/RIQITE}. 
The dataset in Appendix A7.2 is generated from \citet{BlackEtAl11a} with details in the supplementary materials.
The dataset in Appendix 7.3 is available at Harvard Dataverse with link \url{https://doi.org/10.7910/DVN/38X3LX}.

\bibliographystyle{plainnat}
\bibliography{permutation}

\newpage

\setcounter{equation}{0}
\setcounter{section}{0}
\setcounter{figure}{0}
\setcounter{example}{0}
\setcounter{proposition}{0}
\setcounter{corollary}{0}
\setcounter{theorem}{0}
\setcounter{lemma}{0}
\setcounter{table}{0}
\setcounter{condition}{0}
\setcounter{page}{1}
\begin{center}
	\bf \LARGE 
	Supplementary Material 
\end{center}

\renewcommand {\theproposition} {A\arabic{proposition}}
\renewcommand {\theexample} {A\arabic{example}}
\renewcommand {\thefigure} {A\arabic{figure}}
\renewcommand {\thetable} {A\arabic{table}}
\renewcommand {\theequation} {A\arabic{section}.\arabic{equation}}
\renewcommand {\thelemma} {A\arabic{lemma}}
\renewcommand {\thesection} {A\arabic{section}}
\renewcommand {\thetheorem} {A\arabic{theorem}}
\renewcommand {\thecorollary} {A\arabic{corollary}}
\renewcommand {\thecondition} {A\arabic{condition}}

\renewcommand {\thepage} {A\arabic{page}}

\section{Further discussion of primary theoretical results}

\subsection{\lxr Extended discussion of the tie-breaking methods}

{\lxr We} first
discuss 
the subtle issue of defining ranks when there exist ties. 
There are multiple ways to define ranks for ties, see, e.g., the R documentation for the function \textsf{rank} \citep{Rteam}. 
Consider an arbitrary outcome vector $\bm{y} \in \mathbb{R}^n$. 
The method ``random'' puts equal values of $y_i$'s in a random order, 
the methods ``first'' and ``last'' rank equal values of $y_i$'s based on their indices in an increasing {\rev or} decreasing way\footnote{For example, if we rank the coordinates of $\bm{y}$ using the ``first'' method, 
then 
$\rank_i (\bm{y}) < \rank_j(\bm{y})$ if and only if (a) $y_i < y_j$ or (b) $y_i = y_j$ and $i < j$.}, 
and the method ``average'' replaces their ranks by the corresponding average. 
The first three methods will always produce ranks from $1$ to $n$, which is important for the distribution free property in Definition \ref{def:dist_free}, while the ranks from %
the last 
generally depend on the value of $\bm{y}$.

{\lxr
We then give a more detailed discussion of Proposition \ref{prop:stat_property_rank}(b). 
The statistic $t_2(\cdot, \cdot)$ in \eqref{eq:stat_class_rank} is distribution free under the ``random'' method for ties.}
In practice, we can implement the ``random'' method by first randomly permuting the coordinates of $\bm{y}$ and then using the ``first'' {\lxr or ``last'' method}.
We emphasize that using the ``random'' method for ties is crucial.  
First, a defined method for ties is important even for continuous outcomes, because some outcomes will become the same when we later invert tests for, e.g., a sequence of constant treatment effects. 
Second, it is important to rank equal outcomes randomly, since in practice 
a received data set
may be ordered, e.g., 
with all treated units ahead of control ones. 

{\lxr We finally give some discussion regarding Proposition \ref{prop:stat_property_rank}(a).}
For the usual ``average'' method for ties, whether the statistic in \eqref{eq:stat_class_rank} is effect increasing depends on the values of $\phi(\cdot)$ for tied ranks. If we define the value of $\phi(\cdot)$ for tied ranks as the average value of $\phi(\cdot)$ evaluated at those ranks under the ``first'' or ``last'' method, then the statistic \eqref{eq:stat_class_rank} is effect increasing. This follows directly from Proposition \ref{prop:stat_property_rank} by noting that the resulting value of the statistic is essentially the average of it under the ``first'' method of ties over all possible permutations of the ordering.
From the above, we can also know that the classical Wilcoxon rank sum statistic with ``average'' method for ties is effect increasing.

\subsection{Discussion of Theorem {\lxr \ref{thm:null_broader_monotone_control}} in Section \ref{sec:valid_bound}}

Theorem \ref{thm:null_broader_monotone_control}(b) presents a stronger conclusion than (a), but with stronger condition on the test statistic. 
Specifically, 
for any $\bm{\delta} \in \mathbb{R}^n$, 
if the bounded null $H_{\vecle\bm{\delta}}$ holds (i.e., $\bm{\tau} \vecle \bm{\delta}$), 
then Theorem \ref{thm:null_broader_monotone_control}(b) implies that 
the randomization $p$-value
$
p_{\bm{Z}, \bm{\delta}}
\ge 
p_{\bm{Z}, \bm{\tau}}
$
and is thus 
stochastically larger than or equal to $\Unif[0,1]$. 
More importantly, Theorem \ref{thm:null_broader_monotone_control}(b) is useful for constructing confidence sets for the true treatment effect $\bm{\tau}${\lxr, as discussed in the main paper}.
In principle, 
one could invert randomization tests for all possible sharp null hypotheses $H_{\bm{\delta}}$'s to get confidence sets for the true treatment effect $\bm{\tau}$. 
However, enumerating all possible sharp null hypotheses is generally computationally intractable, except in the cases of binary or discrete outcomes \citep[see][]{Rigdon:2015aa}. 
For discrete outcomes where enumeration is possible, 
Theorem \ref{thm:null_broader_monotone_control}(b) can 
help reduce the number of enumerations \citep{Li:2016tw}.

\subsection{Discussion of Theorems {\lxr \ref{thm:con_interval_nc} and \ref{thm:simultaneous_valid} and Corollary \ref{cor:infer_number_larger_threshold}} in 
Section{\lxr s
\ref{sec:ci_qi}--\ref{sec:simultaneous_equivalence}
}
}

    Theorem \ref{thm:con_interval_nc} and Corollary \ref{cor:infer_number_larger_threshold} provide confidence intervals for quantiles of individual effects $\tau_{(k)}$'s as well as number (or equivalently proportion) of units with effects larger than any threshold. 
    However, both theorems do not provide point estimation for these quantities. 
    Indeed, these quantities are generally not identifiable due to no joint observation of the treatment and control potential outcomes for any unit. 
    Consequently, consistent estimators for them generally do not exist. 
    Recently, for binary or ordinal outcomes, 
    \citet{lu2018} and \citet{Rosenblum2019} studied sharp bounds and constructed confidence intervals for the proportions of units with positive effects, i.e, $n(0)/n$. 
    Importantly, our confidence intervals in Corollary \ref{cor:infer_number_larger_threshold} work for general outcomes. 

Furthermore, we could, in principle,  construct confidence sets for the $n$-dimensional individual effect vector $\bm{\tau}$ by inverting tests for all sharp null hypotheses. 
In general, however, testing arbitrary sharp null hypotheses is computationally infeasible, and it does not provide informative inferences because the parameter space is typically too unwieldy (with $n$ units, the space of possible effects is $n$-dimensional). 
As noted by \citet{Rosenbaum10a}, such a confidence set would not be intelligible, since it would be a subset of an $n$-dimensional space. 
Some special forms of outcomes or test statistics can lead to efficient computation and intuitive confidence sets. 
For example, 
\citet{Rosenbaum01a} used carefully designed test statistics involving the attributable effects, 
and \citet{Rigdon:2015aa} considered binary outcomes. 
A key property utilized by these approaches is that many sharp null hypotheses are equally likely in the sense of producing the same randomization $p$-value, 
which not only avoids enumeration over all possible values of $\bm{\tau}$ but also provides a convenient form for the resulting confidence sets. 
By contrast, 
our confidence sets constructed in Theorem \ref{thm:simultaneous_valid} work for general outcomes and rely mainly on the valid $p$-value \eqref{eq:p_kc} for testing null hypotheses about quantiles of individual effects, 
which involves efficient optimization of randomization $p$-value in Theorem \ref{thm:test_Hkc}. 
Moreover, as illustrated in {\rev Section \ref{sec:moti_example}}, 
Theorem \ref{thm:simultaneous_valid} provides nice confidence sets in $\mathbb{R}^n$ that are easy to understand, interpret and visualize.

{\rev 
In fact, our confidence set in Theorem \ref{thm:simultaneous_valid} essentially provides a confidence band for the quantile (or equivalently distribution) function of the individual treatment effects. 
Relatedly, \citet{LeiCandes2021} constructed prediction intervals for the (random) individual treatment effect, assuming random sampling of units from some superpopulation. 
One main difference is that we focus on the fixed population distribution of individual effects, while they focus on a random draw from the superpopulation distribution of individual effects.}

\section{Broader justification for \citet{Imbens15a} style Fisher randomization tests}

    First, as mentioned in the main paper, it should not be surprising that we can also use test statistics of the form $t(\bm{Z}, \bm{Y}_{\bm{Z}, \bm{\delta}}(1))$, which involves the imputed treatment potential outcomes instead of imputed control ones. 
	This can be achieved by 
    switching the labels of treatment and control and changing the signs of the outcomes. 

More broadly, the Fisher randomization tests 
    in \citet{Imbens15a} use test statistics of form $t(\bm{Z}, \bm{Y})$, which often compare the observed outcomes of treated units to those of control units (e.g., the difference in outcome means between treatment and control groups). 
    Analogous to the main paper, 
    such a randomization test can also be valid for testing bounded null hypotheses as we discuss below.
    Unfortunately, our generalization for inferring quantiles of individual effects (as in Section \ref{sec:inf_quant_all}) relies crucially on the use of test statistics of form $t(\bm{Z}, \bm{Y}_{\bm{Z}, \bm{\delta}}(0))$; our results do not extend to this class of statistic.
    This is why we focus on the randomization $p$-value of form \eqref{eq:p_val_imp_control} in the main paper. 

That being said, we next provide the broader justification for the Fisher randomization test
in \citet{Imbens15a}, 
which differs from that in \citet{Rosenbaum02a} in the choice of test statistic. 
We first formally describe the randomization $p$-value in \citet{Imbens15a}, and then study its property in parallel with that in Section \ref{sec:broader}.

\subsection{The randomization $p$-value in \citet{Imbens15a}}

The Fisher randomization test in \citet{Imbens15a} considers 
test statistics of form $t(\bm{Z}, \bm{Y} )$, which often compare the observed outcomes of treated units to those of control units. 
For the null $H_{\bm{\delta}}$ in \eqref{eq:null_delta}, the imputed potential outcomes are $\bm{Y}_{\bm{Z}, \bm{\delta}}(1)$ and $\bm{Y}_{\bm{Z}, \bm{\delta}}(0)$ in \eqref{eq:impute_potential_outcome}. 
For any treatment assignment vector $\bm{a} \in \mathcal{Z}$, 
the corresponding imputed observed outcome vector would then be 
\begin{align*}
	\bm{Y}_{\bm{Z}, \bm{\delta}}(\bm{a}) \equiv \bm{a} \circ \bm{Y}_{\bm{Z}, \bm{\delta}}(1) + (1-\bm{a}) \circ \bm{Y}_{\bm{Z},\bm{\delta}}(0), 
\end{align*}
and the corresponding test statistic would have value $t(\bm{a}, \bm{Y}_{\bm{Z},\bm{\delta}}(\bm{a}))$. 
Thus, 
for the null $H_{\bm{\delta}}$, 
the imputed randomization distribution of the test statistic has the following tail probability: 
\begin{align}\label{eq:tail_prob}
	\tilde{G}_{\bm{Z}, \bm{\delta}} (c) \equiv 
	\Pr\left\{
	t(\bm{A}, \bm{Y}_{\bm{Z},\bm{\delta}}(\bm{A}))
	\ge 
	c
	\right\}
	= 
	\sum_{\bm{a} \in \mathcal{Z}} \Pr(\bm{A} = \bm{a}) \I\left\{
	t(\bm{a}, \bm{Y}_{\bm{Z},\bm{\delta}}(\bm{a}))
	\ge 
	c
	\right\}, \qquad (c\in \mathbb{R})
\end{align}
and the corresponding randomization $p$-value is the tail probability \eqref{eq:tail_prob} evaluated at the observed value of the test statistic: 
\begin{align}\label{eq:p_val_imp_both}
	\tilde{p}_{\bm{Z}, \bm{\delta}} \equiv \tilde{G}_{\bm{Z}, \bm{\delta}}(t(\bm{Z}, \bm{Y})) 
	= 
	\sum_{\bm{a} \in \mathcal{Z}} \Pr(\bm{A} = \bm{a}) \I\left\{
	t(\bm{a}, \bm{Y}_{\bm{Z},\bm{\delta}}(\bm{a}))
	\ge 
	t(\bm{Z}, \bm{Y})
	\right\}. 
\end{align}
When %
$H_{\bm{\delta}}$ is true, the imputed randomization distribution $G_{\bm{Z}, \bm{\delta}}(\cdot)$ in \eqref{eq:tail_prob} is the same as the true one, and the randomization $p$-value $\tilde{p}_{\bm{Z}, \bm{\delta}}$ in \eqref{eq:p_val_imp_both} is stochastically larger than or equal to $\Unif[0,1]$,  i.e., it is a valid $p$-value for testing $H_{\bm{\delta}}$.

The two $p$-values $\tilde{p}_{\bm{Z}, \bm{\delta}}$ in \eqref{eq:p_val_imp_both} and $p_{\bm{Z}, \bm{\delta}}$ in \eqref{eq:p_val_imp_control} are equivalent for testing the Fisher's null $H_{\bm{0}}$ of no effect. 
As commented in 
Appendix \ref{sec:comment_p_value}, 
they are also equivalent for testing the null $H_{c\bm{1}}$ of constant effect $c$ for any $c\in \mathbb{R}$,  if  both of them use difference-in-means as the test statistic.
However, they are generally not equivalent. 
Two  obvious differences are as follows. 
First, given the observed data $(\bm{Z}, \bm{Y})$, $\tilde{p}_{\bm{Z}, \bm{\delta}}$ depends on the speculation of individual effects for all units, 
while  $p_{\bm{Z}, \bm{\delta}}$ depends only on the speculation of individual effects for treated units or equivalently $\bm{Z}\circ \bm{\delta}$. 
Second, 
for $\tilde{p}_{\bm{Z}, \bm{\delta}}$,
the tail probability $\tilde{G}_{\bm{Z}, \bm{\delta}}(\cdot)$ in \eqref{eq:tail_prob} depends on the null of interest $\bm{\delta}$, but the cutoff $t(\bm{Z}, \bm{Y})$ does not depend on $\bm{\delta}$; 
while for $p_{\bm{Z}, \bm{\delta}}$, 
both the tail probability $G_{\bm{Z}, \bm{\delta}}(\cdot)$ in \eqref{eq:tail_prob_tilde} and the cutoff $t(\bm{Z}, \bm{Y}_{\bm{Z}, \bm{\delta}}(0))$ depend on $\bm{\delta}$.

\subsection{Validity of the randomization $p$-value for testing bounded nulls}

Below we study the randomization $\tilde{p}_{\bm{Z}, \bm{\delta}}$ in \eqref{eq:p_val_imp_both} for testing bounded nulls, its monotonicity property and its usage for constructing confidence intervals for the maximum individual effect. 
We summarize the results in the following 
theorem and corollary, 
in parallel with 
Theorem \ref{thm:null_broader_monotone_control} and Corollary \ref{cor:ci_max_imp_control}
in the main paper for the randomization $p$-value $p_{\bm{Z}, \bm{\delta}}$ in \eqref{eq:p_val_imp_control}.

\begin{theorem}\label{thm:null_broader_monotone_both}
    If the statistic $t(\cdot, \cdot)$ is effect increasing,
	then 
	\begin{itemize}
	    \item[(a)] for any constant $\bm{\delta} \in \mathbb{R}^n$, 
	the corresponding randomization $p$-value $\tilde{p}_{\bm{Z}, \bm{\delta}}$ in \eqref{eq:p_val_imp_both} for the sharp null $H_{\bm{\delta}}$ in \eqref{eq:null_delta} is also valid for testing the bounded null $H_{\vecle\bm{\delta}}$ in \eqref{eq:null_less_delta}.  
	Specifically, under $H_{\vecle\bm{\delta}}$,  
	$
	\Pr(\tilde{p}_{\bm{Z}, \bm{\delta}} \le \alpha ) \le \alpha
	$
	for any $\alpha \in (0,1)$; 
	
	    \item[(b)] for any possible assignment $\bm{z} \in \mathcal{Z}$, 
	the corresponding randomization $p$-value $\tilde{p}_{\bm{z}, \bm{\delta}}$ in \eqref{eq:p_val_imp_both}, 
	viewed as a function of $\bm{\delta}\in \mathbb{R}^n$, is monotone %
	increasing. 
	Specifically,  
	$
	\tilde{p}_{\bm{z}, \bm{\delta}} \le \tilde{p}_{\bm{z}, \overline{\bm{\delta}}} 
	$
	for any $\bm{\delta} \vecle \overline{\bm{\delta}}$. 
	\end{itemize}
\end{theorem}

\begin{corollary}\label{cor:ci_max_imp_both}
	(a) If the test statistic $t(\cdot, \cdot)$ is effect increasing, 
	then for any $\alpha \in (0,1)$, 
	the set 
	$
	\{c: \tilde{p}_{\bm{Z}, c\bm{1}} > \alpha, c \in \mathbb{R} \}
	$
	is a $1-\alpha$ confidence set for the maximum individual effect $\tau_{\max}$. 
	(b) Furthermore, 
	the confidence set must have the form of $(\underline{c}, \infty)$ or $[\underline{c}, \infty)$ with $\underline{c} = \inf \{c: \tilde{p}_{\bm{Z}, c\bm{1}} > \alpha, c \in \mathbb{R} \}$. 
\end{corollary}

\section{Proofs for properties of test statistics}\label{sec:proof_property}

\begin{proof}[\bf Proof of Proposition \ref{prop:stat_property_dim}]
Suppose that both $\psi_{1i}(\cdot)$ and $\psi_{0i}(\cdot)$ are %
	monotone increasing 
	functions for all $1\le i \le n$. 
	We first show that the statistic $t_1(\cdot, \cdot)$ in \eqref{eq:stat_class_y} is effect increasing. By definition, for any $\bm{z} \in \mathcal{Z}$ and $\bm{y}, \bm{\eta}, \bm{\xi} \in \mathbb{R}^n$ with $\bm{\eta} \vecge \bm{0} \vecge \bm{\xi}$, 
	\begin{align}\label{eq:diff_t1_eff_incre}
	& \quad \ t_1(\bm{z}, \bm{y} + \bm{z} \circ \bm{\eta} + (\bm{1}-\bm{z}) \circ \bm{\xi} ) - t_1(\bm{z}, \bm{y}) 
	\nonumber
	\\
	& = 
	\left\{ \sum_{i=1}^n z_i \psi_{1i}(y_i + \eta_i ) - \sum_{i=1}^n (1-z_i)\psi_{0i}(y_i + \xi_i) \right\}
	- 
	\left\{ \sum_{i=1}^n z_i \psi_{1i}(y_i) - \sum_{i=1}^n (1-z_i)\psi_{0i}(y_i) \right\}
	\nonumber
	\\
	& = 
	\sum_{i=1}^n z_i \left\{ \psi_{1i}(y_i + \eta_i ) - \psi_{1i}(y_i) \right\} + 
	\sum_{i=1}^n (1-z_i) \left\{ \psi_{0i}(y_i) - \psi_{0i}(y_i + \xi_i)   \right\}. 
	\end{align}
	Because $\eta_i \ge 0 \ge \xi_i$ and $\psi_{1i}(\cdot)$ and $\psi_{1i}(\cdot)$ are %
	monotone increasing 
	functions for all $i$, we can know that all the terms in \eqref{eq:diff_t1_eff_incre} are non-negative. Therefore, \eqref{eq:diff_t1_eff_incre} must also be non-negative. From Definition \ref{def:effect_increase}, $t_1(\cdot, \cdot)$ in \eqref{eq:stat_class_y} is effect increasing. 
	
	We then show that the statistic $t_1(\cdot, \cdot)$ in \eqref{eq:stat_class_y} is differential increasing. 
	By definition, for any $\bm{z}, \bm{a} \in \mathcal{Z}$ and $\bm{y}, \bm{\eta} \in \mathbb{R}^n$ with $\bm{\eta} \vecge 0$, 
	\begin{align}\label{eq:diff_t1_dif_incre}
	& \quad \ t_1(\bm{z}, \bm{y} + \bm{a} \circ \bm{\eta}  ) - t_1(\bm{z}, \bm{y}) 
	\nonumber
	\\
	& = 
	\left\{ \sum_{i=1}^n z_i \psi_{1i}(y_i + a_i\eta_i ) - \sum_{i=1}^n (1-z_i)\psi_{0i}(y_i + a_i \eta_i ) \right\}
	- 	\left\{ \sum_{i=1}^n z_i \psi_{1i}(y_i) - \sum_{i=1}^n (1-z_i)\psi_{0i}(y_i) \right\}
	\nonumber
	\\
	& = 
	\sum_{i=1}^n z_i \left\{ \psi_{1i}(y_i + a_i\eta_i ) - \psi_{1i}(y_i) + \psi_{0i}(y_i + a_i \eta_i ) - \psi_{0i}(y_i)  \right\} 
	- 
	\sum_{i=1}^n 
	\left\{
	\psi_{0i}(y_i + a_i \eta_i ) - \psi_{0i}(y_i)
	\right\}. 
	\end{align}
	For each $1\le i \le n$, because $a_i \in \{0,1\}$, $\eta_i \ge 0$, and both $\psi_{1i}(\cdot)$ and $\psi_{1i}(\cdot)$ are %
	increasing 
	functions, we can know that 
	\begin{align*}
	& \quad \ z_i \left\{ \psi_{1i}(y_i + a_i\eta_i ) - \psi_{1i}(y_i) + \psi_{0i}(y_i + a_i \eta_i ) - \psi_{0i}(y_i)  \right\}\\
	& 
	\begin{cases}
	\le a_i \left\{ \psi_{1i}(y_i + a_i\eta_i ) - \psi_{1i}(y_i) + \psi_{0i}(y_i + a_i \eta_i ) - \psi_{0i}(y_i)  \right\}, & \text{if } a_i = 1, \\
	= 0 = a_i \left\{ \psi_{1i}(y_i + a_i\eta_i ) - \psi_{1i}(y_i) + \psi_{0i}(y_i + a_i \eta_i ) - \psi_{0i}(y_i)  \right\}, & \text{if } a_i = 0.  
	\end{cases}
	\end{align*}
	This immediately implies that \eqref{eq:diff_t1_dif_incre} can be bounded above by
	\begin{align*}
	& \quad \ t_1(\bm{z}, \bm{y} + \bm{a} \circ \bm{\eta}  ) - t_1(\bm{z}, \bm{y}) 
	\\
	& = 
	\sum_{i=1}^n z_i \left\{ \psi_{1i}(y_i + a_i\eta_i ) - \psi_{1i}(y_i) + \psi_{0i}(y_i + a_i \eta_i ) - \psi_{0i}(y_i)  \right\} 
	- 
	\sum_{i=1}^n 
	\left\{
	\psi_{0i}(y_i + a_i \eta_i ) - \psi_{0i}(y_i)
	\right\}\\
	& 
	\le 
	\sum_{i=1}^n a_i \left\{ \psi_{1i}(y_i + a_i\eta_i ) - \psi_{1i}(y_i) + \psi_{0i}(y_i + a_i \eta_i ) - \psi_{0i}(y_i)  \right\} 
	- 
	\sum_{i=1}^n 
	\left\{
	\psi_{0i}(y_i + a_i \eta_i ) - \psi_{0i}(y_i)
	\right\}\\
	& = t_1(\bm{a}, \bm{y} + \bm{a} \circ \bm{\eta}  ) - t_1(\bm{a}, \bm{y}). 
	\end{align*}
	From Definition \ref{def:diff_increase}, $t_1(\cdot, \cdot)$ in \eqref{eq:stat_class_y} is differential increasing. 
	
	From the above, Proposition \ref{prop:stat_property_dim} holds. 
\end{proof}

To prove Proposition \ref{prop:stat_property_rank}, we need  the  following three lemmas.
\begin{lemma}\label{lemma:rank_to_delta}
	For any $n\ge 1$, $\bm{y} \in \mathbb{R}^n$, 
	for any $1\le i, j \le n$, define
	\begin{align}\label{eq:delta}
	\delta_{ij}(y_i, y_j) & = 
	\begin{cases}
	\I(y_i > y_j) + \I(y_i = y_j) \I(i \ge j), & \text{if ``first'' method is used for ties}, \\
	\I(y_i > y_j) + \I(y_i = y_j) \I(i \le j), & \text{if ``last'' method is used for ties}. 
	\end{cases}
	\end{align}
	Then for $1\le i \le n$, $\rank_i (\bm{y}) = \sum_{j=1}^n \delta_{ij}(y_i, y_j)$. 
\end{lemma}

\begin{proof}[Proof of Lemma \ref{lemma:rank_to_delta}]
	Lemma \ref{lemma:rank_to_delta} follows directly from the definition of ranks. 
\end{proof}

\begin{lemma}\label{lemma:prop_delta}
	$\delta_{ij}(\cdot, \cdot)$'s defined in \eqref{eq:delta} have the following properties: 
	\begin{itemize}
		\item[(a)] for any $(i,j)$, $\delta_{ij}(x, y)$ is %
		increasing 
		in $x$ and is %
		decreasing 
		in $y$;
		\item[(b)] for any $(i,j,l)$, if $\delta_{ij}(x, y) = 1$ and $\delta_{jl}(y, z) = 1$, then we must have $\delta_{il}(x, z) = 1$. 
	\end{itemize}
\end{lemma}

\begin{proof}[Proof of Lemma \ref{lemma:prop_delta}]
	Lemma \ref{lemma:prop_delta} follows directly from the definition of $\delta_{ij}(\cdot, \cdot)$'s. 
\end{proof}

\begin{lemma}\label{lemma:rank_sum_eff_incre}
	For any $n\ge 1$, $\bm{z} \in \{0,1\}^n$, and $\bm{y}, \bm{\eta}, \bm{\xi} \in \mathbb{R}^n$ with $\bm{\eta} \vecge \bm{0} \vecge \bm{\xi}$, 
	if $\phi(\cdot)$ is a 
	monotone increasing 
	function and ``first'' or ``last'' method is used for  ties, 
	then 
	$$
		\sum_{i=1}^n z_i \phi\left( \rank_i(\bm{y} + \bm{z} \circ \bm{\eta} + (\bm{1}-\bm{z})\circ \bm{\xi}) \right) 
		\ge 
		\sum_{i=1}^n z_i \phi( \rank_i( \bm{y} ) ).
	$$
\end{lemma}

\begin{proof}[Proof of Lemma \ref{lemma:rank_sum_eff_incre}]
	Let $m = \sum_{i=1}^n z_i$, and $\bm{w} = \bm{y} + \bm{z} \circ \bm{\eta} + (\bm{1}-\bm{z})\circ \bm{\xi}$. We use 
	$(i_1, i_2, \ldots, i_m)$ to denote permutation of the set $\{i: z_i = 1, 1\le i \le n\}$ such that $\rank_{i_1}(\bm{y}) < \rank_{i_2}(\bm{y}) < \ldots < \rank_{i_m}(\bm{y})$, 
	and $(j_1, j_2, \ldots, j_m)$ to denote permutation of the set $\{i: z_i = 1, 1\le i\le n\}$ such that $\rank_{j_1}(\bm{w}) < \rank_{j_2}(\bm{w}) < \ldots <  \rank_{j_m}(\bm{w})$. 
	For each $1\le k \le m$, from Lemma \ref{lemma:rank_to_delta}, we have
	\begin{align}\label{eq:rank_y_ik}
		\rank_{i_k}(\bm{y}) & = \sum_{l=1}^n \delta_{i_k l}(y_{i_k}, y_l) = \sum_{l:z_l = 1} \delta_{i_k l}(y_{i_k}, y_l) + \sum_{l:z_l = 0} \delta_{i_k l}(y_{i_k}, y_l) 
		= k + \sum_{l:z_l = 0} \delta_{i_k l}(y_{i_k}, y_l), 
	\end{align}
	where the last equality holds due to the construction of $i_k$. By the same logic, for $1\le k \le m$, 
	\begin{align}\label{eq:rank_w_jk}
			\rank_{j_k}(\bm{w})
		& = k + \sum_{l:z_l = 0} \delta_{j_k l}(w_{j_k}, w_l). 
	\end{align}
	
	First, we prove $\delta_{j_k i_k}(w_{j_k}, y_{i_k}) = 1$ for all $1\le  k \le m$. 
	For $1\le p\le k$, 
	because $w_{j_p} = y_{j_p} + \eta_{j_p} \ge y_{j_p}$, 
	from Lemma \ref{lemma:prop_delta}(a), we  have  $\delta_{j_k j_p}(w_{j_k}, y_{j_p}) \ge \delta_{j_k j_p}(w_{j_k}, w_{j_p}) =  1$. 
	By the definition of $i_k$, there must exist a $j_q \in \{j_1, j_2, \ldots, j_k\}$ such that 
	$\delta_{j_q i_k} (y_{j_q}, y_{i_k}) = 1$. 
	Thus, 
	by Lemma \ref{lemma:prop_delta}(b), we must have 
	$\delta_{j_k i_k}(w_{j_k}, y_{ik}) = 1$. 
	
	Second, we prove that $\rank_{j_k}(\bm{w}) \ge \rank_{i_k}(\bm{y})$ for all $1\le k \le m$. 
	By definition, 
	for  any $l$ with $z_l=0$, 
	we have $w_l = y_l + \xi_l \le y_l$. 
	Thus, Lemma \ref{lemma:prop_delta}(a) implies that $\delta_{j_k l}(w_{j_k}, w_l) \ge \delta_{j_k l}(w_{j_k}, y_l)$ for any $l$ with $z_l=0$. 
	From the previous discussion, $\delta_{j_k i_k}(w_{j_k}, y_{i_k}) = 1$ for all $1\le  k \le m$. 
	From Lemma \ref{lemma:prop_delta}(b), 
	for any $1\le k \le m$ and $l$ with $z_l = 0$, 
	if $\delta_{i_k l}(y_{i_k}, y_l) = 1$, then we must have $\delta_{j_k l}(w_{j_k}, w_l) \ge \delta_{j_k l}(w_{j_k}, y_l) = 1$. Consequently, $\delta_{j_k l}(w_{j_k}, w_l) \ge \delta_{i_k l}(y_{i_k}, y_l)$ for any  $1\le k \le m$ and $l$ with $z_l = 0$. 
	From \eqref{eq:rank_y_ik} and \eqref{eq:rank_w_jk}, this  immediately implies that for any $1\le k \le n$, 
	\begin{align}\label{eq:rank_jk_ik}
		\rank_{j_k}(\bm{w}) & = k + \sum_{l:z_l = 0} \delta_{j_k l}(w_{j_k}, w_l) 
		\ge 
		 k + \sum_{l:z_l = 0} \delta_{i_k l}(y_{i_k}, y_l) = \rank_{i_k}(\bm{y}). 
	\end{align}
	
	Third, we prove Lemma \ref{lemma:rank_sum_eff_incre}. From \eqref{eq:rank_jk_ik} and the fact  that $\phi(\cdot)$ is a %
	monotone increasing 
	function, we have 
	\begin{align*}
		\sum_{i=1}^n z_i \phi\left( \rank_i(\bm{w}) \right) & = \sum_{k=1}^m \phi\left( \rank_{j_k}(\bm{w}) \right)
		\ge  \sum_{k=1}^m \phi\left( \rank_{j_k}(\bm{y}) \right) = 
		\sum_{i=1}^n z_i \phi(\rank_i(\bm{y})).
	\end{align*}
	By the definition of $\bm{w}$, 
    we then derive Lemma \ref{lemma:rank_sum_eff_incre}. 
\end{proof}

\begin{proof}[\bf Proof of Proposition \ref{prop:stat_property_rank}]
    Below we consider the statistic $t_2(\cdot, \cdot)$ in \eqref{eq:stat_class_rank}, 
    and assume that $\phi(\cdot)$ is a %
	monotone increasing 
	function and ``first'' or ``last'' method is used for ties. 
	From Lemma \ref{lemma:rank_sum_eff_incre}, we  can immediately know that $t_2(\cdot, \cdot)$  is effect increasing. 
	Below we prove that it is also distribution free under exchangeable treatment assignment.  
	
	For any $\bm{y} \in \mathbb{R}^n$, 
	because we use ``first'' or ``last'' method for ties, 
	there must exists a permutations $\pi(\cdot)$ of $\{1, 2, \ldots, n\}$ such that 
	$(r_{\pi(1)}(\bm{y}), r_{\pi(2)}(\bm{y}), \ldots, r_{\pi(n)}(\bm{y}) ) = 
	(1, 2, \ldots, n),$
	where $\pi$ depends on both the outcome vector $\bm{y}$ and the original ordering of units. 
	Consequently, 
	\begin{align}\label{eq:proof_dist_free_rank_sum}
	    \sum_{i=1}^n Z_i \phi(\rank_i(\bm{y})) = \sum_{i=1}^n Z_{\pi(i)} \phi(\rank_{\pi(i)}(\bm{y})) = \sum_{i=1}^n Z_{\pi(i)} \phi(i). 
	\end{align}
	Because the ordering of units is independent of the treatment assignment, 
	by Definition \ref{def:exchange}, 
	conditional on $\pi(\cdot)$, 
	\eqref{eq:proof_dist_free_rank_sum} follows the same distribution as 
	$\sum_{i=1}^n Z_i \phi(i)$, which depends neither on the permutation $\pi(\cdot)$ nor the outcome vector $\bm{y}$. 
	Therefore, the statistic $t_2(\cdot, \cdot)$ in \eqref{eq:stat_class_rank} is distribution free, satisfying Definition \ref{def:dist_free}. 
	
	From the above, Proposition \ref{prop:stat_property_rank} holds. 
\end{proof}

\begin{proof}[\bf Examples of statistics satisfying exactly one of Definitions \ref{def:effect_increase}--\ref{def:dist_free}]
    Below we construct statistics that satisfy only one of the three properties defined in Definitions \ref{def:effect_increase}--\ref{def:dist_free}. 

First, we define $\tilde{t}_1(\bm{z}, \bm{y}) = y_{i(\bm{z})}$, where $i(\bm{z}) = \argmin_{i: z_i = 1} i$. 
It is easy to show that $\tilde{t}_1(\cdot, \cdot)$ is effect increasing, and it is not distribution free even under a CRE. 
Below we give a counterexample to show that it is not diffential increasing. 
Define $\bm{a} = (1,1,0)^\top, \bm{z} = (0,1,1)^\top,  \bm{y} = (0,0,0)^\top$, and $\bm{\eta} = (1, 2, 0)^\top \vecge \bm{0}$. 
Then we have 
$
\tilde{t}_1(\bm{a}, \bm{y}+ \bm{a}\circ \bm{\eta}) - \tilde{t}_1(\bm{a}, \bm{y}) = (y_1 + a_1 \eta_1) - y_1 = \eta_1 = 1,
$
and 
$
	\tilde{t}_1(\bm{z}, \bm{y}+ \bm{a}\circ \bm{\eta}) - \tilde{t}_1(\bm{z}, \bm{y}) = (y_2 + a_2 \eta_2) - y_2 = \eta_2 = 2 > \tilde{t}_1(\bm{a}, \bm{y}+ \bm{a}\circ \bm{\eta}) - \tilde{t}_1(\bm{a}, \bm{y}) . 
$
Thus, $\tilde{t}_1(\cdot, \cdot)$ is not differential increasing. 

Second, we define $\tilde{t}_2(\bm{z}, \bm{y}) = \sum_{i=1}^n z_i y_i - 2 \sum_{i=1}^n y_i$. 
Using Proposition \ref{prop:stat_property_dim} with $t_1(\cdot, \cdot)$ in \eqref{eq:stat_class_y}, $\phi_{1i}$ being identity function and $\phi_{0i}$ being zero function, 
we can know that $t_1(\bm{z}, \bm{y}) = \sum_{i=1}^n z_i y_i$ is differential increasing. 
Therefore, for any $\bm{z}, \bm{a} \in \mathcal{Z}$ and $\bm{y}, \bm{\eta} \in \mathbb{R}^n$ with $\bm{\eta} \vecge \bm{0}$, 
\begin{align*}
	\tilde{t}_2(\bm{z}, \bm{y} + \bm{a} \circ \bm{\eta}) - \tilde{t}_2(\bm{z}, \bm{y})
	& = t_1(\bm{z}, \bm{y} + \bm{a} \circ \bm{\eta}) - 2 \sum_{i=1}^n (y_i + a_i\eta_i)
	- t_1(\bm{z}, \bm{y}) + 2 \sum_{i=1}^n y_i\\
	& =  t_1(\bm{z}, \bm{y} + \bm{a} \circ \bm{\eta}) - t_1(\bm{z}, \bm{y}) - 2 \sum_{i=1}^n (y_i + a_i\eta_i)
	 + 2 \sum_{i=1}^n y_i\\
	 & \le 
	 t_1(\bm{a}, \bm{y} + \bm{a} \circ \bm{\eta}) - t_1(\bm{a}, \bm{y}) - 2 \sum_{i=1}^n (y_i + a_i\eta_i)
	 + 2 \sum_{i=1}^n y_i\\
	 & = \tilde{t}_2(\bm{a}, \bm{y} + \bm{a} \circ \bm{\eta}) - \tilde{t}_2(\bm{a}, \bm{y}).
\end{align*}
Thus, $\tilde{t}_2(\cdot, \cdot)$ is differential increasing. 
It is easy to show that $\tilde{t}_2(\cdot, \cdot)$ is not distribution free even under a CRE. 
Below we give counterexample to show $\tilde{t}_2(\cdot, \cdot)$ is not effect increasing. 
Define 
$\bm{z} = (1,0)^\top, \bm{y} = (0,0)^\top$ and $\bm{\eta} = (1,0)^\top\vecge \bm{0}$. 
Then we have
$\tilde{t}_2(\bm{z}, \bm{y}) = y_1 - 2(y_1+y_2) = 0$, 
and 
$\tilde{t}_2(\bm{z}, \bm{y} + \bm{z} \circ \bm{\eta}) = (y_1+\eta_1) - 2(y_1+\eta_1 + y_2) = -1 < \tilde{t}_2(\bm{z}, \bm{y})$. 
Thus,  $\tilde{t}_2(\cdot, \cdot)$ is not effect increasing.

Third, we define $\tilde{t}_3(\bm{z}, \bm{y}) = - \sum_{i=1}^2 z_i \rank_i(\bm{y})$ and consider a CRE where one unit is randomly assigned to treatment and the remaining one is randomly assigned to control.  
From Proposition \ref{prop:stat_property_rank}, 
it is not hard to see that $\tilde{t}_3(\bm{z}, \bm{y})$ is not effect increasing but it is distribution free. 
Below we give a numerical example to show it is not differential increasing. 
Let $\bm{y} = (1, 2)^\top$, $\bm{\eta} = (2, 0)^\top$, $\bm{a} = (1, 0)^\top$ and $\bm{z} = (0,1)^\top$. 
Then we have 
$
\tilde{t}_3(\bm{z}, \bm{y} + \bm{a} \circ \bm{\eta}) - \tilde{t}_3(\bm{z}, \bm{y}) = 1 > -1 = \tilde{t}_3(\bm{a}, \bm{y} + \bm{a} \circ \bm{\eta}) - \tilde{t}_3(\bm{a}, \bm{y})$.
\end{proof}

\section{Proofs for broader justification of Fisher randomization test}

\subsection{Proofs of Theorems \ref{thm:null_broader_monotone_control},  \ref{thm:null_broader_monotone_both} and \ref{thm:effect_range} 
and Corollaries \ref{cor:ci_max_imp_control} and \ref{cor:ci_max_imp_both}}

To prove Theorems \ref{thm:null_broader_monotone_control} and \ref{thm:null_broader_monotone_both}, we need  the following lemma. 

\begin{lemma}\label{lemma:valid}
	For any function $h(\cdot): \mathcal{Z} \rightarrow \mathbb{R}$, 
	let $\bm{Z}$ be the random treatment assignment vector, 
	and 
	$H(c) = \Pr\{h(\bm{Z}) \ge c \}$ be the tail probability of the random variable $h(\bm{Z})$. 
	Then, $H(h(\bm{Z}))$ is stochastically larger than or equal to $\Unif[0,1]$, in the sense that for any $\alpha \in (0, 1)$, 
	$
	\Pr\{ H(h(\bm{Z})) \le \alpha \} \le \alpha. 
	$
\end{lemma}

\begin{proof}[Proof of Lemma \ref{lemma:valid}]
	By the property of probability measure, the tail probability function $H(\cdot)$ is decreasing and left-continuous. 
	For any $\alpha \in (0,1)$, define $c_{\alpha} = \inf_{c}\{c: H(c) \le \alpha\}$. 
	Below we consider two cases in which $H(c_\alpha) \le \alpha$ and  $H(c_\alpha) > \alpha$, respectively. 
	
	First, we consider the case in which $H(c_\alpha) \le \alpha$. 
	By definition, we can know that $H(h(\bm{Z})) \le \alpha$ is equivalent to $h(\bm{Z}) \ge c_{\alpha}$. 
	This further implies that 
	$
	\Pr\{ H(h(\bm{Z})) \le \alpha \} = 	\Pr\{ h(\bm{Z}) \ge c_{\alpha} \} = H(c_\alpha) \le \alpha.
	$
	
	Second, we consider the case in which $H(c_\alpha) > \alpha$. By definition, we can know that  $H(h(\bm{Z})) \le \alpha$ is equivalent to $h(\bm{Z}) > c_{\alpha}$.  This further implies that 
	$
	\Pr\{ H(h(\bm{Z})) \le \alpha \} = 	\Pr\{ h(\bm{Z}) > c_{\alpha} \} = \lim_{a \rightarrow 0+}\Pr\{ h(\bm{Z}) \ge c_{\alpha} + a \}
	=  \lim_{a \rightarrow 0+} H(c_{\alpha} + a) \le \alpha. 
	$
	
	From the above, Lemma \ref{lemma:valid} holds. 
\end{proof}

\begin{proof}[\bf \lxr Proof of Theorem \ref{thm:null_broader_monotone_control}(a)]
Suppose the bounded null $H_{\vecle\bm{\delta}}$ in \eqref{eq:null_less_delta} holds, i.e., the true treatment effect satisfies 
$
\bm{\tau} \vecle \bm{\delta}. 
$
Then, by the definition in \eqref{eq:impute_potential_outcome}, 
the imputed potential outcomes satisfy
\begin{align}\label{eq:order_impute_control_proof}
	\bm{Y}_{\bm{Z}, \bm{\delta}}(1) - \bm{Y}(1)
	= (\bm{1} - \bm{Z}) \circ  (\bm{\delta} - \bm{\tau})
	\vecge \bm{0}
	\ \ \ 
	\text{and}
	\ \ \ 
	\bm{Y}_{\bm{Z}, \bm{\delta}}(0) - \bm{Y}(0)
	=
	\bm{Z} \circ ( \bm{\tau} - \bm{\delta} ) \vecle \bm{0}. 
\end{align}
This is because any actual treated value is smaller than or equal to the imputed, 
and any actual control value is larger than or equal to the imputed.

We first consider the case in which the test statistic is effect increasing. 
Similar to \eqref{eq:impute_potential_outcome},
for any $\bm{a} \in \mathcal{Z}$, 
the imputed control potential outcome if the observed treatment assignment was $\bm{a}$ would be 
$
\bm{Y}_{\bm{a}, \bm{\delta}}(0) = \bm{Y}(\bm{a}) - \bm{a} \circ \bm{\delta} 
= \bm{a} \circ \{ \bm{Y}(1) - \bm{\delta} \} + (\bm{1}-\bm{a}) \circ \bm{Y}(0),
$
and 
the difference between $\bm{Y}_{\bm{Z}, \bm{\delta}}(0)$ and $\bm{Y}_{\bm{a}, \bm{\delta}}(0)$ has the following equivalent forms: 
\begin{align*}
	\bm{Y}_{\bm{Z}, \bm{\delta}}(0)  - \bm{Y}_{\bm{a}, \bm{\delta}}(0) & = 
	\bm{a} \circ \bm{Y}_{\bm{Z}, \bm{\delta}}(0) + (\bm{1}-\bm{a}) \circ \bm{Y}_{\bm{Z}, \bm{\delta}}(0)  - 
	\bm{a} \circ 
	\left\{ \bm{Y}(1) - \bm{\delta} \right\} 
	- (\bm{1} - \bm{a}) \circ \bm{Y}(0)\\
	& = 
	\bm{a} \circ 
	\left\{
	\bm{Y}_{\bm{Z}, \bm{\delta}}(1) -  \bm{Y}(1) \right\}
	+ (\bm{1} - \bm{a}) \circ 
	\left\{
	\bm{Y}_{\bm{Z}, \bm{\delta}}(0) - \bm{Y}(0)
	\right\}
\end{align*}
From \eqref{eq:order_impute_control_proof} and Definition \ref{def:effect_increase}, 
$t( \bm{a}, \bm{Y}_{\bm{Z}, \bm{\delta}}(0) ) \ge t( \bm{a}, \bm{Y}_{\bm{a}, \bm{\delta}}(0) )$. 
Thus, 
for any $c\in \mathbb{R}$,  the imputed tail probability $G_{\bm{Z}, \bm{\delta}}(c)$ in \eqref{eq:tail_prob_tilde} of the test statistic can be bounded by 
\begin{align*}
    G_{\bm{Z}, \bm{\delta}} (c) 
    = \sum_{\bm{a} \in \mathcal{Z}} \Pr(\bm{A} = \bm{a}) \I\left\{
    t(\bm{a}, \bm{Y}_{\bm{Z},\bm{\delta}}(0))
    \ge 
    c
    \right\}
    \ge 
    \sum_{\bm{a} \in \mathcal{Z}} \Pr(\bm{A} = \bm{a}) \I\left\{
    t(\bm{a}, \bm{Y}_{\bm{a},\bm{\delta}}(0))
    \ge 
    c
    \right\}, 
\end{align*}
where the right hand side is actually the tail probability of the true randomization distribution of the test statistic $t(\bm{Z}, \bm{Y}_{\bm{Z}, \bm{\delta}}(0))$. 
{\lxr From Lemma \ref{lemma:valid}}, 
we can derive
that $p_{\bm{Z}, \bm{\delta}} \equiv G_{\bm{Z}, \bm{\delta}}(t(\bm{Z}, \bm{Y}_{\bm{Z}, \bm{\delta}}(0)))$ is stochastically larger than or equal to $\Unif[0,1]$, i.e., it is a valid $p$-value for testing $H_{\vecle\bm{\delta}}$.
	
We then consider the case in which the test statistic is differential increasing. 
From \eqref{eq:order_impute_control_proof}, 
\begin{align}\label{eq:change_diff_increase_proof}
    t(\bm{a}, \bm{Y}(0)) - t(\bm{a}, \bm{Y}_{\bm{Z}, \bm{\delta}}(0))
	& = 
	t(\bm{a}, \bm{Y}_{\bm{Z}, \bm{\delta}}(0) + \bm{Z} \circ (  \bm{\delta} - \bm{\tau} ) ) - t(\bm{a}, \bm{Y}_{\bm{Z}, \bm{\delta}}(0)). 
\end{align}
From Definition \ref{def:diff_increase}, the change of the statistic in \eqref{eq:change_diff_increase_proof} is maximized at $\bm{a} = \bm{Z}$, which immediately implies that 
$
t(\bm{Z}, \bm{Y}(0) ) - t(\bm{Z}, \bm{Y}_{\bm{Z}, \bm{\delta}}(0))
\ge 
t(\bm{a}, \bm{Y}(0)) - t(\bm{a}, \bm{Y}_{\bm{Z}, \bm{\delta}}(0)) 
$
for any $\bm{a} \in \mathcal{Z}$. 
Thus, 
	$
	t(\bm{a}, \bm{Y}_{\bm{Z},\bm{\delta}}(0))
	- 
	t(\bm{Z}, \bm{Y}_{\bm{Z}, \bm{\delta}}(0))
	\ge 
	t(\bm{a}, \bm{Y}(0))
	- 
	t(\bm{Z}, \bm{Y}(0))
	$
for any $\bm{a} \in \mathcal{Z}$, 
and
the randomization $p$-value $p_{\bm{Z}, \bm{\delta}}$ in \eqref{eq:p_val_imp_control} is bounded by 
\begin{align*}
	p_{\bm{Z}, \bm{\delta}}
	& 
	= 
	\sum_{\bm{a} \in \mathcal{Z}} \Pr(\bm{A} = \bm{a}) \I\left\{
	t(\bm{a}, \bm{Y}_{\bm{Z},\bm{\delta}}(0))
	\ge 
	t(\bm{Z}, \bm{Y}_{\bm{Z}, \bm{\delta}}(0))
	\right\}\\
	& \ge 
	\sum_{\bm{a} \in \mathcal{Z}} \Pr(\bm{A} = \bm{a}) \I\left\{
	t( \bm{a}, \bm{Y}(0) )
	\ge 
	t(\bm{Z}, \bm{Y}(0))
	\right\},
\end{align*} 
which is actually the tail probability of the true randomization distribution of $t(\bm{Z}, \bm{Y}(0))$ evaluated at its realized value. 
{\lxr From Lemma \ref{lemma:valid}}, we can derive that 
$p_{\bm{Z}, \bm{\delta}}$ is stochastically larger than or equal to $\Unif[0,1]$, i.e., it is a valid $p$-value for testing $H_{\vecle\bm{\delta}}$.  
	
{\lxr From the above, Theorem \ref{thm:null_broader_monotone_control}(a) holds. 
As a side note,} although both effect increasing and differential increasing test statistics can lead to valid randomization tests for bounded null hypotheses, 
their proofs are rather different, as shown above. 
Specifically, 
the randomization $p$-value $p_{\bm{Z}, \bm{\delta}}$ using effect increasing test statistic is bounded by the tail probability of $t(\bm{Z}, \bm{Y}_{\bm{Z}, \bm{\delta}}(0))$ evaluated at its realized value, 
while that using differential increasing test statistic is 
bounded by the tail probability of $t(\bm{Z}, \bm{Y}(0))$ evaluated at its realized value. 
The two bounds are generally different, although they are both stochastically larger than or equal to $\Unif[0,1]$.
\end{proof}

\begin{proof}[\bf Proof of Theorem \ref{thm:null_broader_monotone_control}{\lxr (b)}]
    
    For any $\bm{\delta}, \overline{\bm{\delta}} \in \mathbb{R}^n$ with $\bm{\delta} \vecle \overline{\bm{\delta}}$, 
	by the definition in \eqref{eq:impute_potential_outcome}, 
	the differences between the imputed treatment and control potential outcomes under nulls $\overline{\bm{\delta}}$ and $\bm{\delta}$ satisfy
	\begin{align}\label{eq:diff_imput_poten_outcome_two_null}
	\bm{Y}_{\bm{Z}, \overline{\bm{\delta}}}(1) - \bm{Y}_{\bm{Z}, \bm{\delta}}(1)
	& = 
	\bm{Y}+ (\bm{1} - \bm{Z}) \circ \overline{ \bm{\delta} }
	- \left\{ \bm{Y}+ (\bm{1} - \bm{Z}) \circ \bm{\delta} \right\}
	= (\bm{1}-\bm{Z}) \circ ( \overline{ \bm{\delta} } - \bm{\delta}) \vecge \bm{0} , 
	\nonumber
	\\
	\bm{Y}_{\bm{Z}, \overline{ \bm{\delta} } }(0) - \bm{Y}_{\bm{Z}, \bm{\delta}}(0)
	& = 
	\bm{Y} - \bm{Z} \circ \overline{ \bm{\delta} }
	- \left\{ \bm{Y} - \bm{Z} \circ \bm{\delta}  \right\}
	= \bm{Z} \circ (\bm{\delta} - \overline{ \bm{\delta} }) \vecle \bm{0}. 
	\end{align}

	We first consider the case in which the test statistic is differential increasing. 
	From \eqref{eq:diff_imput_poten_outcome_two_null} and Definition \ref{def:diff_increase}, 
	for any $\bm{\delta}, \overline{\bm{\delta}} \in \mathbb{R}^n$ with $\bm{\delta} \vecle \overline{\bm{\delta}}$ and any $\bm{a} \in \mathcal{Z}$, 
	\begin{align*}
	t(\bm{Z}, \bm{Y}_{\bm{Z}, \bm{\delta} }(0) ) - 
	t( \bm{Z},  \bm{Y}_{\bm{Z}, \overline{ \bm{\delta} } }(0) )
	& = 
	t\left(\bm{Z}, \bm{Y}_{\bm{Z}, \overline{ \bm{\delta} } }(0) + \bm{Z} \circ (\overline{ \bm{\delta} } - \bm{\delta}) \right) 
	- 
	t\left( \bm{Z},  \bm{Y}_{\bm{Z}, \overline{ \bm{\delta} } }(0) \right)\\
	& \ge
	t\left(\bm{a}, \bm{Y}_{\bm{Z}, \overline{ \bm{\delta} } }(0) + \bm{Z} \circ (\overline{ \bm{\delta} } - \bm{\delta}) \right) 
	- 
	t\left( \bm{a},  \bm{Y}_{\bm{Z}, \overline{ \bm{\delta} } }(0) \right)\\
	& = 
	t(\bm{a}, \bm{Y}_{\bm{Z}, \bm{\delta} }(0) ) - 
	t( \bm{a},  \bm{Y}_{\bm{Z}, \overline{ \bm{\delta} } }(0) ).
	\end{align*}
	Consequently, for any $\bm{a} \in \mathcal{Z}$, 
	\begin{align*}
	t( \bm{a},  \bm{Y}_{\bm{Z}, \overline{ \bm{\delta} } }(0) )  - 
	t( \bm{Z},  \bm{Y}_{\bm{Z}, \overline{ \bm{\delta} } }(0) ) \ge 
	t(\bm{a}, \bm{Y}_{\bm{Z}, \bm{\delta} }(0) ) - t(\bm{Z}, \bm{Y}_{\bm{Z}, \bm{\delta} }(0) ).
	\end{align*}
	By the definition in \eqref{eq:p_val_imp_control}, 
	\begin{align*}
	p_{\bm{Z}, \overline{\bm{\delta}} }
	& = 
	\sum_{\bm{a} \in \mathcal{Z}} \Pr(\bm{A} = \bm{a}) \I\left\{
	t(\bm{a}, \bm{Y}_{\bm{Z}, \overline{\bm{\delta}} }(0))
	\ge 
	t(\bm{Z}, \bm{Y}_{\bm{Z}, \overline{ \bm{\delta}} }(0))
	\right\}\\
	& \ge 
	\sum_{\bm{a} \in \mathcal{Z}} \Pr(\bm{A} = \bm{a}) \I\left\{
	t(\bm{a}, \bm{Y}_{\bm{Z},\bm{\delta}}(0))
	\ge 
	t(\bm{Z}, \bm{Y}_{\bm{Z}, \bm{\delta}}(0))
	\right\}
	= p_{\bm{Z}, \bm{\delta}}. 
	\end{align*} 
	
	We then consider the case in which the test statistic is effect increasing and distribution free. 
	From \eqref{eq:diff_imput_poten_outcome_two_null} and Definition \ref{def:effect_increase}, 
	for any $\bm{\delta}, \overline{\bm{\delta}} \in \mathbb{R}^n$ with $\bm{\delta} \vecle \overline{\bm{\delta}}$, 
	\begin{align}\label{eq:order_t_control}
	t( \bm{Z}, \bm{Y}_{\bm{Z}, \bm{\delta}}(0) )
	& = 
	t\left(
	\bm{Z}, \bm{Y}_{\bm{Z}, \overline{ \bm{\delta} } }(0) + \bm{Z} \circ (\overline{ \bm{\delta} } - \bm{\delta})
	\right)
	\ge 
	t(
	\bm{Z}, \bm{Y}_{\bm{Z}, \overline{ \bm{\delta} } }(0) 
	).
	\end{align}
	From \eqref{eq:tail_prob_tilde} and Definition \ref{def:dist_free}, we can know that 
	\begin{align}\label{eq:G_0_dist_free}
	G_{\bm{Z}, \bm{\delta}} (c) \equiv \sum_{\bm{a} \in \mathcal{Z}} \Pr(\bm{A} = \bm{a}) \I\left\{
	t(\bm{a}, \bm{Y}_{\bm{Z},\bm{\delta}}(0))
	\ge 
	c
	\right\} = 
	G_0 (c), 
	\end{align}
	a function that does not depend on $\bm{Z}$ or $\bm{\delta}$. 
	Moreover, $G_0 (c)$ is decreasing in $c$. 
	From \eqref{eq:order_t_control} and \eqref{eq:G_0_dist_free}, 
	by the definition in \eqref{eq:p_val_imp_control}, 
	we then have 
	\begin{align*}
	p_{\bm{Z}, \bm{\delta}} & = G_{\bm{Z}, \bm{\delta}} \left\{ t(\bm{Z},\bm{Y}_{\bm{Z}, \bm{\delta}}(0))\right\} 
	= 
	G_0 \left\{ t(\bm{Z},\bm{Y}_{\bm{Z}, \bm{\delta}}(0))\right\} 
	\\
	& \le 
	G_0 \left\{ t(\bm{Z},\bm{Y}_{\bm{Z}, \overline{ \bm{\delta}} }(0))\right\}
	=
	G_{\bm{Z}, \overline{\bm{\delta}} } \left\{ t(\bm{Z},\bm{Y}_{\bm{Z}, \overline{ \bm{\delta}} }(0))\right\}
	=  p_{\bm{Z}, \overline{\bm{\delta}} }.
	\end{align*} 
	
	From the above, Theorem \ref{thm:null_broader_monotone_control}(b) holds. 
\end{proof}

\begin{proof}[\bf Proof of Theorem \ref{thm:null_broader_monotone_both}]
    Suppose Theorem \ref{thm:null_broader_monotone_both}(b) holds. 
    Then when the bounded null $H_{\vecle\bm{\delta}}$ in \eqref{eq:null_less_delta} holds, i.e., $\bm{\tau} \vecle \bm{\delta}$, 
    we must have $\tilde{p}_{\bm{Z}, \bm{\delta}} \ge \tilde{p}_{\bm{Z}, \bm{\tau}}$, which is stochastically larger than or equal to $\Unif[0,1]$ by the validity of Fisher randomization test for sharp null hypothesis $H_{\bm{\delta}}: \bm{\delta} = \bm{\tau}.$ 
    Therefore, Theorem \ref{thm:null_broader_monotone_both}(b) implies \ref{thm:null_broader_monotone_both}(a), and,  to prove Theorem \ref{thm:null_broader_monotone_both}, it suffices to prove Theorem \ref{thm:null_broader_monotone_both}(b). 
    Below we prove Theorem \ref{thm:null_broader_monotone_both}(b).

	Let $\bm{\delta}, \overline{\bm{\delta}} \in \mathbb{R}^n$ be two constant vectors satisfying $\bm{\delta} \vecle \overline{\bm{\delta}}$. 
	For any $\bm{a} \in \mathcal{Z}$, 
	the corresponding observed outcomes from
	the imputed potential outcomes under nulls $\overline{\bm{\delta}}$ and $\bm{\delta}$ satisfy 
	\begin{align*}
	\bm{Y}_{\bm{Z}, \overline{\bm{\delta}}}(\bm{a}) - \bm{Y}_{\bm{Z}, \bm{\delta}}(\bm{a})
	& = 
	\bm{a} \circ \bm{Y}_{\bm{Z}, \overline{\bm{\delta}}}(1) + (\bm{1} - \bm{a}) \circ \bm{Y}_{\bm{Z}, \overline{\bm{\delta}}}(0)
	- 
	\left\{
	\bm{a} \circ \bm{Y}_{\bm{Z}, \bm{\delta}}(1) + (\bm{1} - \bm{a}) \circ \bm{Y}_{\bm{Z}, \bm{\delta}}(0)
	\right\}\\
	& = 
	\bm{a} \circ 
	\left\{ \bm{Y}_{\bm{Z}, \overline{\bm{\delta}}}(1) - \bm{Y}_{\bm{Z}, \bm{\delta}}(1) \right\}
	+ 
	(\bm{1} - \bm{a}) \circ 
	\left\{
	\bm{Y}_{\bm{Z}, \overline{ \bm{\delta} } }(0) - \bm{Y}_{\bm{Z}, \bm{\delta}}(0)
	\right\}.
	\end{align*}
	From \eqref{eq:diff_imput_poten_outcome_two_null} and Definition \ref{def:effect_increase}, for any effect increasing statistic $t(\cdot, \cdot)$ and any $\bm{a} \in \mathcal{Z}$, 
	$
	t(\bm{a}, \bm{Y}_{\bm{Z}, \overline{\bm{\delta}}}(\bm{a}))
	\ge 
	t(\bm{a}, \bm{Y}_{\bm{Z}, \bm{\delta}}(\bm{a}) ).
	$
	By the definition in \eqref{eq:p_val_imp_both}, we then have 
	\begin{align*}
	\tilde{p}_{\bm{Z}, \overline{ \bm{\delta}} }
	& = 
	\sum_{\bm{a} \in \mathcal{Z}} \Pr(\bm{A} = \bm{a}) \I\left\{
	t(\bm{a}, \bm{Y}_{\bm{Z}, \overline{\bm{\delta}} }(\bm{a}))
	\ge 
	t(\bm{Z}, \bm{Y})
	\right\}\\
	& \ge
	\sum_{\bm{a} \in \mathcal{Z}} \Pr(\bm{A} = \bm{a}) \I\left\{
	t(\bm{a}, \bm{Y}_{\bm{Z}, \bm{\delta}}(\bm{a}))
	\ge 
	t(\bm{Z}, \bm{Y})
	\right\}
	= 
	\tilde{p}_{\bm{Z}, \bm{\delta} }. 
	\end{align*}
	Therefore, 
	Theorem \ref{thm:null_broader_monotone_both}(b) holds. 
\end{proof}

\begin{proof}[\bf Proof of Corollary \ref{cor:ci_max_imp_control}]
    We first prove (a) in Corollary \ref{cor:ci_max_imp_control}. 
	By definition, the coverage probability of the set $
	\{c: p_{\bm{Z}, c\bm{1}} > \alpha, c \in \mathbb{R} \}
	$ has the following equivalent forms:
	\begin{align}\label{eq:cov_ci_max_control}
	\Pr\left( \tau_{\max} \in \{c: p_{\bm{Z}, c\bm{1}} > \alpha, c \in \mathbb{R} \}
	\right)
	& = 
	\Pr\left( p_{\bm{Z}, \tau_{\max}\bm{1}} > \alpha
	\right)
	= 
	1 - \Pr\left(
	p_{\bm{Z}, \tau_{\max}\bm{1}} \le \alpha
	\right).
	\end{align}
	Because the test statistic is either differential increasing or effect increasing, 
	from Theorem \ref{thm:null_broader_monotone_control}(a), $p_{\bm{Z}, \tau_{\max}\bm{1}}$ is a valid $p$-value for testing the bounded null $H_{\tau_{\max}\bm{1}}$. Because the null $H_{\tau_{\max}\bm{1}}$ holds by definition, we have $\Pr(
	p_{\bm{Z}, \tau_{\max}\bm{1}} \le \alpha
	) \le \alpha,$ 
	and thus the coverage probability \eqref{eq:cov_ci_max_control} is larger than or equal to $1-\alpha$. 
	Therefore, $\{c: p_{\bm{Z}, c\bm{1}} > \alpha, c \in \mathbb{R} \}$ is a $1-\alpha$ confidence set for $\tau_{\max}$.

	We then prove (b) in Corollary \ref{cor:ci_max_imp_control}. 
	Because the test statistic is either differential increasing, or both effect increasing and distribution free, 
    from Theorem \ref{thm:null_broader_monotone_control}(b), $p_{\bm{Z}, c\bm{1}}$ is %
	increasing 
	in $c$, which implies that the confidence set must have the form of $(\underline{c}, \infty)$ or $[\underline{c}, \infty)$ with $\underline{c} = \inf \{c: p_{\bm{Z}, c\bm{1}} > \alpha, c \in \mathbb{R} \}$. 
	
	From the above, Corollary \ref{cor:ci_max_imp_control} holds. 
\end{proof}

\begin{proof}[\bf Proof of Corollary \ref{cor:ci_max_imp_both}]
    Corollary \ref{cor:ci_max_imp_both} follows from Theorem \ref{thm:null_broader_monotone_both}, and its proof is almost the same as that for Corollary \ref{cor:ci_max_imp_control}. 
    Thus, we omit its proof here. 
\end{proof}

\begin{proof}[\bf Proof of Theorem \ref{thm:effect_range}]
	First, the coverage probability of the interval $[\max\{\hat{\tau}_{\max}^L - \hat{\tau}_{\min}^U, 0\}, \infty)$ is 
	\begin{align*}
	& \quad\ \Pr\left(
	\tau_{\max} - \tau_{\min} \ge \max\{\hat{\tau}_{\max}^L - \hat{\tau}_{\min}^U, 0\}
	\right)\\
	& 
	= \Pr\left(
	\tau_{\max} - \tau_{\min} \ge \hat{\tau}_{\max}^L - \hat{\tau}_{\min}^U
	\right)
	\ge 
	\Pr\left( 
	\tau_{\max}  \ge \hat{\tau}_{\max}^L,
	\tau_{\min} \le \hat{\tau}_{\min}^U
	\right)
	\\
	& = 1 - 
	\Pr\left( 
	\tau_{\max}  < \hat{\tau}_{\max}^L \text{ or }
	\tau_{\min} >  \hat{\tau}_{\min}^U
	\right)
	\ge 1 - \Pr\left( \tau_{\max}  < \hat{\tau}_{\max}^L \right) - \Pr\left( \tau_{\min} >  \hat{\tau}_{\min}^U \right)
	 \\
	 & \ge 1 - \alpha/2 - \alpha/2 = 1-\alpha, 
	\end{align*}
	where the last equality holds because $[\hat{\tau}_{\max}^L, \infty)$ and $(-\infty, \hat{\tau}_{\min}^U]$ are $1-\alpha/2$ confidence intervals for $\tau_{\max}$ and $\tau_{\min}$, respectively. 
	Thus, (a) in Theorem \ref{thm:effect_range} holds. 
	
	Second, suppose the null hypothesis of constant treatment effect is true. 
	Then we must have $\tau_{\max} - \tau_{\min} = 0$. 
	From (a) in Theorem \ref{thm:effect_range},  
	\begin{align*}
	\Pr\left(
	\hat{\tau}_{\max}^L - \hat{\tau}_{\min}^U > 0
	\right)
	& = 
	\Pr\left(
	\tau_{\max} - \tau_{\min} \notin [\hat{\tau}_{\max}^L - \hat{\tau}_{\min}^U, \infty)
	\right)\\
	& = 
	1 - 
	\Pr\left(
	\tau_{\max} - \tau_{\min} \in [\hat{\tau}_{\max}^L - \hat{\tau}_{\min}^U, \infty)
	\right)
	\le 1 - (1-\alpha)\\
	& = \alpha,
	\end{align*}
	which implies that the probability of type-I error is at most $\alpha$. Thus, (b) in Theorem \ref{thm:effect_range} holds. 
	
	From the above, Theorem \ref{thm:effect_range} holds. 
\end{proof}

\subsection{Additional comments on randomization $p$-values}\label{sec:comment_p_value}

\begin{proof}[\bf Example of a non-monotone $p$-value that is valid for bounded null]
    Here we 
    present a numerical example showing that a valid $p$-value for testing the bounded null may not have the monotonicity property. 
    We consider a CRE with 3 units, where 2 units are assigned to treatment group and the remaining 1 is assigned to control group. 
    Suppose that the observed treatment assignment vector is $\bm{Z} = (1, 1, 0)^\top$ and the observed outcome vector is $\bm{Y} = (0, 0, 0)^\top$. 
    We consider using the randomization $p$-value $p_{\bm{Z}, \bm{\delta}}$ in \eqref{eq:p_val_imp_control} 
    with test statistic $t(\bm{z}, \bm{y}) = y_{i(\bm{z})}$, where $i(\bm{z}) = \argmin_{i: z_i = 1} i$. 
    From 
    the examples of statistics discussed at the end of 
    Appendix \ref{sec:proof_property}, 
    the statistic $t(\bm{z}, \bm{y}) = y_{i(\bm{z})}$ is effect increasing. 
    From Theorem \ref{thm:null_broader_monotone_control}(a), the randomization $p$-value $p_{\bm{Z}, \bm{\delta}}$ is valid for testing the bounded null $H_{\vecle\bm{\delta}}$. 
    However, the randomization $p$-value $p_{\bm{Z}, \bm{\delta}}$ is not monotone in $\bm{\delta}$. 
    For example, 
    with $\underline{\bm{\delta}} = (-1, 0, 0)^\top \vecle \bm{\delta} = (0, 0, 0)^\top \vecle \overline{\bm{\delta}} = (0, 1, 0)^\top$, 
    we have 
    $p_{\bm{Z}, \underline{\bm{\delta}}} = 2/3 < p_{\bm{Z}, \bm{\delta}} = 1 >  p_{\bm{Z}, \overline{\bm{\delta}}} = 2/3.$
\end{proof}

\begin{proof}[\bf Comment on the equivalence between randomization $p$-values $p_{\bm{Z}, c\bm{1}}$ and $\tilde{p}_{\bm{Z}, c\bm{1}}$]
    We are going to show that the two randomization $p$-values $p_{\bm{Z}, c\bm{1}}$ in \eqref{eq:p_val_imp_both} and $\tilde{p}_{\bm{Z}, c\bm{1}}$ in \eqref{eq:p_val_imp_control} are equivalent for testing the null hypothesis $H_{c\bm{1}}$ of constant effect $c$ for any $c\in \mathbb{R}$,  if  both of them use difference-in-means as the test statistics. 
For any $c\in \mathbb{R}$, 
the difference-in-means statistics for the two randomization $p$-values  evaluated at assignment $\bm{a} \in \mathcal{Z}$ are, respectively, 
\begin{align*}
t(\bm{a}, \bm{Y}_{\bm{Z}, c\bm{1}}(0) ) 
& = 
\frac{1}{m}\bm{a}^\top\bm{Y}_{\bm{Z}, c\bm{1}}(0)
- 
\frac{1}{n-m} (\bm{1}-\bm{a})^\top \bm{Y}_{\bm{Z}, c\bm{1}}(0), 
\end{align*}
and
\begin{align*}
t(\bm{a}, \bm{Y}_{\bm{Z}, c\bm{1}}(\bm{a}) ) 
& = \frac{1}{m}\bm{a}^\top\bm{Y}_{\bm{Z}, c\bm{1}}(1)
- 
\frac{1}{n-m} (\bm{1}-\bm{a})^\top \bm{Y}_{\bm{Z}, c\bm{1}}(0)
\\
& = 
\frac{1}{m}\bm{a}^\top \left\{ \bm{Y}_{\bm{Z}, c\bm{1}}(0) + c\bm{1} \right\}
- 
\frac{1}{n-m} (\bm{1}-\bm{a})^\top \bm{Y}_{\bm{Z}, c\bm{1}}(0)
\\
& = \frac{1}{m}\bm{a}^\top \bm{Y}_{\bm{Z}, c\bm{1}}(0) 
- 
\frac{1}{n-m} (\bm{1}-\bm{a})^\top \bm{Y}_{\bm{Z}, c\bm{1}}(0) 
+ 
\frac{1}{m}\bm{a}^\top( c\bm{1} )\\
& = t(\bm{a}, \bm{Y}_{\bm{Z}, c\bm{1}}(0) )  + c. 
\end{align*}
Therefore, the randomization $p$-values $p_{\bm{Z}, c\bm{1}}$ and $\tilde{p}_{\bm{Z}, c\bm{1}}$ in \eqref{eq:p_val_imp_control} and \eqref{eq:p_val_imp_both} satisfy 
\begin{align*}
p_{\bm{Z}, c\bm{1}}
& = 
\sum_{\bm{a} \in \mathcal{Z}} \Pr(\bm{A} = \bm{a}) \I\left\{
t(\bm{a}, \bm{Y}_{\bm{Z},c\bm{1}}(0))
\ge 
t(\bm{Z}, \bm{Y}_{\bm{Z}, c\bm{1}}(0))
\right\}\\
& = 
\sum_{\bm{a} \in \mathcal{Z}} \Pr(\bm{A} = \bm{a}) \I\left\{
t(\bm{a}, \bm{Y}_{\bm{Z},c\bm{1}}(\bm{a})) - c
\ge 
t(\bm{Z}, \bm{Y}_{\bm{Z}, c\bm{1}}(\bm{Z})) - c
\right\}\\
& = 
\sum_{\bm{a} \in \mathcal{Z}} \Pr(\bm{A} = \bm{a}) \I\left\{
t(\bm{a}, \bm{Y}_{\bm{Z},c\bm{1}}(\bm{a})) 
\ge 
t(\bm{Z}, \bm{Y} )
\right\}\\
& = \tilde{p}_{\bm{Z}, c\bm{1}}.
\end{align*} 
Therefore,  for any $c\in \mathbb{R}$, the randomization $p$-values $p_{\bm{Z}, c\bm{1}}$ and $\tilde{p}_{\bm{Z}, c\bm{1}}$ with the difference-in-means statistics are equivalent for testing the null hypothesis $H_{c\bm{1}}$ of constant effect $c$.
\end{proof}

\section{Proofs for randomization inference on quantiles of individual effects}

\begin{proof}[\bf Proof of Theorem \ref{thm:general_null}]
	From the property of usual randomization test, $p_{\bm{Z}, \bm{\tau}}$ is a valid $p$-value, i.e., 
	$\Pr\left(
	p_{\bm{Z}, \bm{\tau}} \le \alpha
	\right) \le \alpha$ for any $\alpha \in [0,1]$. 
	When the null hypothesis $\bm{\tau} \in \mathcal{H}$ is true, for any $\alpha \in [0,1]$, we have
	$
	\Pr\left(
	\sup_{\bm{\delta} \in \mathcal{H}} p_{\bm{Z}, \bm{\delta}} \le \alpha
	\right)
	\le \Pr\left(
	p_{\bm{Z}, \bm{\tau}} \le \alpha
	\right)
	\le \alpha.
	$
	Thus, $\sup_{\bm{\delta} \in \mathcal{H}} p_{\bm{Z}, \bm{\delta}}$, 
	is a valid $p$-value for testing the null hypothesis of $\bm{\tau} \in \mathcal{H}$. 
    Consequently, the upper bound on the right hand side of \eqref{eq:sup_p_val} is also a valid $p$-value for testing the null of $\bm{\tau} \in \mathcal{H}$. 
    Therefore, Theorem \ref{thm:general_null} holds. 
\end{proof}

To prove Theorem \ref{thm:test_Hkc}, we need the following three lemmas. 

\begin{lemma}\label{lemma:rank_score_prop}
	When the treatment assignment $\bm{Z}$ is exchangeable as in Definition \ref{def:exchange} and independent of the ordering of the units, 
	the rank score statistic in Definition \ref{def:rank_score} using ``first'' method for ties is both effect increasing as in Definition \ref{def:effect_increase} and distribution free as in Definition \ref{def:dist_free}. 
\end{lemma}

\begin{proof}[Proof of Lemma \ref{lemma:rank_score_prop}]
	Lemma \ref{lemma:rank_score_prop} follows immediately from Proposition \ref{prop:stat_property_rank}. 
\end{proof}

\begin{lemma}\label{lemma:inf_t_over_Hk0}
	For any $1\le k\le n$, $\bm{z} \in \{0,1\}^n$ and $\bm{y} \in \mathbb{R}^n$, 
	let $m = \sum_{i=1}^n z_i$,  and 
	$\mathcal{I}_k$ be the set of indices for the largest $\min(n-k,m)$ coordinates of $\bm{y}$ with corresponding $z_i$'s being 1, i.e., 
	\begin{align*}
	    \mathcal{I}_k \subset \{i: z_i = 1, 1\le i \le n\}, \ \ 
	    |\mathcal{I}_k| = \min(n-k,m), \ \ 
	    \inf_{i \in \mathcal{I}_k} \rank_i(\bm{y}) > \sup_{i: i \notin \mathcal{I}_k, z_i=1} \rank_i(\bm{y}), 
	\end{align*}
	where $\rank(\cdot)$ uses the ``first'' method for ties.
	Then for any rank score statistic $t(\cdot, \cdot)$,
	we have
	$\inf_{\bm{\delta} \in \mathcal{H}_{k,0}} t(\bm{z},\bm{y} - \bm{z} \circ \bm{\delta}) = t(\bm{z}, \bm{y} - \bm{z} \circ \bm{\xi} )$, where $\bm{\xi} = (\xi_1, \xi_2, \ldots, \xi_n)'$ is defined as $\xi_i = \infty$ if $i\in \mathcal{I}_k$ and $0$ otherwise.
\end{lemma}

\begin{proof}[Proof of Lemma \ref{lemma:inf_t_over_Hk0}]	
	To ease the description, define $l=\min(n-k, m)$ and $\mathcal{T} = \{i: z_i = 1, 1\le i \le n\}$. 
	Let $r_1 < r_2 < \ldots < r_m$ be the ranks of $y_i$'s with indices in $\mathcal{T}$, i.e., 
	$\{r_1, r_2, \ldots, r_m\} = \{\rank_i(\bm{y}): i \in \mathcal{T}\}$. 
	
	First, we show that for any $\bm{\delta}\in \mathcal{H}_{k,0}$, there exists a set $\mathcal{J}_k \subset \mathcal{T}$ of cardinality $l$ such that 
	$t(\bm{z}, \bm{y} - \bm{z} \circ \bm{\delta}) \ge t(\bm{z}, \bm{y} - \bm{z} \circ \bm{\eta})$, where $\bm{\eta} = (\eta_1, \eta_2, \ldots, \eta_n)'$  satisfies $\eta_i = \infty$ if $i \in  \mathcal{J}_k$ and zero otherwise. 
	For any $\bm{\delta} \in \mathcal{H}_{k,0}$, 
	let $\mathcal{J}_k$ be the set of indices for the largest $l$ coordinates of $\bm{\delta}$ with indices in $\mathcal{T}$, i.e., $\mathcal{J}_k\subset \mathcal{T}$, $|\mathcal{J}_k| = l$ and $\min_{i \in \mathcal{J}_k} \delta_i > \max_{i: i \notin \mathcal{J}_k, z_i = 1}\delta_i$. Then by the definition of $\mathcal{H}_{k,0}$ in Section \ref{sec:test_quantile}, for any $i \in \mathcal{T} \setminus \mathcal{J}_k$, we must have $\delta_i \le 0$. 
	Define a vector $\bm{\eta} = (\eta_1, \eta_2, \ldots, \eta_n)'$ with $\eta_i = \infty$ if $i \in \mathcal{J}_k$ and $0$ otherwise. 
	For any $i \in \mathcal{T}$, we have $\delta_i \le \eta_i$ and thus $y_i - z_i \delta_i \ge y_i - z_i \eta_i$. For any $i \notin \mathcal{T}$, we have $y_i - z_i \delta_i = y_i =  y_i - z_i \eta_i$. 
	Because the test statistic is effect increasing, from Definition \ref{def:effect_increase}, we have 
	$
	t(\bm{z}, \bm{y} - \bm{z} \circ \bm{\delta}) \ge t(\bm{z}, \bm{y} - \bm{z} \circ \bm{\eta}). 
	$
	
	Second,  for any $\mathcal{J}_k \subset \mathcal{T}$ of cardinality  $l$, we calculate the value of $t(\bm{z}, \bm{y} - \bm{z} \circ \bm{\eta} )$,  where $\bm{\eta} = (\eta_1, \eta_2, \ldots, \eta_n)'$  satisfies $\eta_i = \infty$ if $i \in  \mathcal{J}_k$ and zero otherwise. 
	Recall that $\{r_1 < r_2 < \ldots < r_m\} = \{\rank_i(\bm{y}): i \in \mathcal{T}\}$. 
	Let 
	$r_{j_1} < r_{j_2} < \ldots, r_{j_{m-l}}$ be the ranks of the $y_i$'s with indices in $\mathcal{T}\setminus \mathcal{J}_k$,  
	and 	
	$r_{j_{m-l+1}} < r_{j_{m-l+2}}< \ldots < r_{j_m}$ be the ranks of the $y_i$'s with indices in $\mathcal{J}_k$, i.e., 
	\begin{align*}
	\{\rank_i(\bm{y}): i \in \mathcal{T} \setminus \mathcal{J}_k \} = \{ r_{j_1}, r_{j_2}, \ldots, r_{j_{m-l}} \}, 
	\quad 
	\{\rank_i(\bm{y}): i \in \mathcal{J}_k \} = \{ r_{j_{m-l+1}}, r_{j_{m-l+2}}, \ldots, r_{j_m} \}
	\end{align*}
	where $\{j_1, j_2, \ldots, j_m\}$ is a permutation of $\{1, 2, \ldots, m\}.$ 
	By the definition of $\mathcal{J}_k$ and $\bm{\eta}$, $y_i - z_i\eta_i = -\infty$ for $i \in \mathcal{J}_k$, and thus the ranks of $y_i-z_i \eta_i$'s with $i\in \mathcal{J}_k$ must become $\{1, 2, \ldots, l\}$. 
	For each $i\in \mathcal{T}\setminus \mathcal{J}_k$ with $\rank_i(\bm{y}) = r_{j_p}$ for some $1\le p \le m-l$, there are $(m-j_p)$ coordinates of $\bm{y}$ with indices in $\mathcal{T}$ having larger ranks than $y_i$. However, for coordinates of $\bm{y}-\bm{z}\circ\bm{\eta}$ with indices in $\mathcal{T}$, there are only $m-l-p$ of them ranked higher than $y_i-z_i\eta_i = y_i$. 
	Note that $y_j - z_j \eta_j = y_j$ for $j \notin \mathcal{T}$. 
	These imply that the rank of $y_i-z_i\eta_i$ increases by $(m-j_p) - (m-l-p)=l+p-j_p$ compared to the  rank of the corresponding $y_i$. 
	Consequently, 
	\begin{align}\label{eq:rank_y_eta}
	\{\rank_i(\bm{y}-\bm{z}\circ \bm{\eta}): i \in \mathcal{J}_k \} & = \{ 1, 2, \ldots, l \}, \\ 
	\{\rank_i(\bm{y}-\bm{z}\circ \bm{\eta}): i \in \mathcal{T} \setminus \mathcal{J}_k \} 
	& = \{ r_{j_1}+l+1 - j_1, \ r_{j_2}+l+2-j_2, \ \ldots, \  r_{j_{m-l}}+l +(m-l) - j_{m-l} \}. \nonumber
	\end{align}
	Therefore, the value of the statistic $t(\bm{z}, \bm{y} - \bm{z} \circ \bm{\eta} )$ is 
	\begin{align}\label{eq:value_t_eta_J}
	t(\bm{z}, \bm{y} - \bm{z} \circ \bm{\eta}) & = \sum_{i\in \mathcal{T}}  \phi\left\{ \rank_i(\bm{y} - \bm{z} \circ \bm{\eta}) \right\} =  \sum_{i\in \mathcal{J}_k} \phi\left\{ \rank_i(\bm{y} - \bm{z} \circ \bm{\eta}) \right\}
	+ 
	\sum_{i\in \mathcal{T} \setminus \mathcal{J}_k} 
	\phi\left\{ \rank_i(\bm{y} - \bm{z} \circ \bm{\eta}) \right\}
	\nonumber
	\\
	& = \sum_{i=1}^l \phi(i)  + \sum_{p=1}^{m-l} \phi\left( r_{j_p} + l + p - j_p \right). 
	\end{align}

	Third, we calculate the value of $t(\bm{z}, \bm{y}-\bm{z}\circ \bm{\xi})$, which is a  special case of \eqref{eq:value_t_eta_J} with $\mathcal{J}_k$ and $\bm{\eta}$ being $\mathcal{I}_k$ and $\bm{\xi}$. 
	By definition, we can  know that the ranks of $y_i$'s with indices in $\mathcal{T}\setminus \mathcal{I}_k$ becomes $\{r_1, r_2, \ldots, r_{m-l}\}$,
	and 
	the ranks of $y_i$'s with indices in $\mathcal{I}_k$ must be $\{r_{m-l+1}, r_{m-l+2}, \ldots, r_{m}\}$, i.e.,  
	\begin{align*}
	\{\rank_i(\bm{y}): i \in \mathcal{T} \setminus \mathcal{I}_k \} = \{ r_1, r_2, \ldots, r_{m-l} \}, 
	\quad 
	\{\rank_i(\bm{y}): i \in \mathcal{I}_k \} = \{ r_{m-l+1}, r_{m-l+2}, \ldots, r_m \}. 
	\end{align*}
	Using \eqref{eq:value_t_eta_J} in the special case with $\mathcal{J}_k = \mathcal{I}_k$ and $\bm{\eta} = \bm{\xi}$, 
	we have $(j_1, j_2, \ldots, j_{m-l}) = (1, 2, \ldots, m-l)$, and thus 
	\begin{align}\label{eq:value_t_xi_I}
	t(\bm{z}, \bm{y} - \bm{z} \circ \bm{\xi}) 
	& = \sum_{i=1}^l \phi(i)  + \sum_{p=1}^{m-l} \phi\left( r_{p} + l + p - p \right)
	= 
	\sum_{i=1}^l \phi( i )  + \sum_{p=1}^{m-l} \phi\left( r_p + l  \right).  
	\end{align}
	
	Fourth, we prove that, for any $\mathcal{J}_k \subset \mathcal{T}$ of cardinality $l$ and $\bm{\eta} = (\eta_1, \eta_2, \ldots, \eta_n)'$ with $\eta_i = \infty$ if $i \in  \mathcal{J}_k$ and zero otherwise, 
	$
	t(\bm{z}, \bm{y} - \bm{z} \circ \bm{\eta})  \ge t(\bm{z}, \bm{y} - \bm{z} \circ \bm{\xi}).
	$ 
	From \eqref{eq:value_t_eta_J} and \eqref{eq:value_t_xi_I}, 
	\begin{align*}
	t(\bm{z}, \bm{y} - \bm{z} \circ \bm{\eta}) 
	- 
	t(\bm{z}, \bm{y} - \bm{z} \circ \bm{\xi})
	& = \sum_{i=1}^l \phi( i )  + \sum_{p=1}^{m-l} \phi\left( r_{j_p} + l + p - j_p \right)
	- 
	\sum_{i=1}^l \phi( i )  - \sum_{p=1}^{m-l} \phi\left( r_p + l  \right)\\
	& = \sum_{p=1}^{m-l} \left\{
	\phi\left( r_{j_p} + l + p - j_p \right) - \phi\left( r_p + l  \right)
	\right\}.
	\end{align*}
	By the definition of $(r_{j_1}, r_{j_2}, \ldots, r_{j_p})$, for any $1\le p \le m-l$, we must have $r_{j_p}\ge r_p$, or equivalently $j_p \ge p$. This further implies that for $1\le p\le m-l$, 
	$
	r_{j_p} - r_p = \sum_{i = p+1}^{j_p}(r_i - r_{i-1}) \ge  \sum_{i = p+1}^{j_p}1 = j_p - p.
	$
	Consequently, $r_{j_p}+l+p-j_p \ge r_p+l$. 
	Therefore, we must have 
	\begin{align*}
	t(\bm{z}, \bm{y} - \bm{z} \circ \bm{\eta}) 
	- 
	t(\bm{z}, \bm{y} - \bm{z} \circ \bm{\xi})
	& = 
	\sum_{p=1}^{m-l} \left\{
	\phi\left( r_{j_p} + l + p - j_p \right) - \phi\left( r_p + l  \right)
	\right\}
	\ge 0. 
	\end{align*}
	
	From the above, we can derive Lemma \ref{lemma:inf_t_over_Hk0}. 
\end{proof}

\begin{lemma}\label{lemma:inf_t_over_Hkc}
	For any $1\le k\le n$, $\bm{z} \in \{0,1\}^n$ and $\bm{y} \in \mathbb{R}^n$, 
	let $m = \sum_{i=1}^n z_i$,  and 
	$\mathcal{I}_k$ be the set of indices for the largest $\min(n-k,m)$ coordinates of $\bm{y}$ with corresponding $z_i$'s being 1, i.e., 
	\begin{align*}
	    \mathcal{I}_k \subset \{i: z_i = 1, 1\le i \le n\}, \ \ 
	    |\mathcal{I}_k| = \min(n-k,m), \ \      
	    \inf_{i \in \mathcal{I}_k} \rank_i(\bm{y}) > \sup_{i: i \notin \mathcal{I}_k, z_i=1} \rank_i(\bm{y}), 
	\end{align*}
	where $\rank(\cdot)$ uses the ``first'' method for ties.
 Then for any rank score statistic $t(\cdot, \cdot)$
	and any constant $c\in \mathbb{R}$, 
	\begin{itemize}
	    \item[(a)] $\inf_{\bm{\delta} \in \mathcal{H}_{k,c}} t(\bm{z},\bm{y} - \bm{z} \circ \bm{\delta}) = t(\bm{z}, \bm{y} - \bm{z} \circ \bm{\xi}_{k,c} )$, where $\bm{\xi} = (\xi_{1k,c}, \xi_{2k,c}, \ldots, \xi_{nk,c})'$ is defined as $\xi_{ik,c} = \infty$ if $i\in \mathcal{I}_k$ and $c$ otherwise.
	    
	    \item[(b)] $t(\bm{z}, \bm{y} - \bm{z}\circ \bm{\xi}_{k,c}) = t(\bm{z}, \bm{y} - \bm{z}\circ \bm{\zeta}_{k,c})$, 
	    where $\bm{\xi} = (\xi_{1k,c}, \xi_{2k,c}, \ldots, \xi_{nk,c})'$ is defined as $\xi_{ik,c} = \Delta$ if $i\in \mathcal{I}_k$ and $c$ otherwise, 
	   and $\Delta$ is a constant larger than $\max_{j:z_j=1} y_j - \min_{j:z_j=0} y_j$
	\end{itemize}

\end{lemma}

\begin{proof}[Proof of Lemma \ref{lemma:inf_t_over_Hkc}]
	We first prove (a) in  Lemma \ref{lemma:inf_t_over_Hkc}. 
	By definition, $\mathcal{H}_{k,c} = \mathcal{H}_{k,0}+c$. This implies that 
	\begin{align*}
	\inf_{\bm{\delta} \in \mathcal{H}_{k,c}} t(\bm{z},\bm{y} - \bm{z} \circ \bm{\delta}) 
	= 
	\inf_{\bm{\eta} \in \mathcal{H}_{k,0}}t(\bm{z},\bm{y} - \bm{z} \circ (\bm{\eta}+c))
	= 
	\inf_{\bm{\eta} \in \mathcal{H}_{k,0}}t(\bm{z}, (\bm{y} - \bm{z} \circ c) -  \bm{z} \circ \bm{\eta}).
	\end{align*}
	Using Lemma \ref{lemma:inf_t_over_Hk0}, we then have 
	\begin{align*}
	\inf_{\bm{\delta} \in \mathcal{H}_{k,c}} t(\bm{z},\bm{y} - \bm{z} \circ \bm{\delta}) 
	& = 
	\inf_{\bm{\eta} \in \mathcal{H}_{k,0}}t(\bm{z}, (\bm{y} - \bm{z} \circ c) -  \bm{z} \circ \bm{\eta})
	= 
	t(\bm{z}, (\bm{y} - \bm{z} \circ c) -  \bm{z} \circ \bm{\zeta})\\
	& = t(\bm{z}, \bm{y} -  \bm{z} \circ (\bm{\zeta} + c) ),
	\end{align*}
	where $\bm{\zeta} = (\zeta_1, \zeta_2, \ldots, \zeta_n)'$ is defined as $\zeta_i = \infty$ if $i\in \mathcal{I}_k$ and $0$ otherwise. 
	Note that $\bm{\xi}_{k,c} = \bm{\zeta}+c$. We can then derive Lemma \ref{lemma:inf_t_over_Hkc}.

	We then prove (b) in Lemma \ref{lemma:inf_t_over_Hkc}.  
	Let $\bm{\beta} = \bm{y} - \bm{z}\circ \bm{\xi}_{k,c}$ and 
	$\bm{\gamma} = \bm{y} - \bm{z}\circ \bm{\zeta}_{k,c}$. 
	We can verify that 
	$
	t(\bm{z}, \bm{\beta}) = 
	\sum_{i=1}^n \phi(i) - \sum_{i=1}^n (1-z_i) \phi(\rank_i(\bm{\beta}))
	$
	and 
	$
	t(\bm{z}, \bm{\gamma}) = 
	\sum_{i=1}^n \phi(i) - \sum_{i=1}^n (1-z_i) \phi(\rank_i(\bm{\gamma})).
	$
	Thus, to prove (b), it suffices to prove that $\rank_i(\bm{\beta}) = \rank_i(\bm{\gamma})$ for any $i$ with $z_i=0$. 
	Let $\mathcal{T} = \{l: z_l = 1\}$. 
	Given $i\notin \mathcal{T}$, we count the number of coordinates of $\bm{\beta}$ and $\bm{\gamma}$ that have smaller ranks then the corresponding $i$th coordinates, respectively. 
	First, for any $j\in \mathcal{I}_k$, by construction, $\beta_j = y_j - \infty = -\infty < y_i =  \beta_i$ and 
	$\gamma_j = y_j - \Delta < \min_{l:z_l=0} y_l < y_i = \gamma_i$. These imply that 
	$\rank_j(\bm{\beta}) < \rank_i(\bm{\beta})$ and 
	$\rank_j(\bm{\gamma}) < \rank_i(\bm{\gamma})$ for $j \in \mathcal{I}_k$. 
	Second, for any $j \in \mathcal{T} \setminus \mathcal{I}_k$, 
	by construction, $\beta_j = y_j - z_j c = \gamma_j$, and $\beta_i = y_i = \gamma_i$. 
	Thus, $\rank_j(\bm{\beta}) < \rank_i(\bm{\beta})$ if and only if $\rank_j(\bm{\gamma}) < \rank_i(\bm{\gamma})$ for $j \in \mathcal{T} \setminus \mathcal{I}_k$. 
	From the above, we must have $\rank_i(\bm{\beta}) = \rank_i(\bm{\gamma})$ for $i\notin \mathcal{T}$. Thus, (b) in Lemma \ref{lemma:inf_t_over_Hkc} holds. 
	
	Therefore, Lemma \ref{lemma:inf_t_over_Hkc} holds.  
\end{proof}

\begin{proof}[\bf Proof of Theorem \ref{thm:test_Hkc}]
    First, from Theorem \ref{thm:general_null}, $p_{\bm{Z}, k, c} \equiv 
	\sup_{\bm{\delta} \in \mathcal{H}_{k,c}} p_{\bm{Z}, \bm{\delta}}$ is a valid $p$-value for testing the null hypothesis $H_{k, c}$ in \eqref{eq:H_nc_k}. 
	Second, from \eqref{eq:sup_p_val} and Lemma \ref{lemma:inf_t_over_Hkc}, $t(\bm{Z}, \bm{Y}_{\bm{Z}, \bm{\delta}}(0))$ can achieve its infimum at some $\bm{\delta} \in \mathcal{H}_{k,c}$, 
	and thus the equivalent forms of $p_{\bm{Z}, k, c}$ in \eqref{eq:p_kc} holds. 
	Therefore, Theorem \ref{thm:test_Hkc} holds. 
\end{proof}

\begin{proof}[\bf Proof of Theorem \ref{thm:con_interval_nc}]
	First, because $\mathcal{H}_{k,c} \subset \mathcal{H}_{k, \overline{c}}$ for any $c\le \overline{c}$, from \eqref{eq:p_kc}, we have 
	\begin{align*}
	p_{\bm{Z}, k, c} = 
	\sup_{\bm{\delta} \in \mathcal{H}_{k,c}} p_{\bm{Z}, \bm{\delta}}  
	\le 
	\sup_{\bm{\delta} \in \mathcal{H}_{k,\overline{c}}} p_{\bm{Z}, \bm{\delta}}
	= 
	p_{\bm{Z}, k, \overline{c}}.
	\end{align*}
	Therefore, the $p$-value $p_{\bm{z}, k, c}$, viewed as a function of $c$, is monotone increasing for any fixed $\bm{z}$ and $k$. 
	
	Second, the coverage probability of the set $\{c: p_{\bm{Z}, k, c} > \alpha, c \in \mathbb{R} \}$ is 
	\begin{align*}
	\Pr\left(
	\tau_{(k)} \in \{c: p_{\bm{Z}, k, c} > \alpha, c \in \mathbb{R} \}
	\right)
	& = 
	\Pr\left(
	p_{\bm{Z}, k, \tau_{(k)}} > \alpha
	\right)
	= 
	1 - \Pr\left(
	p_{\bm{Z}, k, \tau_{(k)}} \le \alpha
	\right)
	\ge 
	1 - \alpha,
	\end{align*}
	where the last inequality holds because $p_{\bm{Z}, k, \tau_{(k)}}$ is a valid $p$-value for testing the null hypothesis of  $\tau_{(k)}\le \tau_{(k)}$ that always holds by definition. 
	Therefore, $\{c: p_{\bm{Z}, k, c} > \alpha, c \in \mathbb{R} \}$ is a $1-\alpha$ confidence set for $\tau_{(k)}$. 
	
	Third, because  $p_{\bm{z}, k, c}$ is %
	increasing 
	in $c$, the confidence set $\{c: p_{\bm{Z}, k, c} > \alpha, c \in \mathbb{R} \}$ must have the form of $(\underline{c}, \infty)$ or $[\underline{c}, \infty)$ with $\underline{c} = \inf \{c: p_{\bm{Z}, k, c} > \alpha, c \in \mathbb{R} \}$. 
	
	From the above, Theorem \ref{thm:con_interval_nc} holds. 
\end{proof}

\begin{proof}[\bf Proof of Corollary \ref{cor:infer_number_larger_threshold}]
	First, (a) in Corollary \ref{cor:infer_number_larger_threshold} follows directly from Theorem \ref{thm:test_Hkc}. 
	Second,  for any $1\le \underline{k} \le k \le n$, because $\mathcal{H}_{k,c} \subset \mathcal{H}_{\underline{k},c}$, we have  
	$p_{\bm{Z}, k, c} \equiv 
	\sup_{\bm{\delta} \in \mathcal{H}_{k,c}} p_{\bm{Z}, \bm{\delta}}
	\le 
	\sup_{\bm{\delta} \in \mathcal{H}_{\underline{k},c}} p_{\bm{Z}, \bm{\delta}}
	\equiv 
	p_{\bm{Z}, \underline{k}, c}.
	$
	This implies (b) in Corollary \ref{cor:infer_number_larger_threshold}. 
	Third, the coverage probability of the set 
	$
	\{n-k: p_{\bm{Z}, k, c} > \alpha, 0\le k \le n \}, 
	$
	satisfies
	\begin{align*}
	\Pr\left(
	n(c) \in \{n-k: p_{\bm{Z}, k, c} > \alpha, 0\le k \le n \}
	\right)
	& = 
	\Pr\left(
	p_{\bm{Z}, n-n(c), c} > \alpha
	\right)
	=
	1  - \Pr\left(
	p_{\bm{Z}, n-n(c), c} \le \alpha
	\right)\\
	& \ge 1 - \alpha,
	\end{align*} 
	where the last inequality holds because, from Theorem \ref{thm:test_Hkc} and \eqref{eq:H_nc_k_equ}, $p_{\bm{Z}, n-n(c), c}$ is a valid $p$-value for testing the null hypothesis $H_{n-n(c), c}$ of $n(c) \le n - \{n-n(c)\}$ that always holds by definition. 
	Moreover, because $p_{\bm{Z}, k, c}$ is %
	decreasing 
	in $k$, the confidence set 
	$
	\{n-k: p_{\bm{Z}, k, c} > \alpha, 0\le k \le n \}
	$
	must have the form of 
	$\{j: n -\overline{k}  \le j \le n\}$
	with $\overline{k} = \sup \{k: p_{\bm{Z}, k, c} > \alpha, 0\le k \le n \}$.
	Therefore, (c) in Corollary \ref{cor:infer_number_larger_threshold} holds. 
	From the above, Corollary \ref{cor:infer_number_larger_threshold} holds. 
\end{proof}

\begin{proof}[\bf Proof of Theorem \ref{thm:simultaneous_valid}]
	First, we prove \eqref{eq:con_tau_k_equi_form}. 
	We first show that if the left hand size of \eqref{eq:con_tau_k_equi_form} holds, then the right hand side must also hold. 
	For any $1\le k \le n$, 
	suppose $\tau_{(k)}$ is in $\{c: p_{\bm{Z}, k, c} > \alpha, c \in \mathbb{R} \}$. 
	This implies that $p_{\bm{Z}, k, \tau_{(k)}} > \alpha$. 
	For any $c\in \mathbb{R}$ with $p_{\bm{Z}, k, c} \le \alpha$, because the $p$-value $p_{\bm{Z}, k, c}$ is %
	increasing 
	in $c$ from Theorem \eqref{thm:con_interval_nc}, we must have  $\tau_{(k)} > c$, or equivalently $\bm{\tau} \in \mathcal{H}_{k, c}^\c$. Thus, $\bm{\tau} \in 
	\bigcap_{c: p_{\bm{Z}, k, c} \le \alpha } \mathcal{H}_{k, c}^\c$. 
	We then show that if the left hand size of \eqref{eq:con_tau_k_equi_form} fails, then the right hand side must also fail. 
	For any $1\le k \le n$, 
	suppose $\tau_{(k)}$ is not in $\{c: p_{\bm{Z}, k, c} > \alpha, c \in \mathbb{R} \}$. 
	This implies that 
	$p_{\bm{Z}, k, \tau_{(k)}} \le \alpha$. Because $\bm{\tau} \in \mathcal{H}_{k, \tau_{(k)}}$ by definition, we have 
	$
	\bm{\tau} \in  \mathcal{H}_{k, \tau_{(k)}} \subset \bigcup_{c: p_{\bm{Z}, k, c} \le \alpha} \mathcal{H}_{k,c}, 
	$
	or equivalently 
	$
	\bm{\tau} \notin 
	\bigcap_{c: p_{\bm{Z}, k, c} \le \alpha } \mathcal{H}_{k, c}^\c. 
	$
	Therefore, \eqref{eq:con_tau_k_equi_form} holds. 
	
	Second, we prove \eqref{eq:con_n_c_equi_form}. We first show that if the left hand side of \eqref{eq:con_n_c_equi_form} holds, then the right hand side must also hold. 
	For any $c\in \mathbb{R}$, suppose $n(c)$ is in $ \{n-k: p_{\bm{Z}, k, c} > \alpha, 0\le k \le n \}$. 
	This implies that $p_{\bm{Z}, n-n(c), c} > \alpha.$ 
	For any $k$ with $p_{\bm{Z}, k, c} \le \alpha$, because $p_{\bm{Z}, k, c}$ is %
	decreasing 
	in $k$ from Corollary \ref{cor:infer_number_larger_threshold}, we must have $n-n(c) < k$, or equivalently $\bm{\tau} \in \mathcal{H}_{k, c}^\c.$ 
	Thus, $\bm{\tau} \in 
	\bigcap_{k: p_{\bm{Z}, k, c} \le \alpha } \mathcal{H}_{k, c}^\c.$
	We then show that if the left hand side of \eqref{eq:con_n_c_equi_form} fails, then the right hand side must also fail. 
	For any $c\in \mathbb{R}$, suppose $n(c)$ is not in $ \{n-k: p_{\bm{Z}, k, c} > \alpha, 0\le k \le n \}$. 
	This implies that $p_{\bm{Z}, n-n(c), c} \le \alpha.$ 
	Because $\bm{\tau} \in \mathcal{H}_{n-n(c), c}$ by definition, we have 
	$
	\bm{\tau} \in \mathcal{H}_{n-n(c), c} \subset \bigcup_{k: p_{\bm{Z}, k, c} \le \alpha } \mathcal{H}_{k, c},
	$
	or equivalently 
	$
	\bm{\tau} \notin \bigcap_{k: p_{\bm{Z}, k, c} \le \alpha } \mathcal{H}_{k, c}^\c.
	$
	Therefore, \eqref{eq:con_n_c_equi_form} holds. 
	
	Third, we prove that $\bigcap_{k, c: p_{\bm{Z}, k, c} \le \alpha } \mathcal{H}_{k, c}^\c$ is a $1-\alpha$ confidence set for the  true treatment effect $\bm{\tau}$. 
	The coverage probability of the set $\bigcap_{k, c: p_{\bm{Z}, k, c} \le \alpha } \mathcal{H}_{k, c}^\c$ has the following  equivalent forms: 
	\begin{align}\label{eq:cov_set_tau1}
	\Pr\Big(
	\bm{\tau} \in 
	\bigcap_{k, c: p_{\bm{Z}, k, c} \le \alpha } \mathcal{H}_{k, c}^\c 
	\Big) 
	& = 
	1 - \Pr\Big(
	\bm{\tau} \in 
	\bigcup_{k, c: p_{\bm{Z}, k, c} \le \alpha } \mathcal{H}_{k, c}
	\Big). 
	\end{align}
	If $\bm{\tau} \in 
	\bigcup_{k, c: p_{\bm{Z}, k, c} \le \alpha } \mathcal{H}_{k, c}$, then there must exists $(k,c)$ such that $p_{\bm{Z}, k, c} \le \alpha$ and $\bm{\tau} \in \mathcal{H}_{k,c}$. 
	By definition, this implies that 
	$
	p_{\bm{Z}, \bm{\tau}}
	\le 
	\sup_{\bm{\delta} \in \mathcal{H}_{k,c}} p_{\bm{Z}, \bm{\delta}} 
	\equiv 
	p_{\bm{Z}, k, c}
	\le
	\alpha. 
	$
	Thus, from \eqref{eq:cov_set_tau1}, we have  
	\begin{align*}
	\Pr\Big(
	\bm{\tau} \in 
	\bigcap_{k, c: p_{\bm{Z}, k, c} \le \alpha } \mathcal{H}_{k, c}^\c 
	\Big) 
	& = 
	1 - \Pr\Big(
	\bm{\tau} \in 
	\bigcup_{k, c: p_{\bm{Z}, k, c} \le \alpha } \mathcal{H}_{k, c}
	\Big) 
	\ge
	1 - \Pr\left(
	p_{\bm{Z}, \bm{\tau}} \le \alpha
	\right). 
	\end{align*}
	By the validity of usual randomization test, 
	$\Pr(
	p_{\bm{Z}, \bm{\tau}} \le \alpha
	)\le \alpha,
	$
	and therefore, 
	\begin{align*}
	\Pr\Big(
	\bm{\tau} \in 
	\bigcap_{k, c: p_{\bm{Z}, k, c} \le \alpha } \mathcal{H}_{k, c}^\c 
	\Big) 
	\ge 1 - \Pr\left(
	p_{\bm{Z}, \bm{\tau}} \le \alpha
	\right)
	\ge 1-\alpha. 
	\end{align*} 
	
	From the  above, Theorem \ref{thm:simultaneous_valid} holds. 
\end{proof}

\section{Further simulation results}

\subsection{Simulation demonstrating power of different Stephenson rank statistics}

In the main paper's simulations we can see how power changes as a function of $s$.
Extending those results, Figure \ref{fig:simu_power_all_quantile}, built from those same simulation results, shows the power of various Stephenson rank statistics for testing whether each quantile of individual effect is bounded by zero, i.e., $H_{k,0}: \tau_{(k)}\le 0$ over all $1\le k\le n$, under the data generating model with $n=120$ and $(\tau_0, \omega, \rho)$ equal to $(1, 1, -0.9)$. 
(We omit those values of $k$ for which all tests under consideration have zero power.)
From Figure \ref{fig:simu_power_all_quantile}, 
we can see that $s=2$ is preferred for larger quantiles of individual effects, while $s=4$ is preferred for smaller quantiles.

\begin{figure}
    \centering
 	\includegraphics[width=0.8\linewidth]{plots/Simulation_Quantile_Normal_random_n0_legend_20221212.pdf}
    \includegraphics[width=0.4\linewidth]{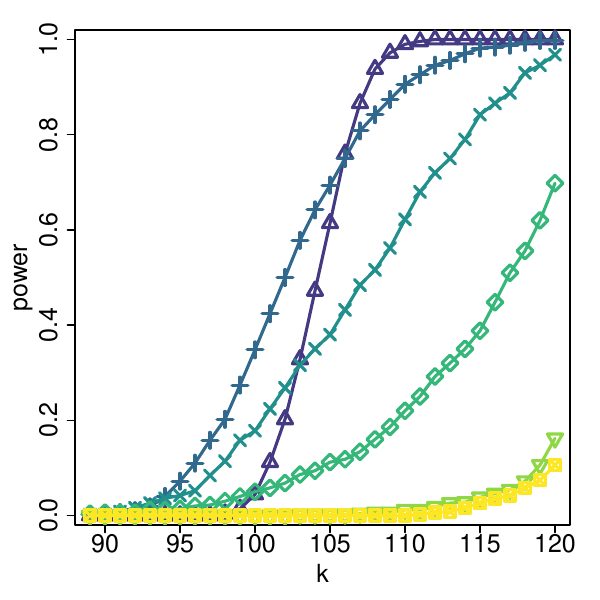}
    \caption{
	Power of various Stephenson rank statistics for testing the null hypothesis of $H_{k,0}: \tau_{(k)}\le 0$ over all $k$, under model \eqref{eq:generate} with $n=120$ and $(\tau_0, \omega, \rho)$ equal to $(1, 1, -0.9)$. 
	}\label{fig:simu_power_all_quantile}
\end{figure}

\subsection{Simulation for heavy-tailed outcomes and individual effects}\label{sec:simu_outlier}

We next conduct a simulation with heavy-tailed outcomes and individual treatment effects. We will demonstrate that, compared to inference on average treatment effects \citep[e.g.,][]{Neyman23a}, 
the proposed inference on quantiles does not require any large-sample approximation and can be more robust to outliers. 
In particular, we consider the case of generally positive constant individual effects with a few extreme negative outliers. 
In the next subsection, we also consider a scenario with heavy-tailed distributions of the potential outcomes under Fisher's null of no effect.

We simulate an experiment with $120$ units, among which two thirds will be randomly assigned to treatment and the remaining to control. 
We assume the treatment increases a certain risk factor by $2$ for $95\%$ of the units, and decreases it by $50$ for the remaining $5\%$ of units.
We simulate the control potential outcomes from the standard normal distribution.
Once generated, all the potential outcomes are fixed.
For the final sample, about 5\% of the units receive large benefit, but the remainder would incur a harmful effect.
The true average treatment effect is negative, indicating an average benefit, even though the majority of units are harmed.

Figure \ref{fig:tau_outlier}(a) shows the histogram of the sampling distribution of the usual difference-in-means estimator. 
Over all the simulated assignments, the difference-in-means estimator is negative about $77\%$ of the time, and its average is close to the true average effect $-0.6$. 
In general, the difference in means estimator would indicate that the treatment is reducing overall risk level and is apparently beneficial. 
However, such a conclusion would potentially be misleading, because the treatment is harmful for most units. 

Figure \ref{fig:tau_outlier}(b) shows the histogram of the $90\%$ lower confidence limit of the number of units with higher risk level under treatment than control (i.e., $n(0)$), based on the Stephenson rank statistic with $s=6$.
Over all simulated assignments, 
the lower confidence limit of $n(0)$ has an average value of 37.5 (about $31\%$ of the units), and can sometimes reach 49 (about $41\%$ of the units).  
Our method reliably detects
that the treatment harms a significant amount of units; such a finding would signal to researchers that the treatment would need further careful investigation despite an apparent average benefit. 

\begin{figure}[htbp]
	\centering
	\begin{subfigure}{.4\textwidth}
		\centering
		\includegraphics[width=0.825\linewidth]{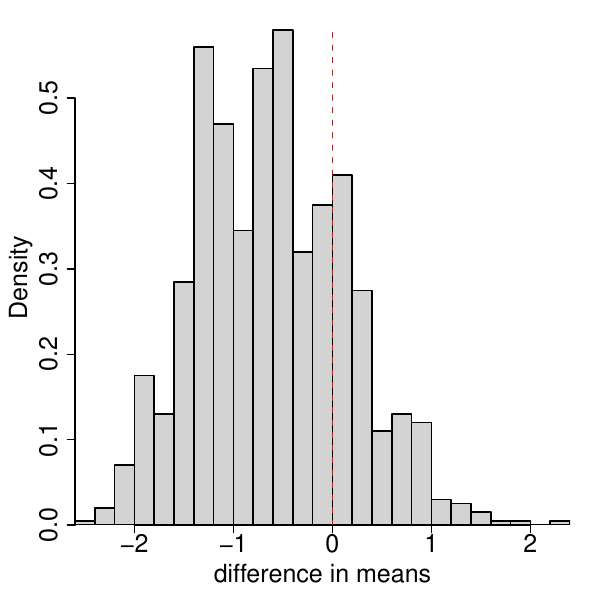}
		\caption{}
	\end{subfigure}%
	\begin{subfigure}{.5\textwidth}
		\centering
		\includegraphics[width=0.667\linewidth]{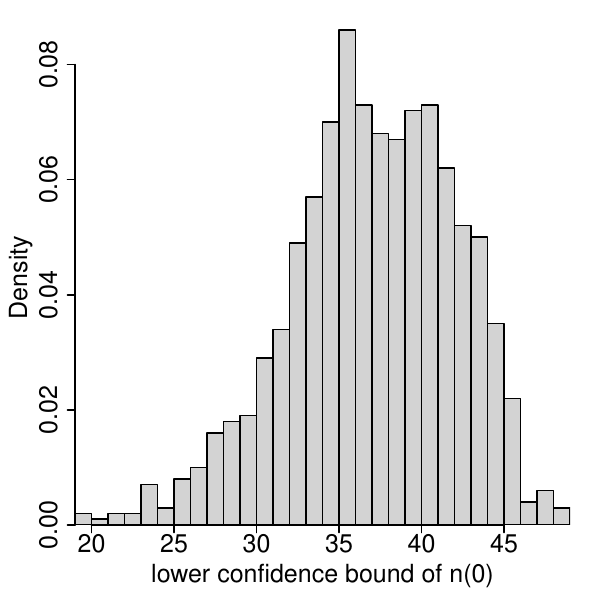}
		\caption{}
	\end{subfigure}
	\caption{
	Simulation when individual treatment effects have extreme values or outliers. 
	(a) shows the histogram of the usual difference-in-means estimator. 
	(b) shows the histogram of the $90\%$ lower confidence limit of the number of units with positive effects using Theorem \ref{thm:con_interval_nc}. 
	}\label{fig:tau_outlier}
\end{figure}

\subsection{Simulation for heavy-tailed outcomes under Fisher's null}\label{sec:simu_heavy_tail}

We simulate the potential outcomes $Y_i(1) = Y_i(0)$'s from a skewed $t$ distribution with degrees of freedom 1.5 and skewing parameter 5 \citep[see, e.g.,][]{skewt_rpackage}, under which all the individual treatment effects are zero, i.e., Fisher's null $H_{\bm{0}}$ holds. 
To mimic the finite population inference, all the potential outcomes are kept fixed once generated. 
We conduct simulation under a CRE with $n=120$ units, where half of the units are assigned to each treatment group. 
Figure \ref{fig:heavy_tail_outcome}(a) shows the normal Q-Q plot of the control potential outcomes, from which we can see that the potential outcomes are right-skewed and heavy-tailed. 
Figure \ref{fig:heavy_tail_outcome}(b) shows the histogram of the difference-in-means estimator under the CRE, which has multiple modes. 
This implies that the large-sample normal approximation as in \citet{Neyman23a} may work poorly with heavy-tailed outcomes. 
Consequently, it is not surprising to see that the $p$-value using the two-sample $t$-statistic and normal approximation does not control the type-I error well, as shown in Figure \ref{fig:heavy_tail_outcome}(c).\footnote{
\citet{Wu:2018aa} and \citet{cohen2020gaussian} suggest using the permutation distribution of the $t$-statistic as the reference distribution for calculating the $p$-value. 
However, 
since the $t$-statistic is in general neither effect increasing nor differential increasing, 
Theorem \ref{thm:null_broader_monotone_control}(a) cannot guarantee the resulting $p$-value to be valid for testing the bounded null. 
}
From the same figure, the Fisher randomization $p$-value using the Stephenson rank sum statistic with $s=10$ is almost uniformly distributed on $(0,1)$. 
Moreover, this $p$-value is also valid for testing the bounded null $H_{\vecle \bm{0}}$.

\begin{figure}[htbp]
	\begin{subfigure}{.33\textwidth}
		\centering
		\includegraphics[width=1\linewidth]{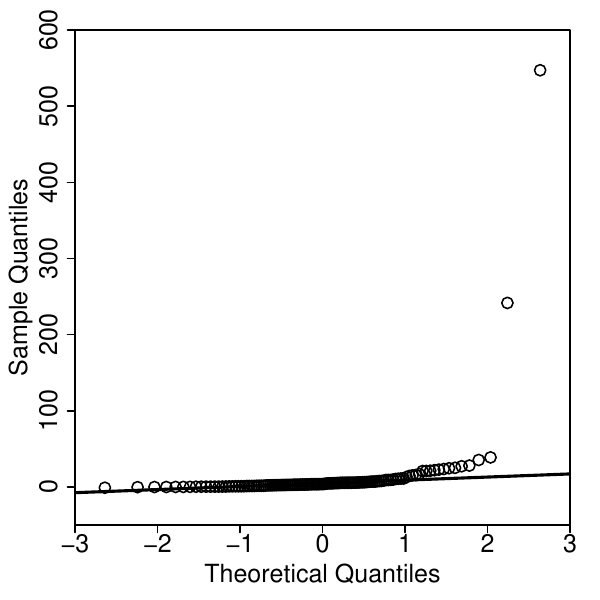}
		\caption{} 
	\end{subfigure}
	\begin{subfigure}{.33\textwidth}
		\centering
		\includegraphics[width=1\linewidth]{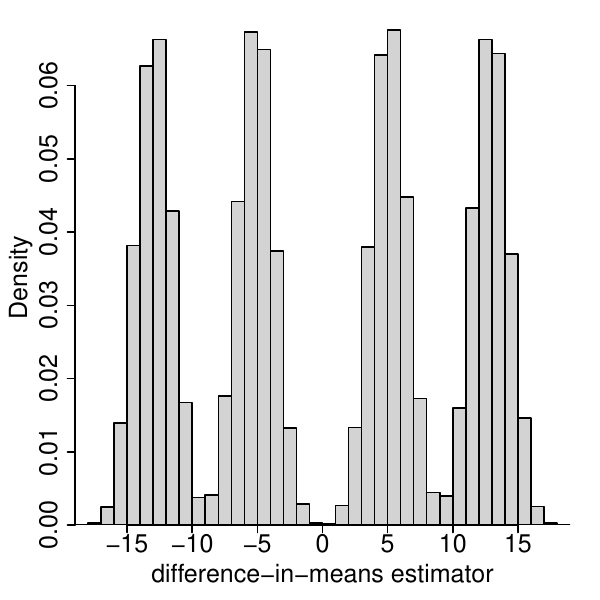}
		\caption{} 
	\end{subfigure}%
	\begin{subfigure}{.33\textwidth}
	\centering
	\includegraphics[width=1\linewidth]{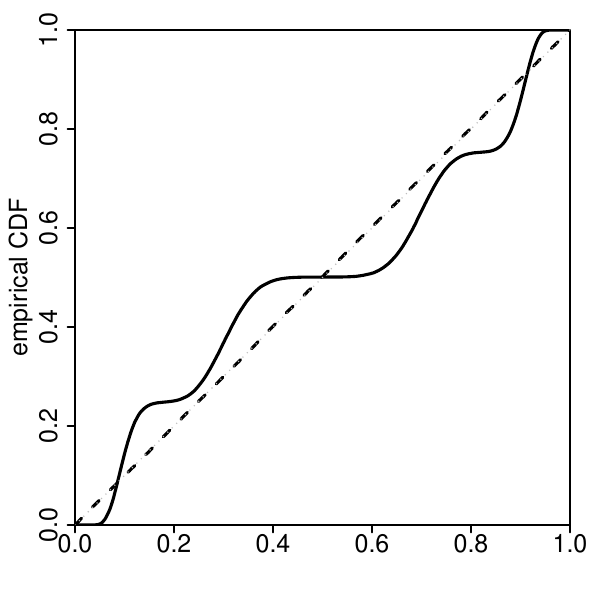}
	\caption{}
	\end{subfigure}
	\caption{
	Simulation under skewed and heavy-tailed outcomes. 
	(a) shows the normal Q-Q plot of the potential outcomes $Y_i(1) = Y_i(0)$'s. 
	(b) shows the histogram of the difference-in-means estimator under a CRE. (c) shows the empirical distribution functions of the $p$-value from \citet{Neyman23a} using $t$-statistic with normal approximation (denoted by the solid line)
	and that from Fisher randomization test using Stephenson rank statistic with $s=10$ (denoted by the dashed line). 
	}\label{fig:heavy_tail_outcome}
\end{figure}

\subsection{Simulation with varying sample size}\label{sec:vary_sample_size}
We conduct the same simulation as in Section \ref{sec:simu_visual}, except that we vary the sample size $n$ from $240$ to $960$. 
Figures \ref{fig:simu_quant_random_240}--\ref{fig:simu_quant_random_960} show the $90\%$ lower confidence limit for all quantiles of individual effects, averaging over all simulated treatment assignment, when the sample size $n$ equals $240$, $480$ and $960$, respectively.  
The implications from these simulation results are mostly the same as that from Section \ref{sec:simu_visual}. 
First, when the average effect is non-positive, the Wilcoxon rank has almost no power to detect any positive individual effects, while the Stephenson rank with larger $s$ tends to give more informative results. 
Second, 
when the average treatment effect is positive, too large an $s$ for Stephenson rank can deteriorate the power, especially when the individual treatment effect is negatively correlated with the
control potential outcome.
It will be interesting to theoretically investigate the role of $s$ for inferring quantiles of individual effects, and we leave it for future study. 

\begin{figure}[htbp]
    \centering
		\includegraphics[width=0.8\linewidth]{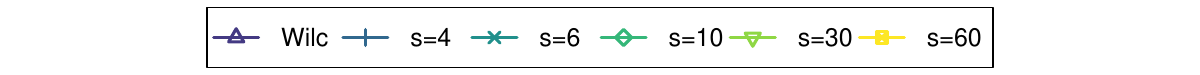}
	\begin{subfigure}{.33\textwidth}
		\centering
		\includegraphics[width=1\linewidth]{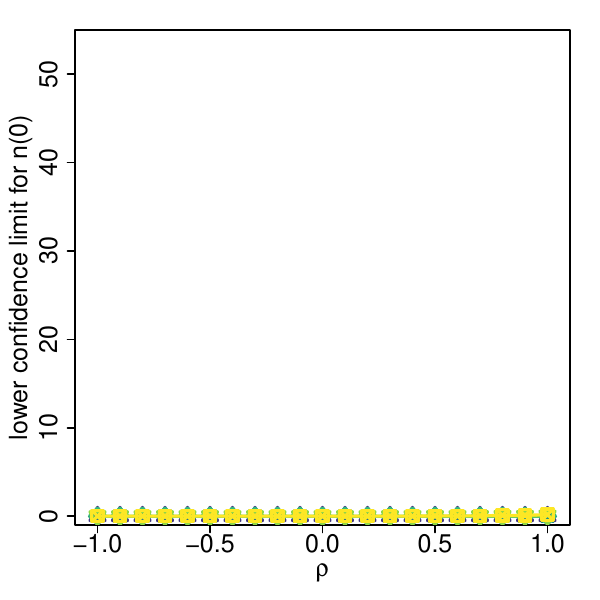}
		\caption{$\tau_0 = -1, \omega = 0.5$}
	\end{subfigure}%
	\begin{subfigure}{.33\textwidth}
		\centering
		\includegraphics[width=1\linewidth]{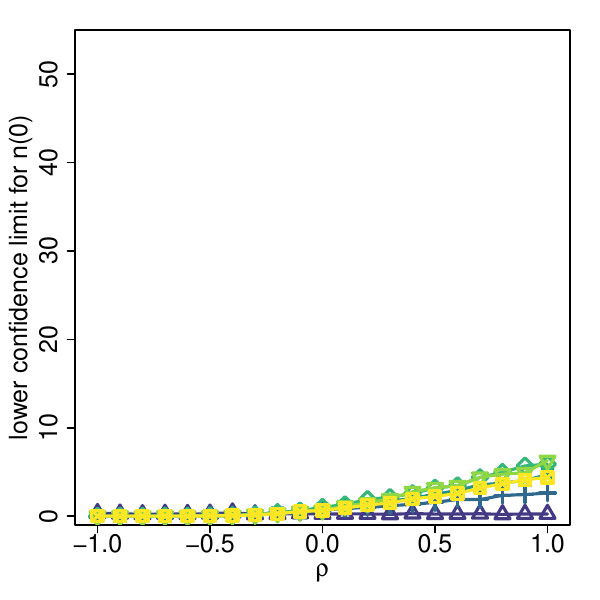}
		\caption{$\tau_0 = 0, \omega = 0.5$}
	\end{subfigure}%
	\begin{subfigure}{.33\textwidth}
		\centering
		\includegraphics[width=1\linewidth]{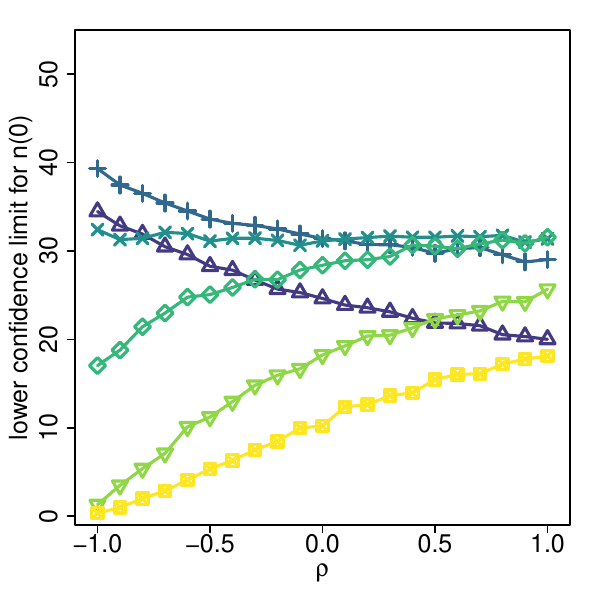}
		\caption{$\tau_0 = 1, \omega = 0.5$}
	\end{subfigure}
	\begin{subfigure}{.33\textwidth}
		\centering
		\includegraphics[width=1\linewidth]{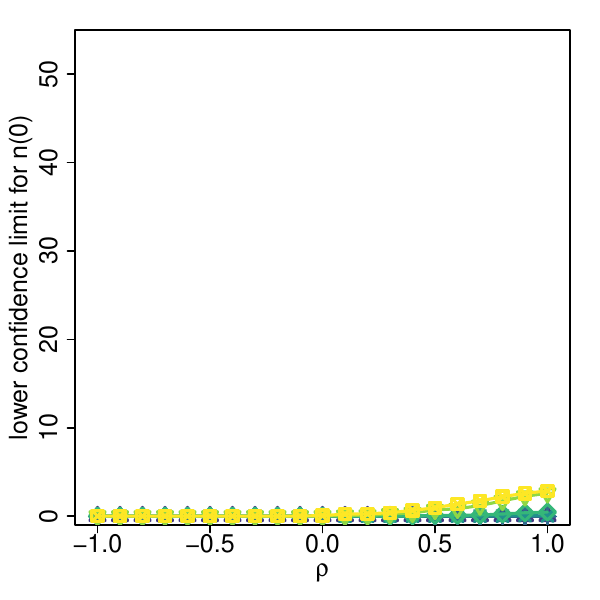}
		\caption{$\tau_0 = -1, \omega = 1$} 
	\end{subfigure}
	\begin{subfigure}{.33\textwidth}
		\centering
		\includegraphics[width=1\linewidth]{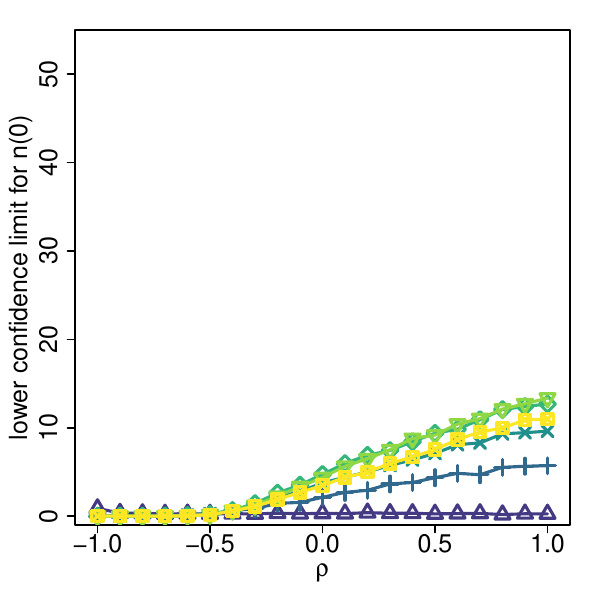}
		\caption{$\tau_0 = 0, \omega = 1$} 
	\end{subfigure}%
	\begin{subfigure}{.33\textwidth}
	\centering
	\includegraphics[width=1\linewidth]{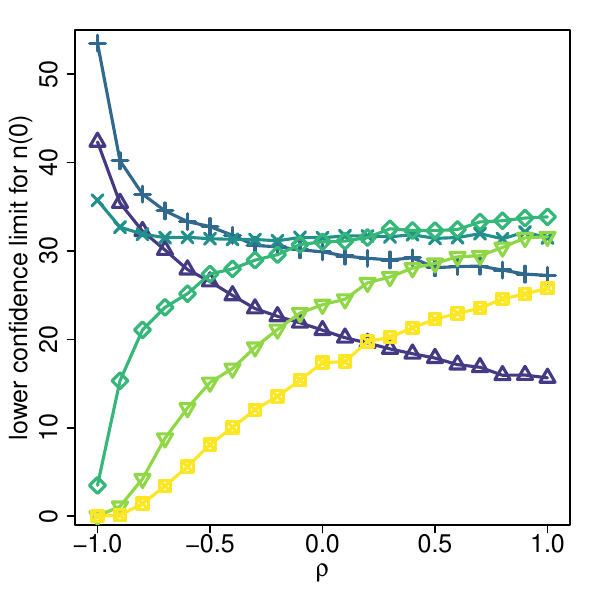}
	\caption{$\tau_0 = 1, \omega = 1$}
	\end{subfigure}
	\caption{Average lower limits of $90\%$ confidence intervals for the number of units with positive effects $n(0)$. 
		The potential outcomes are generated from \eqref{eq:generate} with sample size $n=240$ and different values of $(\tau_0, \omega, \rho)$. 
	}\label{fig:simu_quant_random_240}
\end{figure}

\begin{figure}[htbp]
    \centering
		\includegraphics[width=0.8\linewidth]{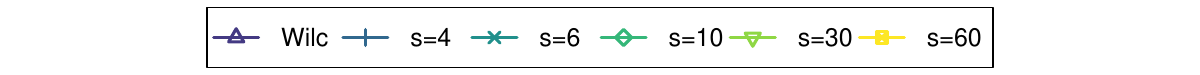}
	\begin{subfigure}{.33\textwidth}
		\centering
		\includegraphics[width=1\linewidth]{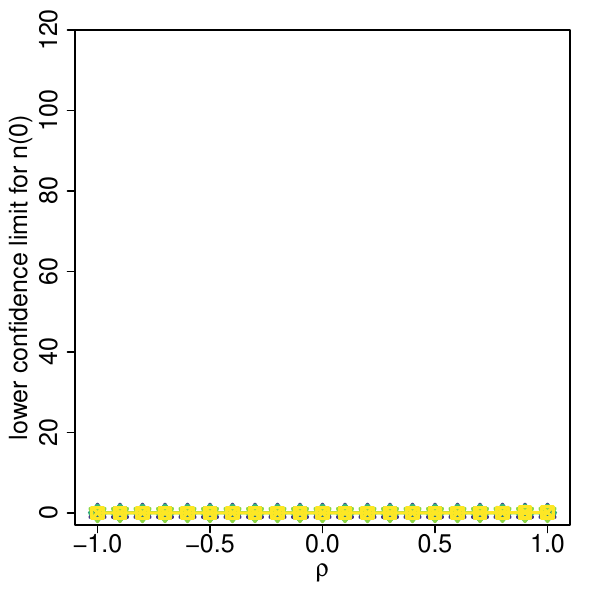}
		\caption{$\tau_0 = -1, \omega = 0.5$}
	\end{subfigure}%
	\begin{subfigure}{.33\textwidth}
		\centering
		\includegraphics[width=1\linewidth]{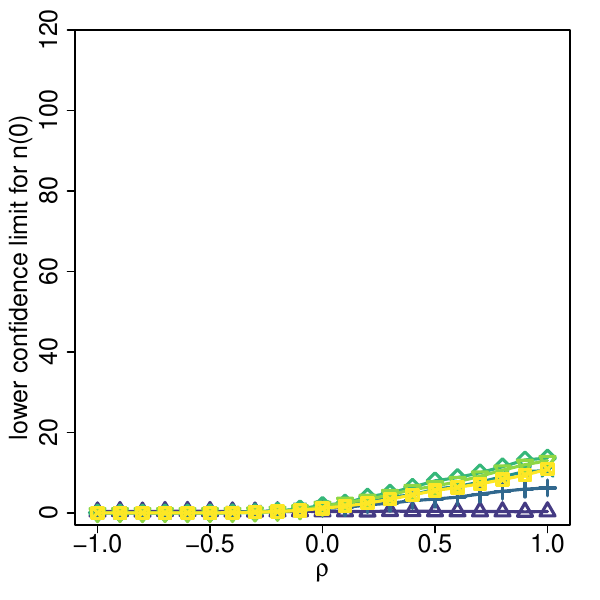}
		\caption{$\tau_0 = 0, \omega = 0.5$}
	\end{subfigure}%
	\begin{subfigure}{.33\textwidth}
		\centering
		\includegraphics[width=1\linewidth]{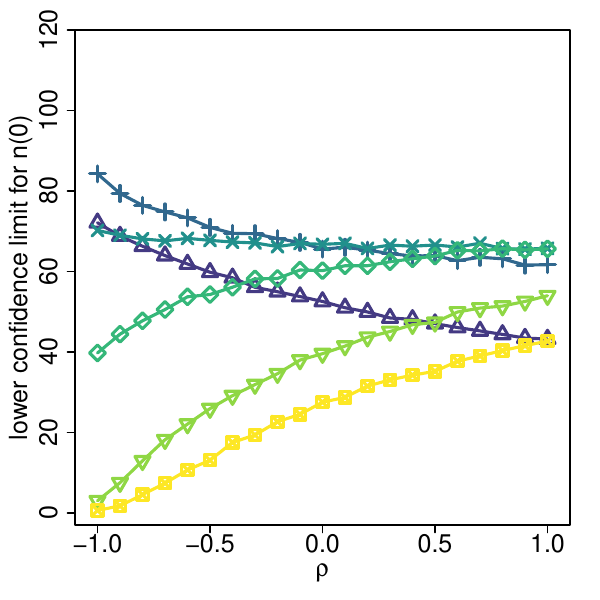}
		\caption{$\tau_0 = 1, \omega = 0.5$}
	\end{subfigure}
	\begin{subfigure}{.33\textwidth}
		\centering
		\includegraphics[width=1\linewidth]{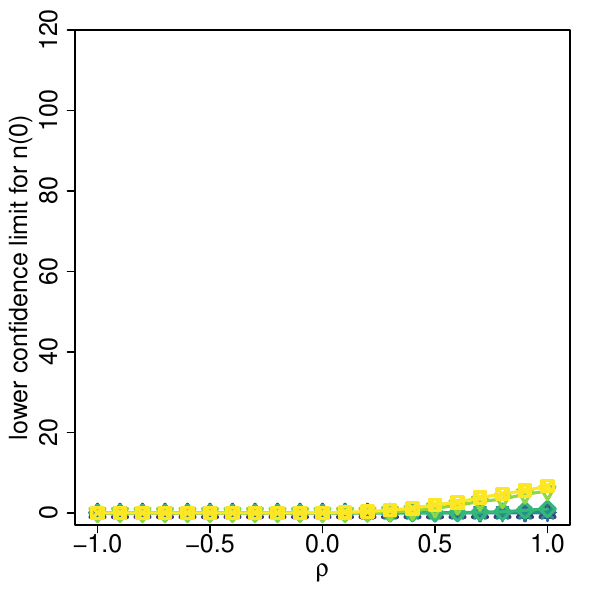}
		\caption{$\tau_0 = -1, \omega = 1$} 
	\end{subfigure}
	\begin{subfigure}{.33\textwidth}
		\centering
		\includegraphics[width=1\linewidth]{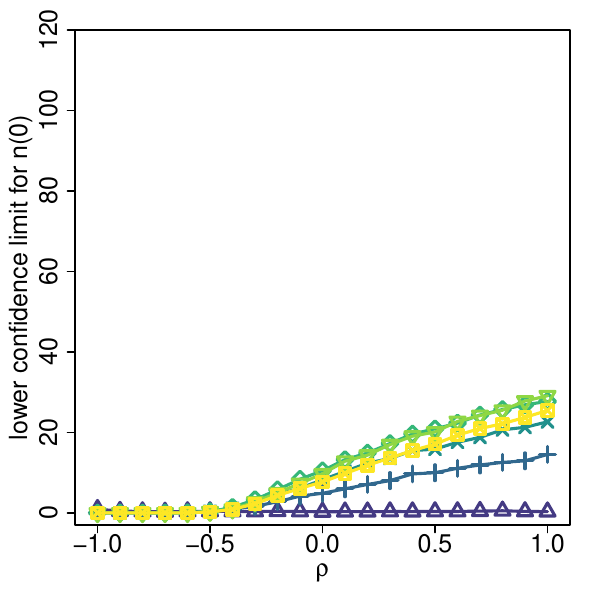}
		\caption{$\tau_0 = 0, \omega = 1$} 
	\end{subfigure}%
	\begin{subfigure}{.33\textwidth}
	\centering
	\includegraphics[width=1\linewidth]{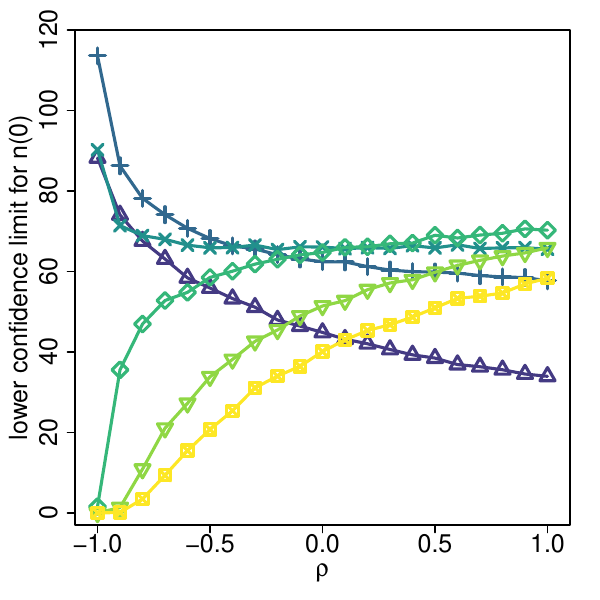}
	\caption{$\tau_0 = 1, \omega = 1$}
	\end{subfigure}
	\caption{Average lower limits of $90\%$ confidence intervals for the number of units with positive effects $n(0)$. 
		The potential outcomes are generated from \eqref{eq:generate} with sample size $n=480$ and different values of $(\tau_0, \omega, \rho)$. 
	}\label{fig:simu_quant_random_480}
\end{figure}

\begin{figure}[htbp]
    \centering
		\includegraphics[width=0.8\linewidth]{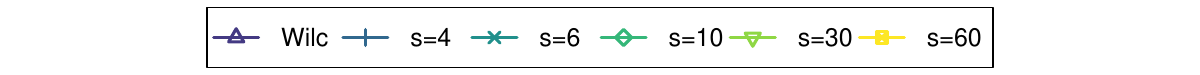}
	\begin{subfigure}{.33\textwidth}
		\centering
		\includegraphics[width=1\linewidth]{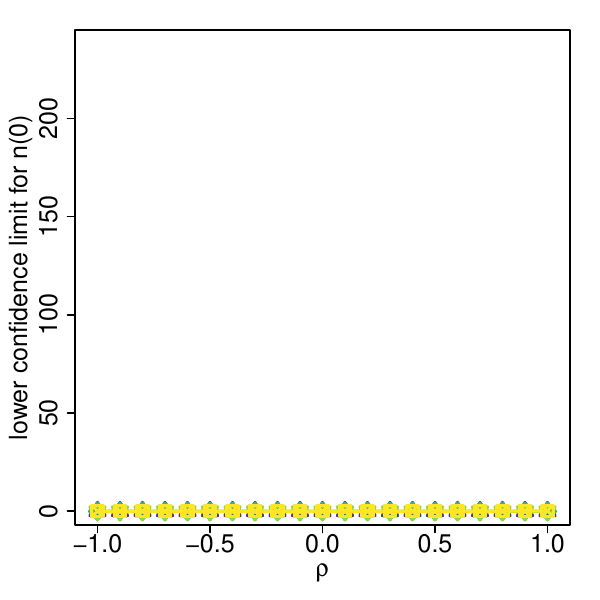}
		\caption{$\tau_0 = -1, \omega = 0.5$}
	\end{subfigure}%
	\begin{subfigure}{.33\textwidth}
		\centering
		\includegraphics[width=1\linewidth]{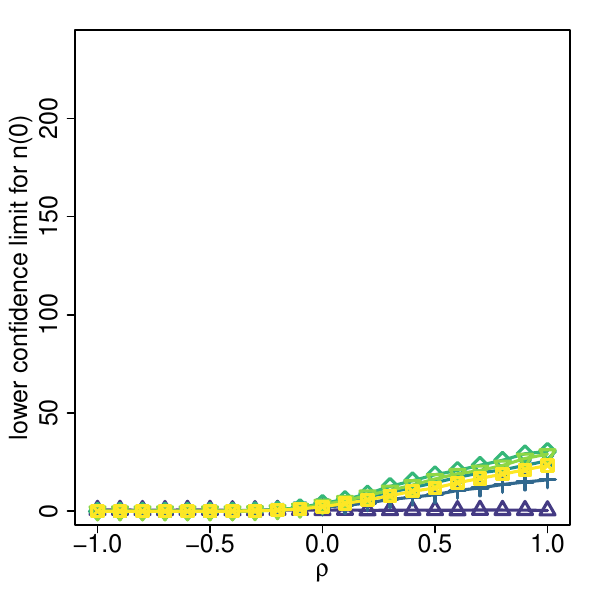}
		\caption{$\tau_0 = 0, \omega = 0.5$}
	\end{subfigure}%
	\begin{subfigure}{.33\textwidth}
		\centering
		\includegraphics[width=1\linewidth]{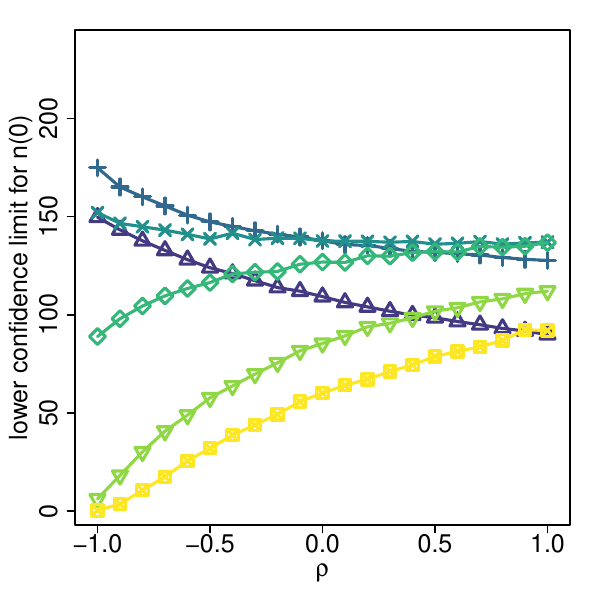}
		\caption{$\tau_0 = 1, \omega = 0.5$}
	\end{subfigure}
	\begin{subfigure}{.33\textwidth}
		\centering
		\includegraphics[width=1\linewidth]{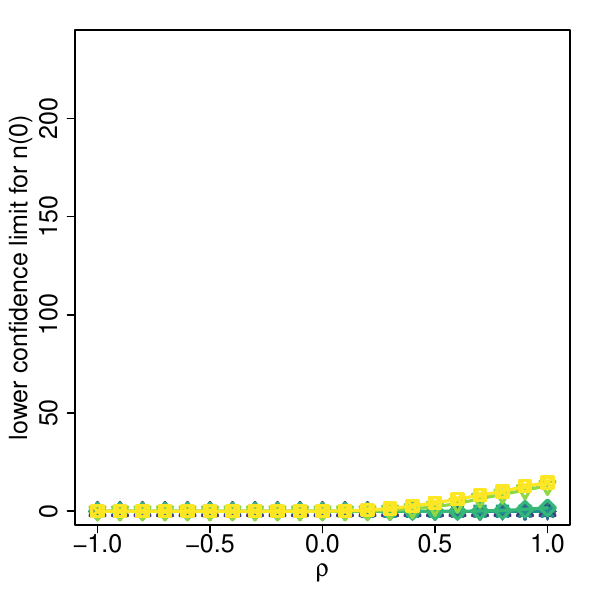}
		\caption{$\tau_0 = -1, \omega = 1$} 
	\end{subfigure}
	\begin{subfigure}{.33\textwidth}
		\centering
		\includegraphics[width=1\linewidth]{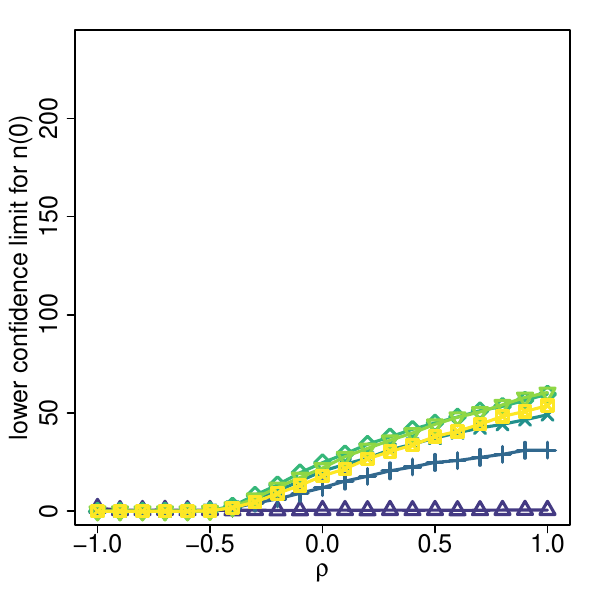}
		\caption{$\tau_0 = 0, \omega = 1$} 
	\end{subfigure}%
	\begin{subfigure}{.33\textwidth}
	\centering
	\includegraphics[width=1\linewidth]{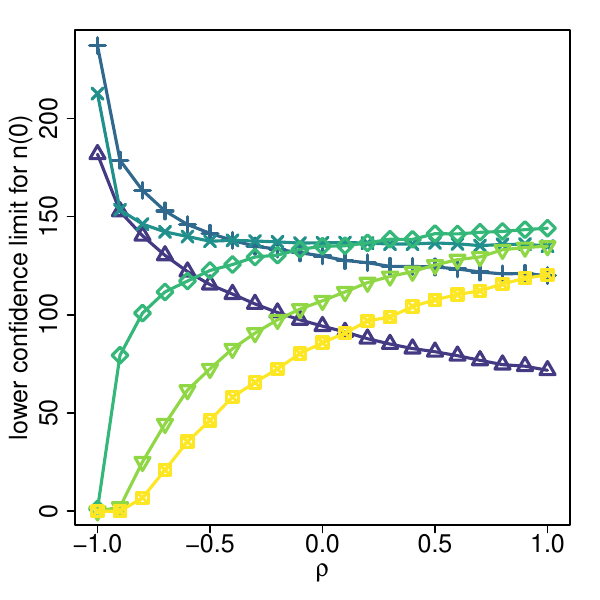}
	\caption{$\tau_0 = 1, \omega = 1$}
	\end{subfigure}
	\caption{Average lower limits of $90\%$ confidence intervals for the number of units with positive effects $n(0)$. 
		The potential outcomes are generated from \eqref{eq:generate} with sample size $n=960$ and different values of $(\tau_0, \omega, \rho)$. 
	}\label{fig:simu_quant_random_960}
\end{figure}

\section{Further Applications}

\subsection{Selecting $s$ for the Professional Development Application}

The main paper outlines a series of power simulations across a family of data generating processes to see how well different $s$ values for the Stephenson rank test performed for detecting different numbers and types of treatment effects.
We next look into these results in more detail.

To run these simulations we use a method in the \texttt{RIQITE} package, \texttt{explore\_stephenson\_s}, that runs the simulations for us.
We provide the empirical distribution of the control-side outcomes as the reference distribution, and specify a variety of treatment impact models that are arguably consistent with the difference in final distributions between treatment and control.
For example, we specify impact models with a true average treatment impact in line with our estimate from the data.

Figure~\ref{fig:s_selection} shows median confidence bounds over different quantiles as a function of $s$.
We grouped the top 100 quantiles into groups of top 10, next 20, and next 30, and plot the median of the median bound across the quantiles in each group (the last 40 are not shown as the intervals were generally not informative).
Generally, the curves are flat, showing that there is some latitude for selecting $s$ across the contexts explored.
Overall, smaller $s$, such as $s=6$, seem preferred. 
See the code file for more specifics on how these results are generated; that file is designed to provide code to make conducting similar simulations on other datasets very straightforward.

\begin{figure}[htbp]
	\centering
		\includegraphics[width=1\linewidth]{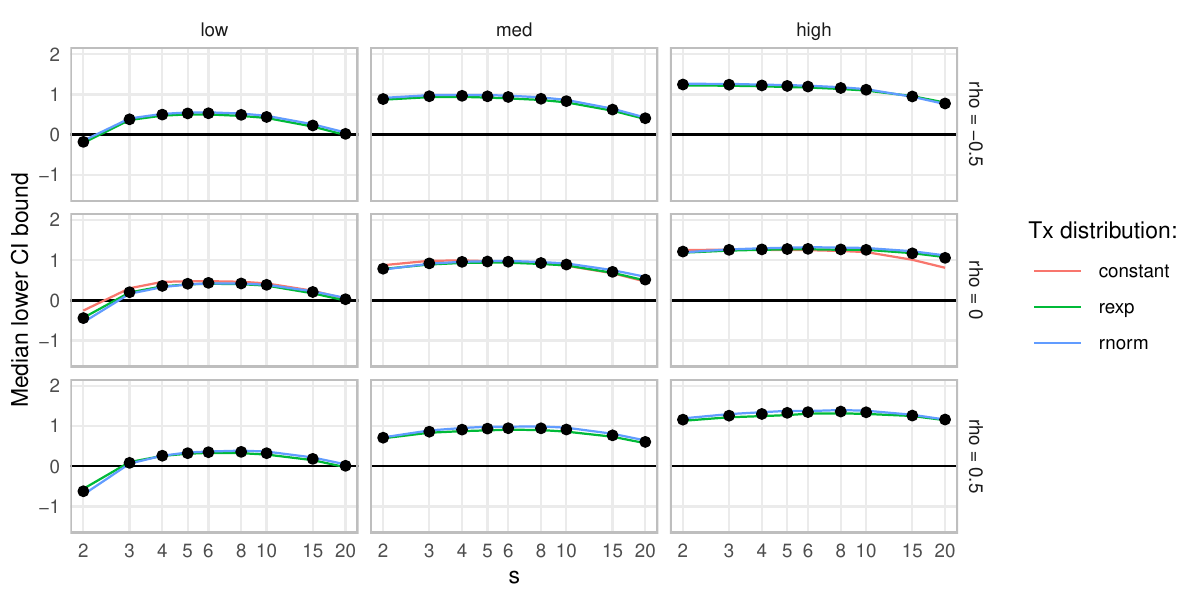}
	\caption{
	Median CI lower bounds over different groups of quantiles and simulation scenarios as a function of $s$ for an exploratory simulation calibrated to the teacher professional development data. From right to left we have top 10 quantiles, next 20, and next 30. Rows correspond to different correlations of $Y_i(0)$ and $\tau_i$.  Black dots give empirical averages across all scenarios considered.	}\label{fig:s_selection}
\end{figure}

We can also look at the number of units found to be significant across the different scenarios and see which $s$ values identify the most units as significant, on average.
Figure~\ref{fig:s_selection_n} shows how $s=8$ seems to maximize, although $s=6$ or $s=10$ perform quite similarly.

\begin{figure}[htbp]
	\centering
		\includegraphics[width=1\linewidth]{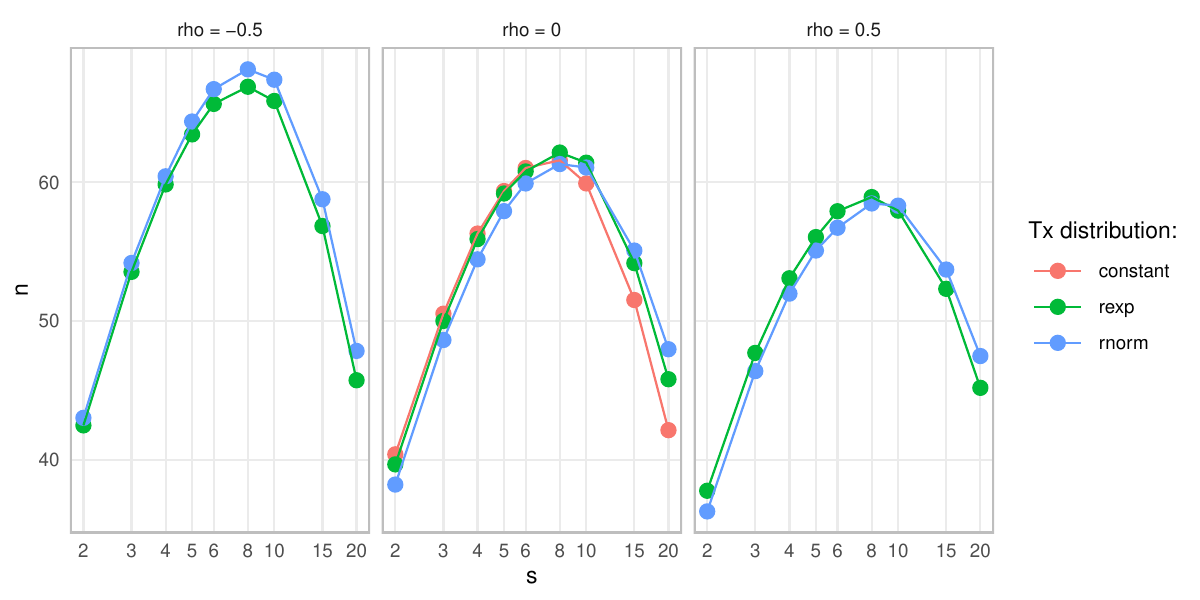}
	\caption{
	Average number of units found to be significant as a function of $s$.}\label{fig:s_selection_n}
\end{figure}

We are facing somewhat of a trade-off: more units being tagged as significant calls for higher $s$, and more informative bounds on the higher end of the distribution calls for slightly lower $s$.

\subsection{Testing monotonicity of an instrumental variable}
\label{sec:test-monot-an}

\begin{figure}
	\centering
	\includegraphics[width=.4\textwidth]{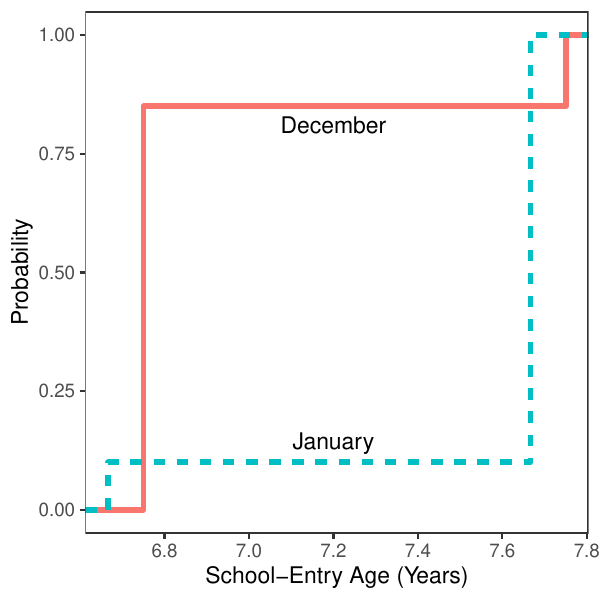}
	\caption{
	Empirical cumulative distribution functions of school-entry age for units born in December and January, respectively \citep{BlackEtAl11a}. 
	}
	\label{fig:IV}
\end{figure}

The assumption that the instrument has monotonic effects on the treatment, though conventionally invoked for identification of instrumental variable (IV) estimates \citep{AngristImbensRubin96a}, is rarely evaluated in empirical applications.  Recently, however, the issue of non-monotonicity has received attention in the active literature on school-entry age, in which numerous studies involve instruments
based on laws regulating entry age  
(\citealt{Aliprantis12a, BaruaLang16a}; for an overview, see \citealt{FioriniStevens14a}). 
Typically, these laws select an arbitrary date of birth before which children are allowed to enter school in a given calendar year. If the cutoff date is January 1, for example, most children born in December will be about 11 months younger when they enter school than children born in January. Due to imperfect compliance with the instrument, however, some fraction of December children may ``redshirt'' and start the following school year, at which time they will be one month \textit{older} than January children who started on time. Unless the December children who redshirt would also have redshirted had they been born in January, monotonicity is violated. That is, the effect of December birth on school-entry age is typically negative, but for a few children it is positive
(an analogous logic holds for January children who start early).

For a simple illustration of how randomization inference can be used to evaluate monotonicity, we re-analyze data on 104,000 children born in December or January %
from \citet{BlackEtAl11a}'s IV study of school-entry age, which have previously been analyzed from a sampling-based perspective by \citet{FioriniStevens14a}\footnote{\rev Table 1 of \citet{BlackEtAl11a} cross-tabulates month of birth and school entry age (early/on-time/late) in terms of proportions, and Table 3 there reports the total number of subjects born in January or December (the ``discontinuity subsample''): 104,023. 
Similar to \citet{FioriniStevens14a}, we assume equal numbers of subjects born in
January and December, and round the total number of subjects to 104,000, to ensure integer number of subjects in each subgroup.}. 
The latter authors note that the cumulative distribution functions 
of December- and January-born children in these data cross each other, suggesting a violation of monotonicity.
Figure \ref{fig:IV} indicates this clearly. Most children started school at an older age if they were born in January rather than December. Indeed, this first-stage relationship is incredibly strong, with an average effect of $-0.667$ years and an $F$ statistic of 106,256. This would conventionally be considered persuasive evidence of a valid instrument. Note, however, that at the tails of the distribution the relationship between month of birth and entry age reverses: January birthdays predominate among the youngest starters, and December does so among the oldest. 

\begin{table}[htbp]
	\centering
	\caption{Randomization tests for the bounded null $H_{\vecle \bm{0}}$, which is equivalent to the monotonicity assumption. 
		The intervals are $90\%$ confidence intervals for the maximum individual effect, constructed by inverting the corresponding randomization tests.
	Columns 2--4 show results from randomization inference using different test statistics. The last column shows results from the classical Student's $t$-test.  
}\label{tab:IV}
	\begin{tabular}{ccccc}
		\toprule
		& Difference-in-means &  Wilcoxon &  Stephenson &  Student's $t$-test\\
		\midrule
		$p$-value & 1 & 1 & 0 & 1
		\\
		$90\%$ CI & $[-0.669, \infty)$ & $[-0.916, \infty)$ & $[0.084, \infty)$ & $[-0.670, \infty)$
		\\
		\bottomrule
	\end{tabular}%
\end{table}

If we designate December birth as the assigned-to-treatment condition and January birth as control, then the monotonicity assumption is equivalent to the null hypothesis $H_{\vecle \bm{0}}$: 
being born in December did not cause any child to go to school at an older age than they would have if born in January. 
We assume that birth month is as-if randomly assigned, and conduct randomization inference for the effect of December birth on the school-entry age. 
We test the bounded null $H_{\vecle \bm{0}}$ using the randomization $p$-value $p_{\bm{Z}, \bm{0}}$ 
with various %
test statistics satisfying the conditions in Theorem
\ref{thm:null_broader_monotone_control}(a), including the difference-in-means, Wilcoxon rank sum and Stephenson rank sum with $s=10$. 
Table \ref{tab:IV} lists the results from randomization tests using these three test statistics, supplemented by the classical Student's $t$-test. 
We emphasize that, although
both the randomization $p$-values and the corresponding intervals are numerically the same as that under the usual constant treatment effect assumption, 
they are also valid $p$-values for the bounded null $H_{\vecle \bm{0}}$ and valid confidence intervals for the maximum individual effect $\tau_{\max}$, as demonstrated in Section \ref{sec:broader}. 
Because the difference-in-means estimator for the average effect is negative, 
from the simulation results in Section \ref{sec:simu_max}, 
it is not surprising that neither difference-in-means or Wilcoxon rank sum give significant $p$-values. 
However, the Stephenson rank sum gives an almost zero $p$-value, 
strong evidence of
the existence of units violating the monotonicity assumption. 
Intuitively, the significant $p$-value is driven by the Stephenson rank placing more weight on larger outcomes coupled with $15\%$ of December-born children having a school-entry age of $7.75$, which is larger than the maximum school-entry age $7.67$ for the January-born children.  This means the treatment group has a large share of the most extreme observations.
The corresponding $90\%$ lower confidence limit is $0.084$ year, or equivalently about 1 month, 
suggesting some children would first enter the school one-month older if born in December than in January. 
We can therefore confidently conclude that being born in 
December increased
school-entry age for at least some students, i.e., the IV monotonicity assumption is violated in this application.

{\rev 
We now apply Theorem \ref{thm:effect_range} to study the effect range. 
Using the Stephenson rank sum statistic with $s=10$, 
the $95\%$ upper confidence limit for the minimum effect of December birth is $-0.917$ years ($-11$ months) while the $95\%$ lower confidence limit for the maximum effect is $0.084$ years ($1$ month). We are therefore $90\%$ confident that the range of the effect of birth month on school-entry age is at least 1 year, indicating significant individual effect heterogeneity. This of course is hardly surprising given Figure \ref{fig:IV}, which shows that despite the negative average effect of December birth, the children with the oldest school-entry age were born in this month.
}

\subsection{Evaluating the effects of six-month nutrition therapy}\label{sec:homefood}

The homefood study \citep{Blondal2021} is a recent randomized controlled clinical trial that tries to evaluate the effects of home delivered food and nutrition therapy for discharged geriatric hospital patients.\footnote{The details of the study can be found at \url{https://clinicaltrials.gov/ct2/show/NCT03995303}, and the data are publicly available at Harvard Dataverse with link \url{https://doi.org/10.7910/DVN/38X3LX}.}
Participants were randomized into two groups, and those in the treated group will be given free food for 24 weeks to fulfill protein and energy needs, which are based on individualized nutrition care plans designed by the dietitians. 
Here we focus on the effect of the treatment on the increase of lean body mass (kg) measured before and after the trial. 
We exclude two units with missing outcomes, resulting in 52 treated units and 52 control units. Figure \ref{fig:conf_homefood}(a) shows the histograms of the lean body mass changes in treated and control groups, respectively. 

We then infer the lower confidence limits for all quantiles of individual treatment effects. 
Figures \ref{fig:conf_homefood}(b) and (c) show the $90\%$ lower confidence limits for all quantiles of individual effects using Stephenson rank statistics with $s=6$ and $s=2$ (i.e., Wilcoxon rank), respectively. 
Obviously, the Stephenson rank with $s=6$ gives more informative results, based on which we are $90\%$ confident that at least $19.2\%$ units would have higher lean body mass if receiving treatment instead of control.  

\begin{figure}[htbp]
	\centering
	\begin{subfigure}{.33\textwidth}
		\centering
		\includegraphics[width=1\linewidth]{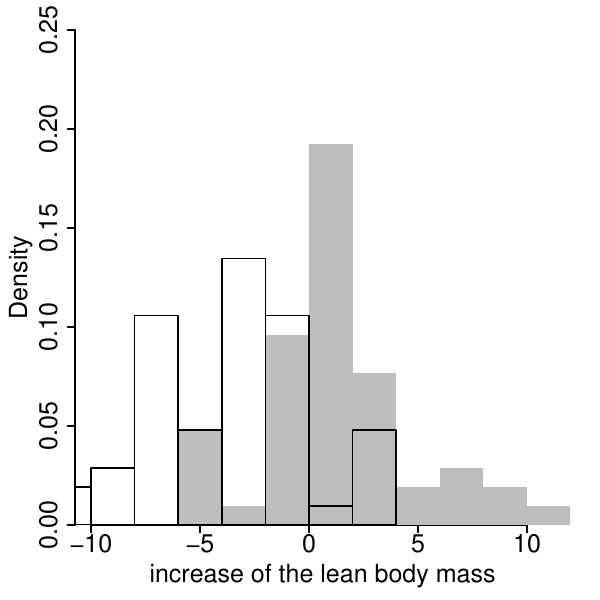}
		\caption{}
	\end{subfigure}%
	\begin{subfigure}{.33\textwidth}
		\centering
		\includegraphics[width=1\linewidth]{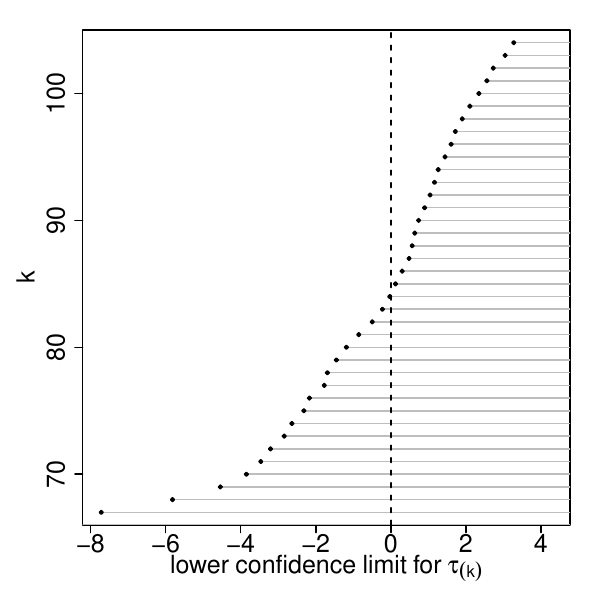}
		\caption{$s=6$}
	\end{subfigure}%
	\begin{subfigure}{.33\textwidth}
		\centering
		\includegraphics[width=1\linewidth]{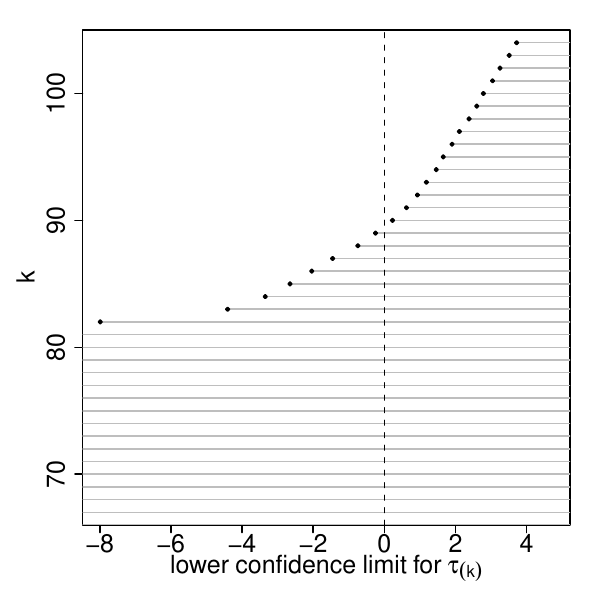}
		\caption{$s=2$}
	\end{subfigure}
	\caption{
	Histograms of observed lean body mass increase and 
	$90\%$ confidence intervals for quantiles of individual effects. 
    (a) shows the histograms of observed lean body mass increase in treatment (grey) and control (white) groups, respectively. 
    (b) and (c) shows the $90\%$ lower confidence limits for all quantiles of individual effects using the Stephenson rank statistics with $s$ equals $6$ and $2$, respectively. 
    The uninformative confidence intervals of $(-\infty, \infty)$ for quantiles of lower ranks are omitted from (b) and (c). 
	}\label{fig:conf_homefood}
\end{figure}

\end{document}